\documentclass[12pt]{article} 
\pdfoutput=1 
\usepackage[sectionbib]{natbib}
\usepackage{array,epsfig,fancyhdr,rotating}
\usepackage[]{hyperref}  
\usepackage{sectsty, secdot}
\sectionfont{\fontsize{12}{14pt plus.8pt minus .6pt}\selectfont}
\renewcommand{\theequation}{\thesection\arabic{equation}}
\subsectionfont{\fontsize{12}{14pt plus.8pt minus .6pt}\selectfont}

\usepackage{amsmath}
\usepackage{amssymb}
\usepackage{amsfonts}
\usepackage{multirow}
\usepackage{amsthm}
\usepackage{algorithm}
\usepackage{algpseudocode}

\newtheorem{theorem}{Theorem}

\newtheorem{proposition}{Proposition}
\theoremstyle{definition}

\usepackage{calrsfs, wasysym, verbatim, bbm, color, graphics, graphicx, geometry, amsbsy, amsthm, amsfonts, latexsym, graphicx, xifthen, manfnt, ifthen, threeparttable, booktabs, multirow, mdsymbol, placeins, caption, subfig, amsfonts}

\def\bV{{\mathbf V}}

\def\bZ{{\mathbf{Z}}}

\def\bbeta{{\boldsymbol{\beta}}}

\def\boldeta{{\boldsymbol{\eta}}}

\def\boldeta{{\boldsymbol{\eta}}}

\def\blambda{{\boldsymbol{\lambda}}}

\def\bzero{\boldsymbol{0}}

\newtheorem{condition}{Condition}

\begin{document}
	
	
	\renewcommand{\baselinestretch}{2}
	
	\markright{ \hbox{\footnotesize\rm Statistica Sinica
		}\hfill\\[-13pt]
		\hbox{\footnotesize\rm
		}\hfill }
	
	\markboth{\hfill{\footnotesize\rm FIRSTNAME1 LASTNAME1 AND FIRSTNAME2 LASTNAME2} \hfill}
	{\hfill {\footnotesize\rm Feature selection in high dimension with interval-censored data} \hfill}
	
	\renewcommand{\thefootnote}{}
	$\ $\par
	
	
	\fontsize{12}{14pt plus.8pt minus .6pt}\selectfont \vspace{0.8pc}
	\centerline{\large\bf  Variable Selection in Ultra-high Dimensional Feature}
	\vspace{2pt}
	\centerline{\large\bf Space for the Cox Model with Interval-Censored Data}
    \vspace{2pt}
	\centerline{\large\bf}
	\vspace{.4cm}
	\centerline{Daewoo Pak$^{1\#}$\footnote{$^{\#}$Contributed equally to this work}, Jianrui Zhang$^{2\#}$, Di Wu$^2$, Haolei Weng$^2$ and Chenxi Li$^{2*}\footnote{$^{*}$Corresponding author}$}
	\vspace{.4cm}
	\centerline{\it $^1$Yonsei University and $^2$Michigan State University}
	\vspace{.55cm} \fontsize{9}{11.5pt plus.8pt minus.6pt}\selectfont
	
	
	\begin{quotation}
		\noindent {\it Abstract:} We develop a set of variable selection methods for the Cox model under interval censoring, in the ultra-high dimensional setting where the dimensionality can grow exponentially with the sample size. The methods
 select covariates via a penalized nonparametric maximum likelihood estimation with some popular penalty functions, including lasso, adaptive lasso, SCAD, and MCP. We prove that our penalized variable selection methods with folded concave penalties or adaptive lasso penalty enjoy the oracle property. Extensive numerical experiments show that the proposed methods have satisfactory empirical performance under various scenarios.  The utility of the methods is illustrated through an application to a genome-wide association study of age to early childhood caries.

		\vspace{9pt}
		\noindent {\it Key words and phrases:}
		Cox model, interval censoring, oracle property, ultra-high dimension, variable selection.
		\par
	\end{quotation}\par

	\def\thefigure{\arabic{figure}}
	\def\thetable{\arabic{table}}
	
	\renewcommand{\theequation}{\thesection.\arabic{equation}}

	\fontsize{12}{14pt plus.8pt minus .6pt}\selectfont
	
\section{Introduction}
In the field of biomedical research, a large number of genetic, environmental, and clinical factors are often gathered to investigate their links to a disease. For example, genome-wide association studies (GWAS) genotype hundreds of thousands of genetic variants across the genomes of individuals to discover genetic factors that are associated with a disease phenotype. Such high-dimensional data usually have a much larger number of variables ($p$) than the sample size ($n$). Regression analysis with numerous predictors may lead to statistical problems such as over-fitting and non-uniqueness of parameter estimates, and thus variable selection has become an important topic in high-dimensional regression.

Since the seminal paper of lasso \citet{tibshirani1996regression}, penalized variable selection has been studied extensively, because it can select relevant predictors and estimate their effects on the response simultaneously. Most of the penalized variable selection methods were proposed for linear or generalized linear models. In addition to the lasso, other popular methods include the adaptive lasso \citep{Zou2006}, the smoothly clipped absolute deviation (SCAD)  \citep{fanli2001}, and MC+, which employs a minimax concave penalty (MCP)  \citep{zhang2010}, among others. While the lasso is widely used, it leads to biased parameter estimators \citep{fanli2001}. In contrast, the other three methods enjoy the so-called oracle property in both low-dimensional ($p<n$) and high-dimensional ($p>n$) settings \citep{fanli2001,Zou2006,huang2008adaptive,fan2011nonconcave}, which means that these methods can identify the true model with probability tending to one and their estimators of nonzero coefficients have the same asymptotic distribution as the most efficient estimator one could obtain if the set of relevant predictors were known. There have been also many penalized variable selection methods for time-to-event outcomes subject to right censoring, e.g.,  \citet{tibshirani1997lasso, fan2002variable, zhang2007adaptive, bradic2011regularization} for the Cox model, \citet{huang2006regularized, huang2010variable, hu2013adjusted} for the accelerated failure time model, \citet{ma2007additive, LinLv2013} for the additive hazards model, and \citet{LiuZeng2013} for semiparametric transformation models.


Recently, there has been a surge in statistical research in variable selection with interval-censored time-to-event outcomes, which are frequently encountered in longitudinal studies of chronic conditions like diabetes, dental caries and Alzheimer's disease. In the framework of parametric survival regression,
 \citet{wu2015penalized} developed algorithms for lasso-, adaptive lasso-, and SCAD-based penalized likelihood estimation of a proportional hazards model with a piecewise constant baseline hazard function, and \citet{scolas2016variable} proposed a penalized likelihood estimation with double adaptive lasso penalties for interval-censored data with a cure fraction based on a class of parametric accelerated failure time mixture cure models. In the framework of semiparametric survival regression,  \citet{zhao2019simultaneous} proposed a broken adaptive ridge regression (BAR) based on the Cox model, \citet{li2020adaptive} developed an adaptive lasso procedure for the Cox model, which can also deal with left truncation, \citet{li2020penalized} generalized BAR and adaptive lasso to semiparametric transformation models and considered other penalties like SCAD and MCP as well, and \citet{wu2020variable} developed a penalized sieve maximum likelihood estimation for the partially linear Cox model with various penalties such as lasso, SCAD, MCP and BAR.  We refer interested readers to \citet{du2022variable} for a more detailed review of variable selection methods with interval-censored data. Despite so many methods, only the methods of \citet{wu2020variable} and \citet{TianSun} were devised for high-dimensional feature space, but they do not have any theoretical guarantee. To the best of our knowledge, the oracle property of penalized variable selection with interval censored data has only been established in fixed-dimensional settings \citep[e.g.,][]{li2020adaptive} and the settings where $p$ diverges at a slower rate than $n$ \citep{zhao2019simultaneous,WuZhSun}.

In this paper, we develop the first set of high-dimensional variable selection methods for interval-censored data that have the oracle property.
The methods are devised for the Cox proportional hazards model and perform variable selection via a penalized nonparametric maximum likelihood estimation with popular penalties like SCAD, MCP and adaptive lasso that were proved to possess the oracle property in (generalized) linear models. We also develop a lasso-based method, although not investigating its theoretical properties. There are three major differences between our methods and the only two existing high-dimensional variable selection methods for interval-censored data \citep{wu2020variable,TianSun}. First, our methods with adaptive lasso or a folded concave penalty \citet{fan2011nonconcave} (e.g., SCAD and MCP) are proved to have the oracle property in ultra-high dimensional settings, meaning $\log p=O(n^{\delta})$ for some $\delta\in(0,1)$. Second, our methods select the tuning parameter based on the generalized information criterion \citep{FaTa2013}, which was proved to lead to consistent model selection in penalized generalized linear regression under ultra-high dimensional scenarios \citep{FaTa2013}, whereas \citet{wu2020variable} and \citet{TianSun} employ $k$-fold cross-validation and the extended Bayesian information criterion \citep{chen2008extended} respectively, both of which have not been shown to consistently select covariates in any ultra-high dimensional regression. Third,
our methods model the baseline hazard function of the Cox model nonparametrically, while \citet{wu2020variable} and \citet{TianSun} use a sieve approach to approximate that function. The nonparametric modeling of the baseline hazard enables us to extend the EM algorithm in \citet{wang2016flexible} to a penalized EM algorithm for implementing the proposed penalized maximum likelihood estimation. Compared to the sieve-based algorithms in \citet{wu2020variable} and \citet{TianSun}, the penalized EM algorithm has an explicit updating formula for the baseline hazard estimator, thereby accelerating the variable selection procedure.

We would like to also point out that our proofs of the oracle property are not an extension of the proofs in \citet{zhao2019simultaneous} and \citet{WuZhSun} from the setting of $p$ diverging with $n$ to that of $p>>n$. Their proofs require
a good initial estimator and the consistency of the negative hessian of the log likelihood as an estimator of the efficient information matrix; see Conditions C4 and C8 in \citet{zhao2019simultaneous} and Conditions 4 and 8 in \citet{WuZhSun}. These are still open problems in the existing literature. In contrast, our proofs for the folded concave penalties do not need those two conditions, and our proof for adaptive lasso does not require the negative hessian of the log likelihood to be consistent for the efficient information matrix.


The remainder of the paper is organized as follows. In Section \ref{sec:method}, we describe the algorithms of the proposed penalized variable selection methods with the lasso, adaptive lasso, SCAD, and MCP penalties. In Section \ref{sec:asymp}, we provide the theoretical properties, including the oracle property, of the methods with folded concave penalties and the adaptive lasso. The proofs of these properties are relegated to the supplementary material. Extensive numerical experiments are conducted to evaluate the finite-sample performance of the proposed methods in Section \ref{sec:simul}. In Section \ref{sec:appl}, we illustrate the utility of the proposed methods with an application to a genome‐wide association study of age to early childhood caries. We give some concluding remarks in Section \ref{sec:disc}.

For convenience, we define some generic notations used throughout the paper here. We let $I(\cdot)$ denote an indicator function. For a number $a$, we define $(a)_+=\max(a,0)$. For two numbers $a$ and $b$, the notation $a\vee b$ means $\max(a,b)$.
We use $O_p(\cdot)$ and $o_p(\cdot)$ to denote probabilistic order relations, and $O(\cdot)$ and $o(\cdot)$ to denote deterministic order relations. For two sequences $\{a_n\}$ and $\{b_n\}$, we use $a_n\gg b_n$ to denote $b_n=o(a_n)$. For a symmetric matrix $H$, we use $\lambda_{\min}(H)$ to denote its minimum eigenvalue. For a set $A$, we use $|A|$ to denote the cardinality of $A$. Let $\|\cdot\|_q$ be the $L_q$ norm of a vector.


\section{Methodology}\label{sec:method}
\subsection{Set-up}
We adopt the commonly-used Cox proportional hazards model to investigate the association between the time to a failure event, denoted by $T$, and a $p$-dimensional vector of covariates, denoted by $\bZ$. Specifically, the hazard function of $T$ given $\bZ$ is assumed to be
\begin{align} \label{coxmodel}
\lambda(t|\bZ) = \lambda(t) \exp(\bbeta^\top \bZ),
\end{align}
where $\lambda(t)$ is an unspecified baseline hazard function and $\bbeta$ is a vector of regression parameters. In health science research, $\bZ$ often represents high-dimensional covariates, including, e.g., Single Nucleotide Polymorphisms (SNPs) and gene expressions. The number of these covariates typically surpass the number of subjects $n$ in a random sample selected from the target population to fit the Cox model, i.e., $p > n$. The objective of this paper is to develop statistical methods to identify the covariates truly associated with $T$ among $\bZ$ when $T$ is subject to interval censoring and $p$ can be of exponential order of $n$.


Denote the sequence of inspection times for subject $i$'s failure event ($i = 1, \ldots, n$) by $\bV_i = (V_{i1}, \ldots, V_{iK_i})^\top$. We assume that the number of inspection times, $K_i$, and the times themselves, $\bV_i$, are random variables that are independent of $T_i$ given $\bZ_i$, which is known as mixed case interval censoring \citep{schick2000consistency}. The censoring interval for $T_i$, denoted by $(L_i, R_i]$, can be formed by selecting two inspection times bracketing $T_i$ with no other inspection times falling between them. We set $L_i=0$ when $T_i$ is smaller than $V_{i1}$ and $R_i=\infty$ when $T_i$ is larger than $V_{iK_i}$, and thus $0 \le L_i < T_i \le R_i \le \infty$. Then the observed data can be represented by $\mathcal{O} \equiv \lbrace \mathcal{O}_i: i = 1, \ldots, n\rbrace$, where $\mathcal{O}_i = (L_i, R_i, \bZ_i^\top)^\top$.


\subsection{Likelihood}\label{subsec1}
Under the Cox model (\ref{coxmodel}), the logarithm of the likelihood function for the observed interval-censored data is
\begin{align}\label{loglik0}
l_{n}(\bbeta,\Lambda)=\sum_{i=1}^{n}\log\left[\exp\{-\Lambda(L_i|\bZ_i; \bbeta)\}-\exp\{-\Lambda(R_i|\bZ_i; \bbeta)\}\right],
\end{align}
where $\Lambda(t|\bZ_i; \bbeta) = \Lambda(t)\exp(\bbeta^\top\bZ_i)$, $\Lambda(t)$ is the cumulative baseline hazard function, i.e., $\Lambda(t)=\int_0^{t}\lambda(s)ds$, and we set a convention that $\exp\{-\Lambda(\infty|\bZ; \bbeta)\} = 0$. The nonparametric maximum likelihood estimator (NPMLE) for $\Lambda(t)$, denoted by $\tilde{\Lambda}(t)$, increase only within a particular support set, which is referred to as the maximal intersections \citep{alioum1996proportional}. Specifically, the maximal intersections are a set of disjoint intervals of the form $(l,u]$, where $l \in \lbrace L_i: i = 1, \ldots, n\rbrace$, $u \in \lbrace R_i: i = 1, \ldots, n\rbrace$ and there is no $L_i$ or $R_i$ in $(l,u)$ for any $i$. In addition, only the jump sizes of $\Lambda(t)$ over $(l,u]$'s $(u<\infty)$ affect the log likelihood \eqref{loglik0}, and given those jump sizes, the log likelihood is indifferent to how $\Lambda(t)$ increases within a maximal intersection. Denote the maximal intersections with a finite right endpoint by $ \cup_{k = 1}^m (l_k, u_k]$. Then the NPMLE $\tilde{\Lambda}(t)$ can be represented solely by its jump sizes $\lambda_k = \Lambda(u_k) - \Lambda(l_k)$ over $(l_k, u_k]$ $(k=1,\ldots,m)$. Hence, just for the ease of computing $\tilde{\Lambda}(t)$, we can assume that the mass of $\Lambda(t)$ on  $(l_k, u_k]$ concentrates at $u_k$ with a magnitude of $\lambda_k$ $(k=1,\ldots,m)$ and $\Lambda(t)$ remains flat on $\overline{\cup_{j = 1}^m (l_k, u_k]}$, the complement of $\cup_{j = 1}^m (l_k, u_k]$ on $[0,\infty)$. Let $\blambda=(\lambda_1,\ldots,\lambda_m)^\top$. Then we can easily show that $l_{n}(\bbeta,\Lambda)$ can be expressed as $l_n(\bbeta,\blambda)$, where \begin{align}\label{loglik}
l_n(\bbeta,\blambda)=\sum_{i=1}^{n}\log\left(\exp\{-\sum_{u_k\le L_i}\lambda_k\exp(\bbeta^\top\bZ_i)\}\left[1-\exp\{-\sum_{L_i<u_k\le R_i}\lambda_k\exp(\bbeta^\top\bZ_i)\}\right]^{I(R_i<\infty)}\right),
\end{align}
by replacing $\Lambda(t)$ with $\sum_{u_k\le t}\lambda_k$ in (\ref{loglik0}).



\subsection{Variable selection}
Without loss of generality, we assume that every covariate $Z_j$ has been standardized such that $\sum_{i = 1}^n Z_{ij} = 0$ and $\sum_{i = 1}^n Z^2_{ij}/n = 1$ $(j=1,\ldots,p)$. We perform variable selection by maximizing
\begin{align}\label{penlik}
\frac{1}{n}l_n(\bbeta,\Lambda)- \sum_{j = 1}^p p_{\theta, \alpha}(|\beta_j|),
\end{align}
with respect to $(\bbeta, \Lambda)$, where $p_{\theta, \alpha}(\cdot)$ is a penalty function penalizing the magnitude of $\beta_j$, $\theta$ is a thresholding parameter that controls the trade-off between the likelihood and $|\beta_j|$, and $\alpha$ is a tuning parameter that shapes the penalty. In this article, we consider four penalty functions, lasso, adaptive lasso, SCAD and MCP. Since the penalty term in equation (\ref{penlik}) does not depend on $\Lambda$, the penalized nonparametric maximum likelihood estimator (PNPMLE) $\hat{\Lambda}(t)$ has the same support set as the NPMLE $\tilde{\Lambda}(t)$. Therefore, we can write \eqref{penlik} as
\begin{align}\label{penlik_with_blambda}
\frac{1}{n}l_n(\bbeta,\blambda)- \sum_{j = 1}^p p_{\theta, \alpha}(|\beta_j|).
\end{align}
The corresponding penalized maximum likelihood estimators for $\bbeta$ and $\blambda$ are denoted by $\hat{\bbeta}$ and $\hat{\blambda}$.

 Directly maximizing (\ref{penlik_with_blambda}) poses challenges because the estimator of $\blambda$ given $\bbeta$ does not have an analytic form and thus an explicit profile likelihood in terms of $\bbeta$ is not available. Following a computation technique in \cite{wang2016flexible}, we consider a set of independent latent variables, $W_{ik}$ $(i = 1, \ldots, n; k = 1, \ldots, m)$, each of which follows a Poisson distribution with a mean of $\lambda_k \exp(\bbeta^\top\bZ_i)$, and define $A_i=\sum_{u_k\le L_i}W_{ik}$, $B_i=I(R_i<\infty)\sum_{L_i<u_k\le R_i}W_{ik}$
  and
 $\tilde{\mathcal{O}}\equiv\{(\mathcal{O}_i,A_i=0); 1\le i\le n, R_i=\infty\}\cup\{(\mathcal{O}_i,A_i=0,B_i>0); 1\le i\le n, R_i<\infty\}$, where $A_i=0$ means that $A_i$ is observed to be zero and $B_i>0$ means that $B_i$ is observed to be positive. The log likelihood of $\tilde{\mathcal{O}}$ takes the form
\begin{align*}
&\sum_{i=1}^{n}\log\left[\left\{\prod_{u_k\le L_i}\mbox{pr}(W_{ik}=0)\right\}\left\{1-\mbox{pr}\left(\sum_{L_i<u_k\le R_i}W_{ik}=0\right)\right\}^{I(R_i<\infty)}\right]&\\
&=\sum_{i=1}^{n}\log\left(\exp\{-\sum_{u_k\le L_i}\lambda_k\exp(\bbeta^\top\bZ_i)\}\left[1-\exp\{-\sum_{L_i<u_k\le R_i}\lambda_k\exp(\bbeta^\top\bZ_i)\}\right]^{I(R_i<\infty)}\right),
\end{align*}
which is the same as $l_n(\bbeta,\blambda)$.
Thus, the maximization of (\ref{penlik_with_blambda}) can be achieved using a penalized EM algorithm \citep{Green1990}, wherein we treat $\tilde{\mathcal{O}}$ as the observed data and  $\lbrace (\mathcal{O}_i,W_{ik}); i=1,\ldots, n, u_k\le R_i^*\rbrace$ with $R_i^*=L_iI(R_i=\infty)+R_iI(R_i<\infty)$ as the complete data. The log likelihood of the complete data is
\begin{align}\label{loglik_comp}
\ell^c(\bbeta,\blambda)=\sum_{i=1}^{n}\sum_{k=1}^{m}I(u_k\le R_i^*)\left[W_{ik}\log\{\lambda_k\exp(\bbeta^\top\bZ_i)\}-\lambda_k\exp(\bbeta^\top\bZ_i)-\log W_{ik}!\right].
\end{align}

In the E-step of the $(s+1)$-th EM iteration $(s\ge 0)$, we compute the expectation $\hat{E}(W_{ik})$ of $W_{ik}$ given the observed data and the $s$-th updates of $\hat{\bbeta}$ and $\hat{\blambda}$, $({\hat{\bbeta}}^{(s)},{\hat{\blambda}}^{(s)})$, as follows: for $u_k\le L_i$, $\hat{E}(W_{ik})=0$ due to $A_i = 0$, and for $L_i<u_k\le R_i$ with $R_i<\infty$,
\begin{align*}
\hat{E}(W_{ik}) &=E(W_{ik}|\sum_{L_i<u_r\le R_i}W_{ir}>0)\\
&=\frac{\hat{\lambda}_k^{(s)}\exp(\hat{\bbeta}^{(s)\top}\bZ_i)}{1-\exp\{-\sum_{L_i<u_r\le R_i}\hat{\lambda}_r^{(s)}\exp(\hat{\bbeta}^{(s)\top}\bZ_i)\}}.
\end{align*}
In the M-step, we first get the updates for $\lambda_k$'s by maximizing the expected complete-data log likelihood with respect to $\blambda$ given $\bbeta$. It is easy to show that such updates are
\begin{align}\label{mlambda}
\hat{\lambda}_k^{(s+1)}(\bbeta) = \frac{\sum_{i=1}^{n}I(u_k\le R_i^*)\hat{E}(W_{ik})}{\sum_{i=1}^{n}I(u_k\le R_i^*)\exp(\bbeta^\top\bZ_i)}\quad (k=1,\ldots,m).
\end{align}
Plugging in \eqref{mlambda} into the conditional expectation of (\ref{loglik_comp}), we get a profile expected complete-data log likelihood,
\begin{align*}
\hat{E}\{\ell^c(\bbeta,\hat{\blambda}^{(s+1)}(\bbeta))\}=\sum_{i=1}^{n}\sum_{k=1}^{m}I(u_k\le R_i^*)\hat{E}(W_{ik})\left[-\log\left\{\sum_{l=1}^{n}I(u_k\le R_l^*)\exp(\bbeta^\top\bZ_l)\right\}+\bbeta^\top\bZ_i\right].
\end{align*}
The update of $\hat{\bbeta}$ can be obtained by maximizing
\begin{align}\label{Q}
\frac{1}{n}\hat{E}\{\ell^c(\bbeta,\hat{\blambda}^{(s+1)}(\bbeta))\}-\sum_{j=1}^p p_{\theta, \alpha}(|\beta_j|)
\end{align}
with respect to $\bbeta$. However, maximizing \eqref{Q} is not an easy task due to the non-differentiability of the penalty at zero, the high dimensionality of $\bbeta$, and the nonconcavity of the whole objective function if a nonconvex penalty like SCAD or MCP is used. The remaining of this subsection is devoted to our proposed approaches to maximizing $\eqref{Q}$ with the lasso, adaptive lasso, SCAD and MCP penalties respectively.

With a little bit of notation abuse, let $\bZ$ denote the $n \times p$ matrix of $(\bZ_1, \ldots, \bZ_n)^\top$ and $\boldeta = \bZ\bbeta$. We can write $\hat{E}\{\ell^c(\bbeta,\hat{\blambda}^{(s+1)}(\bbeta))\}$ as $Q(\boldeta)$ since the former depends on $\bbeta$ through $\boldeta$. In the spirit of \cite{simon2011regularization}, we approximate $-Q(\boldeta)$ with a second-order Taylor expansion around $\hat{\boldeta}^{(s)}=\bZ\hat{\bbeta}^{(s)}$ and reformulate the problem of maximizing \eqref{Q} as a penalized weighted least squares problem. Let $Q'(\boldeta)$ and $Q''(\boldeta)$ be the first and second derivatives of $Q(\boldeta)$ with respect to $\boldeta$, respectively. Then,
\begin{align*}
-Q(\boldeta) \approx (e(\hat{\boldeta}^{(s)}) - \bZ\bbeta)^\top W(\hat{\boldeta}^{(s)}) (e(\hat{\boldeta}^{(s)}) - \bZ\bbeta)/2+ C(\hat{\boldeta}^{(s)}),
\end{align*}
where $e(\hat{\boldeta}^{(s)}) = \hat{\boldeta}^{(s)} - \left[Q''(\hat{\boldeta}^{(s)})\right]^{-}Q'(\hat{\boldeta}^{(s)})$, $\left[Q''(\hat{\boldeta}^{(s)})\right]^{-}$ is a generalized inverse of $\left[Q''(\hat{\boldeta}^{(s)})\right]$, $W(\hat{\boldeta}^{(s)}) = -Q''(\hat{\boldeta}^{(s)})$, and $C(\hat{\boldeta}^{(s)})$ is the term that does not depend on $\bbeta$. Thus, the maximization of (\ref{Q}) with respect to $\bbeta$ becomes a penalized weighted least squares problem, which is to minimize
\begin{align*}
M(\bbeta) =\frac{1}{2n}(e(\hat{\boldeta}^{(s)}) - \bZ\bbeta)^\top W(\hat{\boldeta}^{(s)}) (e(\hat{\boldeta}^{(s)}) - \bZ\bbeta) +  \sum_{j=1}^p p_{\theta, \alpha}(|\beta_j|).
\end{align*}
We adopt a coordinate descent algorithm similar to that in \citet{breheny2011coordinate} to minimize $M(\bbeta)$. The algorithm minimizes $M(\bbeta)$ iteratively
with respect to each component of $\bbeta$ while fixing the other components at the current updates, and it cycles through $\bbeta$ for only one time. Therefore, the minimization boils down to a series of univariate penalized least squares problems, each of which has analytical solutions for the lasso, adaptive lasso, SCAD and MCP penalties, as described below.

In each univariate penalized least squares problem, we write the objective function as $M(\beta_j;\hat{\bbeta}^{(s+1,s)}_{-j})$, where $ \hat{\bbeta}^{(s+1,s)}_{-j}$ denotes the subvector of the current update of $\hat{\bbeta}$ excluding the $j$-th component $(j=1,\ldots,p)$. By simple algebra, we have
\begin{equation}\label{Mtitle}
    M(\beta_j;\hat{\bbeta}^{(s+1,s)}_{-j}) = \frac{1}{2}v_j(\frac{y_j}{v_j}-\beta_j)^2+p_{\theta, \alpha}(|\beta_j|)+D,
\end{equation}
where $y_j = \frac{1}{n} \tilde{\bZ}^\top_{j}W(\hat{\boldeta}^{(s)}) (e(\hat{\boldeta}^{(s)}) - \bZ_{-j}\hat{\bbeta}^{(s+1,s)}_{-j})$, $\tilde{\bZ}_{j}$ is the $j$-th column of $\bZ$, $\bZ_{-j}$ is the submatrix of $\bZ$ excluding $\tilde{\bZ}_{j}$, $v_j = \frac{1}{n}\tilde{\bZ}^\top_{j}W(\hat{\boldeta}^{(s)}) \tilde{\bZ}_j$, and $D$ is a constant that does not depend on $\beta_j$.
Define a soft-thresholding operator $S(x, \theta)$ as in \citep{donoho1994ideal}, that is,
\begin{align*}
S(x, \theta) =   \left\{ \begin{array}{ll}
    x-\theta,& \text{if } x > \theta,\\
    0,              & \text{if } |x| \le \theta,\\
    x+\theta ,& \text{if } x < -\theta.
\end{array}\right.
\end{align*}
Denote the solution to the univariate penalized least squares problem of minimizing \eqref{Mtitle} by $\hat{\beta}^{(s+1)}_j$, which is a function of $(y_j, v_j, \theta, \alpha)$, denoted by $f(y_j, v_j, \theta, \alpha)$. In the case of the lasso and adaptive lasso penalties, defined respectively as $p_{\theta, \alpha}(|\beta_j|) = \theta |\beta_j|$ and $p_{\theta, \alpha}(|\beta_j|) = \theta |\beta_j|/|\tilde{\beta}_j|$, where $\tilde{\beta}_j$  is an initial estimator of $\beta_j$ (see detailed discussions in Section \ref{sec:asymp}), the updates of $\hat{\beta}_j$ for the lasso and the adaptive lasso are respectively $\hat{\beta}^{(s+1)}_j =S\left(y_j, \theta\right)/v_j$ and $\hat{\beta}^{(s+1)}_j = S\left(y_j, \theta/|\tilde{\beta}_j|\right)/v_j$. In our implementation of adaptive lasso, we set $\tilde{\beta}_j$ to be the lasso estimator for $\beta_j$. The SCAD penalty and its first derivative are, respectively,
\begin{align*}
p_{\theta, \alpha}(|\beta_j|) =   \left\{ \begin{array}{ll}
    \theta |\beta_j|,& \text{if } |\beta_j| \le \theta,\\
    -\frac{|\beta_j|^2 - 2\alpha\theta |\beta_j| + \theta^2}{2(\alpha-1)},  & \text{if } \theta < |\beta_j| \le \alpha\theta,\\
    \frac{(\alpha+1)\theta^2}{2} ,& \text{if } |\beta_j| > \alpha\theta,
\end{array}\right.
\end{align*}
and
\begin{align*}
p^\prime_{\theta, \alpha}(|\beta_j|)  =\theta I(|\beta_j|\le \theta) + \frac{(\alpha\theta - |\beta_j|)_+}{\alpha-1} I(|\beta_j| >\theta),
\end{align*}
for some $\alpha > 2$. The solution to the univariate penalized least squares problem with the SCAD penalty is thoroughly elucidated in Appendix B of \citet{fan2011nonconcave}.
The MCP penalty function is
\begin{align*}
p_{\theta, \alpha}(|\beta_j|) =   \left\{ \begin{array}{ll}
    \theta |\beta_j| - \frac{\beta^2_j}{2\alpha},& \text{if } |\beta_j| \le \alpha\theta,\\
    \frac{1}{2}\alpha \theta^2,  & \text{if } |\beta_j| > \alpha\theta,
\end{array}\right.
\end{align*}
for some $\alpha > 1$, and its first derivative is $p^\prime_{\theta, \alpha}(|\beta_j|) = (\theta - |\beta_j|/\alpha)I(|\beta_j| \le \alpha\theta)$. Following similar arguments to Appendix B of \citet{fan2011nonconcave}, we derived the solution to the univariate penalized least squares problem with the MCP penalty as follows:
\begin{align*}
\hat{\beta}^{(s+1)}_j=   \left\{
\begin{array}{ll}
   \frac{S(y_j, \theta)}{v_j - \alpha^{-1}}, & \text{if } |y_j| \le v_j\alpha\theta,\\
    \frac{y_j}{v_j}I(|y_j| > \sqrt{v_j \alpha}\theta),  & \text{if } |y_j| > v_j\alpha\theta.
\end{array}\right.
\end{align*}
Because our numerical experiments and real data application all use minor allele counts of SNPs as covariates, we follow \citet{breheny2011coordinate} to choose $\alpha=2.5$ for SCAD and $\alpha=1.5$ for MCP.



\subsection{Thresholding parameter tuning}
We have described our penalized EM algorithm for maximizing (\ref{penlik}) given a fixed value of $\theta$. Similar to other variable selection methods using regularization, the choice of the thresholding parameter $\theta$ is crucial for the performance of our estimators for $(\bbeta, \Lambda)$. We seek the optimal $\theta$ along a path of its values as follows. We initiate the path by starting with a sufficiently large value of $\theta$ that results in $\hat{\beta}_j = 0$ $(j=1,\ldots,p)$. To do so, we first choose $\hat{\bbeta}^{(0)}$ to be zero and $\hat{\blambda}^{(0)}$ to be a vector of $1/n$'s and calculate the corresponding $y_j$ and $v_j$ $(j=1,\ldots,p)$. Then, our first value of $\theta$ can be set as $\theta_{\max} = \mbox{max}_j|y_j|$ for the lasso penalty, $\theta_{\max} = \mbox{max}_j\left\lbrace |y_j||\tilde{\beta}_j|\right\rbrace$ for the adaptive lasso penalty, $\theta_{\max} = \mbox{max}_j\left\lbrace|y_j|\vee |y_j|/v_j \right\rbrace$ for SCAD, and $\theta_{\max} = \mbox{max}_j\left\lbrace|y_j|\vee |y_j|/(v_j \alpha)\right\rbrace$ for MCP. We do not seek the optimal $\theta$ all the way to zero, because the parameter estimates behave poorly in terms of any model selection criterion when $\theta$ is near zero. Thus, following the approach used in \citet{simon2011regularization}, we set the last value of $\theta$ in the path to be $\theta_{\text{min}} = \epsilon\theta_{\text{max}}$, where $\epsilon$ is chosen to be 0.05 for the lasso, SCAD and MCP penalties and 0.0001 for the adaptive lasso based on our simulation experiences. Also following \citet{simon2011regularization}, we generate a grid of 101 values of $\theta$ over the interval $[\theta_{\text{min}},\theta_{\text{max}}]$ by setting $\theta_{r} = \theta_{\text{max}}(\theta_{\text{min}}/\theta_{\text{max}})^{(r-1)/100}$ for $r = 1, \ldots, 101$. We then compute the solution $(\hat{\bbeta}_{\theta_r},\hat{\blambda}_{\theta_r})$ that maximizes \eqref{penlik_with_blambda} with $\theta=\theta_r$ $(r=1,\ldots,101)$  using the above penalized EM algorithm. To speed up the convergence, we employ warm starts along the solution path, that is, set $(\hat{\bbeta}_{\theta_r},\hat{\blambda}_{\theta_r})$ as the initial value for computing $(\hat{\bbeta}_{\theta_{r+1}},\hat{\blambda}_{\theta_{r+1}})$ with the EM algorithm. Finally, we select the $\theta_r$ that minimizes the following generalized information criterion (GIC), defined as
\begin{align*}
\mbox{GIC}(\theta) = -2 l_n(\hat{\bbeta}_{\theta},\hat{\blambda}_{\theta}) + \log(\log(n)) \log(p) \|\hat{\bbeta}_{\theta}\|_0
\end{align*}
This GIC has been shown to lead to the true model with asymptotic probability one in the penalized generalized linear regression \citep{FaTa2013}. So we conjecture that it is adequate for consistently identifying the true Cox model with interval-censored data. Now we summarize the entire penalized variable selection algorithm as follows.

\begin{algorithm}[H]
        \caption*{\textsc{Penalized variable selection algorithm.}}
\begin{algorithmic}[1]
   \State Set $r = 1$, and let $\hat{\bbeta}_{\theta_0} = \mathbf{0}$ and  $\hat{\blambda}_{\theta_0}$ be a vector of $1/n$'s.
   \State Set $s = 0$,  $\hat{\bbeta}^{(s)}_{\theta_r} = \hat{\bbeta}_{\theta_{r-1}}$ and $\hat{\blambda}_{\theta_r}^{(s)} = \hat{\blambda}_{\theta_{r-1}}$
   \State Compute $\hat{E}(W_{ik})$ with $(\hat{\blambda}_{\theta_r}^{(s)}, \hat{\bbeta}^{(s)}_{\theta_r})$. \Comment{E-step}
    \State Maximize the expected complete-data log likelihood. \Comment{M-step}
    \begin{enumerate}
        \item[(a)] Calculate ${\hat{\boldeta}^{(s)}}= \bZ\hat{\bbeta}^{(s)}_{\theta_r}$
        \item[(b)] Successively for $j \in (1, \ldots, p)$, calculate the $j$-th component of $\hat{\bbeta}^{(s+1)}_{\theta_r}$ as $\hat{\beta}^{(s+1)}_j = f(y_j, v_j, \theta_{r}, \alpha)$, the solution to the univariate penalized least squares problem \eqref{Mtitle} with $\hat{\bbeta}_{-j}^{(s+1,s)} = (\hat{\beta}^{(s+1)}_{\theta_{r},1}, \ldots, \hat{\beta}^{(s+1)}_{\theta_{r},j-1}, \hat{\beta}^{(s)}_{\theta_{r},j+1}, \ldots, \hat{\beta}^{(s)}_{\theta_{r},p})^\top$.
        \item[(c)] Update the $k$-th component of $\hat{\blambda}_{\theta_r}^{(s+1)}$ as $\hat{\lambda}_k^{(s+1)}(\hat{\bbeta}_{\theta_r}^{(s+1)})$ in (\ref{mlambda}) $(k=1,\ldots,m)$. Set $s \leftarrow s + 1$.
    \end{enumerate}
    \State Repeat Steps 3-4 until convergence is attained, defined by the relative distance between two successive updates of $(\hat{\bbeta}_{\theta_r}, \hat{\blambda}_{\theta_r})$ being less than 0.01, or $s = 101$. Set $r \leftarrow r + 1.$
    \State Repeat Steps 2-5 until $r = 101$. Select the estimates from $(\hat{\blambda}_{\theta_1}, \hat{\bbeta}_{\theta_1}), \ldots, (\hat{\blambda}_{\theta_{101}}, \hat{\bbeta}_{\theta_{101}})$ that minimizes the GIC.
\end{algorithmic}
\end{algorithm}

\section{Theoretical properties} \label{sec:asymp}

Let $l_{n}(\bbeta,\Lambda)$ be the logarithm of the likelihood function for the observed data evaluated at $(\bbeta,\Lambda)$ and $l_{pn}(\bbeta)=\sup_{\Lambda} l_n(\bbeta, \Lambda)/n$ be the logarithm of the profile likelihood function evaluated at $\bbeta$. Define
\[
\dot{l}_{pn}(\mathbf{\boldsymbol{\beta}})=\frac{\partial l_{pn}(\mathbf{\boldsymbol{\beta}})}{\partial\mathbf{\boldsymbol{\beta}}}\quad\mbox{and}\quad\ddot{l}_{pn}(\mathbf{\boldsymbol{\beta}})=\frac{\partial^{2}l_{pn}(\mathbf{\boldsymbol{\beta}})}{\partial\mathbf{\boldsymbol{\beta}}\partial\mathbf{\boldsymbol{\beta}}^{\top}}.
\]
We aim to derive the theoretical properties of the variable selection method that works by maximizing
\begin{equation}\label{penloglik}
   \frac{1}{n}l_{n}(\bbeta,\Lambda)-\sum_{j=1}^{p}p_{\theta_{n},\alpha}(|\beta_{j}|)
\end{equation}
with respect to $(\beta,\Lambda)$, where the tuning parameter $\theta_n$  changes with $n$. Maximizing \eqref{penloglik} is equivalent to maximizing
\begin{equation}\label{penpl}
   \mathcal{C}(\bbeta)=l_{pn}(\bbeta)-\sum_{j=1}^{p}p_{\theta_{n},\alpha}(|\beta_{j}|).
\end{equation}
Compared to \eqref{penloglik}, the objective function \eqref{penpl} only involves $\bbeta$ and thus is easier to study for deriving the theoretical properties of the variable selection method. Therefore, we use \eqref{penpl} in place of \eqref{penloglik} in the sequel.

 Without loss of generality, we write the true regression coefficients  as $\bbeta_{0}=(\bbeta_{01}^{\top},0^{\top})^{\top}$ and let $s=\|\bbeta_{01}\|_{0}$. Denote the $j$th component of $\bbeta_{0}$ by $\beta_{0j}$ and let $\mathcal{M}_{\ast}=\left\{ j:\beta_{0j}\neq0\right\} =\left\{ 1,\dots,s\right\} $ be the index set of relevant features. For any $\boldsymbol{\beta}\in\mathbb{R}^{p}$,
it can be decomposed as $\boldsymbol{\beta}=(\mathbf{\boldsymbol{\beta}}_{1}^{\top},\bbeta_{2}^{\top})^{\top}$ accordingly. Let $\mathcal{I}_{11}$ denote the efficient information matrix for $\bbeta_1$ in the oracle model $l_n(\bbeta_1,\Lambda)$ evaluated at $(\mathbf{\boldsymbol{\beta}}_{01},\Lambda_{0})$, where $\Lambda_0$ is the true cumulative baseline hazard. Thus $\mathcal{I}_{11}$ is an $s\times s$ matrix. We also decompose $\bZ$ as $\bZ=(\bZ_1,\bZ_2)$ accordingly. Furthermore, we write \[
\dot{l}_{pn}(\mathbf{\boldsymbol{\beta}})=\begin{pmatrix}\dot{l}_{pn1}(\mathbf{\boldsymbol{\beta}})\\
\dot{l}_{pn2}(\mathbf{\boldsymbol{\beta}})
\end{pmatrix}\quad\text{and}\quad\ddot{l}_{pn}(\mathbf{\boldsymbol{\beta}})=\begin{pmatrix}\ddot{l}_{pn11}(\mathbf{\boldsymbol{\beta}}) & \ddot{l}_{pn12}(\mathbf{\boldsymbol{\beta}})\\
\ddot{l}_{pn21}(\mathbf{\boldsymbol{\beta}}) & \ddot{l}_{pn22}(\mathbf{\boldsymbol{\beta}})
\end{pmatrix},
\]
where $\dot{l}_{pn1}(\mathbf{\boldsymbol{\beta}})$ is a $s\times1$ vector and $\ddot{l}_{pn11}(\mathbf{\boldsymbol{\beta}})$ is a $s\times s$ matrix. Define a coordinate subspace of $\mathbb{R}^{p}$ to be a linear space spanned by a subset of the natural basis $\{e_1,\ldots,e_p\}$, where every $e_j$ is a $p$-dimensional vector with the $j$th component being 1 and the others 0. Let $\mathbb{S}_s$ be the union of all $s$-dimensional coordinate subspaces of $\mathbb{R}^p$.

We assume that $\alpha$ is fixed throughout the analysis. Let $\rho(t;\theta_{n})=\theta_{n}^{-1}p_{\theta_{n},\alpha}(t)$ and write $\rho(t;\theta_{n})$ as $\rho(t)$ for simplicity if no confusion will be caused. Following \cite{fan2011nonconcave}, we define the \textquotedblleft local concavity\textquotedblright{} of the penalty function $\rho(\cdot)$ at $\boldsymbol{v}=(v_{1},\dots,v_{q})^{\top}\in\mathbb{R}^{q}$ with $\|\boldsymbol{v}\|_{0}=q$ to be
\[
\kappa(\rho,\boldsymbol{v})=\lim_{\epsilon\to0+}\max_{1\le j\le q}\sup_{t_{1}<t_{2}\in(|v_{j}|-\epsilon,|v_{j}|+\epsilon)}-\frac{\rho^{\prime}(t_{2})-\rho^{\prime}(t_{1})}{t_{2}-t_{1}}.
\]
For such a $\boldsymbol{v}$, we also define $\rho^{\prime}(\boldsymbol{v})=\left(\rho^{\prime}(v_{1}),\dots,\rho^{\prime}(v_{q})\right)^{\top}$ and $\text{sgn}(\boldsymbol{v})=\left(\text{sgn}(v_{1}),\dots,\text{sgn}(v_{q})\right)^{\top}$.

Following \cite{huang2008adaptive}, we say that the initial estimator
$\tilde{\boldsymbol{\beta}}$ used in adaptive lasso is \textit{zero-consistent with rate $r_{n}$} if $r_{n}\max_{j\in\mathcal{M}_{\ast}^{c}}|\tilde{\beta}_{j}|=O_{P}(1)$, where $r_{n}\to\infty$, and there exists a positive constant $\xi_{b}>0$ such that for any $\epsilon>0$,
\[
\text{P}\left(\min_{j\in\mathcal{M}_{\ast}}|\tilde{\beta}_{j}|>\xi_{b}d_{n}^{\ast}\right)>1-\epsilon
\]
for sufficiently large $n$, where $d_{n}^{\ast}=0.5\min\left\{ |\beta_{0j}|,j\in\mathcal{M}_{\ast}\right\} $
is half of the minimum signal strength.

The following regularity conditions are assumed in deriving the theoretical properties of our variable selection methods. Conditions 1 and 6 are only for the nonconcave penalized likelihood methods.

\begin{condition}
\label{condition.Fan.1}
$\rho(t;\theta_{n})$ is increasing and concave in $t\in[0,\infty)$ and has a continuous derivative $\rho^{\prime}(t;\theta_{n})$ with $\rho^{\prime}(0+;\theta_{n})>0$. In addition, $\rho^{\prime}(t;\theta_{n})$ is increasing in $\theta_{n}\in(0,\infty)$ and $\rho^{\prime}(0+;\theta_{n})$ is independent of $\theta_{n}$.
\end{condition}

\begin{condition}
\label{condition.D1}
The components of $\mathbf{\boldsymbol{\beta}}_{01}$ are bounded by a known positive constant $C_{\beta}$. The union of the supports of inspection times $(V_{1},\dots,V_{K})$ is a finite interval $[\zeta,\tau]$, where $0<\zeta<\tau<\infty$. $\Lambda_{0}$ is strictly increasing and continuously differentiable on $[\zeta,\tau]$, and there exists a positive constant $C_{\Lambda}$ such that $C_{\Lambda}^{-1}<\Lambda_{0}(\zeta)<\Lambda_{0}(\tau)<C_{\Lambda}$.
\end{condition}

\begin{condition}
\label{condition.D23}
Each component of the covariate vector $\bZ$ is bounded almost surely by a positive constant $C_{\bZ}$ and the covariance matrix of $\bZ_1$ is positive definite.
\end{condition}

\begin{condition}
\label{condition.D4}
The number of inspection times, $K$, is positive almost surely, and $E(K)<\infty$. Additionally, $\text{P}(V_{j+1}-V_{j}\ge\eta|\bZ,K)=1$ ($j=1,\dots,K-1$) for some positive constant $\eta$. Furthermore, the conditional densities of $(V_{j},V_{j+1})$ given $(\bZ,K)$, denoted by $f_{j}(s,t|\bZ,K)$ ($j=1,\dots,K-1$), have continuous second-order derivatives with respect to $s$ and $t$ when $t-s>\eta$ and are continuously differentiable with respect to $\bZ$.
\end{condition}

\begin{condition}
\label{condition.p}
The full dimensionality $p$ satisfy $\log p=O(n^{\delta})$ for some $\delta\in(0,1)$. The effective dimensionality $s$ and the nonzero regression coefficients $\bbeta_{01}$ do not change with $n$.
\end{condition}

\begin{condition}
\label{condition.Fan.7}
Let $\kappa_{0}=\max_{\boldsymbol{v}\in\mathcal{N}_{0}}\kappa(\rho,\boldsymbol{v})$, where $\mathcal{N}_{0}=\left\{ \boldsymbol{v}\in\mathbb{R}^{s}:\|\boldsymbol{v}-\mathbf{\boldsymbol{\beta}}_{01}\|_{\infty}\le d_{n}^{\ast}\right\} $. The tuning parameter $\theta_{n}$ and the half minimal signal strength $d_{n}^{\ast}$ satisfy (i) $\theta_{n}\rho^{\prime}(d_{n}^{\ast})=o(1)$ and (ii) $\lambda_{\min}(\mathcal{I}_{11})>\theta_{n}\kappa_{0}$.
\end{condition}

Condition \ref{condition.Fan.1} defines the folded concave penalties introduced in \cite{fan2011nonconcave}, which include SCAD and MCP among others. Conditions \ref{condition.D1}--\ref{condition.D4} are the regularity conditions on the covariates and inspection times. They are adapted from Conditions C1--C4 of \cite{li2020adaptive}, which are a special case of the conditions in \cite{zeng2017maximum} in the context of low dimensional Cox models with interval censored data. The first part of Condition \ref{condition.p} indicates that our methods allow the full dimensionality $p$ to grow exponentially faster than the sample size. The second part of Condition \ref{condition.p} requires the nonzero part of the true regression coefficient vector $\bbeta$ to be constant when $n$ and $p$ increase, which is stronger than the counterparts for generalized linear models \citet{fan2011nonconcave} and the Cox model under right censoring \citet{bradic2011regularization} but is reasonable in the GWAS application presented later where the SNPs associated with early childhood caries are thought to be sparse relative to the sample size according to the literature. This requirement is due to the lack of closed-form expression of the hessian matrix of the log profile likelihood $l_{pn}(\bbeta)$---unlike the case of the Cox model with right censored data---making it difficult to prove with a changing $\bbeta_{01}$ the consistency of $-\ddot{l}_{pn11}(\bbeta)$ for $\mathcal{I}_{11}$ (Lemma 3 in the supplementary material),  a crucial result for establishing the asymptotic properties of this section. Condition \ref{condition.Fan.7} is assumed to ensure that the penalized oracle estimator in Lemma 1 of the supplementary material is a local maximizer of $\mathcal{C}(\beta)$ over $\mathbb{R}^p$ with probability tending to one, and is akin to Condition 7 in \cite{bradic2011regularization}. Under Condition \ref{condition.p}, $d_n^{\ast}$ is constant. So Condition \ref{condition.Fan.7}(i) implies $d_{n}^{\ast}\gg n^{-1/2}+\theta_{n}\rho^{\prime}(d_{n}^{\ast})$, which is similar to Condition 7(i) of \cite{bradic2011regularization}. Condition \ref{condition.Fan.7}(ii) is easily attainable if $\theta_n\to0$ and $\kappa_0$ is bounded. As discussed in \cite{bradic2011regularization}, Condition \ref{condition.Fan.7}(ii) is satisfied for SCAD and MCP penalties when $d_{n}^{\ast} \gg  \theta_n$, since this implies $\kappa_0=0$ for sufficiently large $n$.

\subsection{Properties of the nonconcave penalized likelihood methods}
We now list our main theoretical results for the penalized variable selection methods with folded concave penalties. The proofs of the following theorems and proposition are provided in the supplementary material.

\begin{theorem}
\label{thm.oracle}
Under Conditions \ref{condition.Fan.1}--\ref{condition.Fan.7} and assuming that $\theta_{n}\gg n^{\max\{-1/3,(\delta-1)/2\}}$ and $\theta_{n}\rho^{\prime}(d_{n}^{\ast})=O(n^{-1/3})$, with probability tending to one,
there exists a local maximizer $\hat{\bbeta}=(\hat{\bbeta}_1^{\top},\hat{\bbeta}_2^{\top})^{\top}$ of $\mathcal{C}(\bbeta)$ such that $\hat{\bbeta}_2=0$ and \[
\|\hat{\bbeta}-\bbeta_{0}\|_{2}=O_P\left(n^{-1/2}+\theta_{n}\rho^{\prime}(d_{n}^{\ast})\right).
\]
\end{theorem}

\begin{theorem} \label{thm.strong}
Under Conditions \ref{condition.Fan.1}--\ref{condition.Fan.7} and assuming that $\theta_{n}\gg n^{\max\{-1/3,(\delta-1)/2\}}$ and $\theta_{n}\rho^{\prime}(d_{n}^{\ast})=o(n^{-1/2})$, the penalized profile likelihood estimator for $\bbeta_1$ from the local maximizer in Theorem \ref{thm.oracle} satisfies
\[
\sqrt{n}(\hat{\bbeta}_{1}-\bbeta_{01})\stackrel{d}{\longrightarrow}N(0,\mathcal{I}_{11}^{-1}).
\]
\end{theorem}

\begin{proposition} \label{prop.global}
 Let $\hat{\bbeta}$ be a strict local maximizer of the nonconcave penalized log profile likelihood $\mathcal{C}(\bbeta)$ such that $\|\hat{\bbeta}\|_0\le s$ and $l_{pn}(\hat{\bbeta})\ge c$, where $c$ is a constant such that $c<l_{pn}(\bzero)$. For each $A\subset \left\{1,2,\dots,p\right\}$, we use $\ddot{l}_{pn,A}(\bbeta)$ to denote the corresponding $|A|\times |A|$ submatrix of $\ddot{l}_{pn}(\bbeta)$ and let $\mathcal{L}_{c}^{A}=\left\{ \bbeta\in\mathbb{R}^{p}: \beta_j=0\ \mbox{for}\ j\notin A\ \mbox{and}\ l_{pn}(\bbeta)\ge c\right\} $.

Assume that for each $A\subset \left\{1,2,\dots,p\right\}$ such that $|A|=2s$,
\[
\min_{\bbeta\in\mathcal{L}_{c}^{A}}\lambda_{\min}\left(-\ddot{l}_{pn,A}(\bbeta)\right)>\kappa(p_{\theta_{n}}),    \label{eq.global}
\]
where
\[
\kappa(p_{\theta_{n}})=\sup_{t_{1}<t_{2}\in(0,\infty)}-\frac{p_{\theta_{n}}^{\prime}(t_{2})-p_{\theta_{n}}^{\prime}(t_{1})}{t_{2}-t_{1}}.
\]
Then $\hat{\bbeta}$ is a global maximizer of $\mathcal{C}(\bbeta)$ on $\mathbb{S}_s$.
\end{proposition}

Theorem \ref{thm.oracle} and \ref{thm.strong} establish the oracle property \citep{fanli2001} of the penalized variable selection method with a folded concave penalty. Theorem \ref{thm.oracle} states that one can obtain a penalized profile likelihood estimator for $\bbeta$ with the same convergence rate as the counterpart under right censoring  \citep[][Theorem 4.2]{bradic2011regularization}. The rates of $\theta_n$ and $\rho^{\prime}(d_{n}^{\ast})$ are specified to ensure in conjunction with Condition 6 that the penalized oracle estimator in Lemma 1 of the supplementary material is a local maximizer of $\mathcal{C}(\beta)$ over $\mathbb{R}^p$ with probability tending to one. For SCAD and MCP, $\rho^{\prime}(d_{n}^{\ast})=0$ as long as $d_{n}^{\ast}>\alpha\theta_n$, and thus the requirement $\theta_{n}\rho^{\prime}(d_{n}^{\ast})=O(n^{-1/3})$ can be removed in this case. For lasso, the two conditions, $\theta_n \gg n^{-1/3}$ and $\theta_{n}\rho^{\prime}(d_{n}^{\ast})=O(n^{-1/3})$, cannot be satisfied simultaneously since $\rho^{\prime}(d_{n}^{\ast})=1$. That $n^{-1/3}$ appears in Theorem \ref{thm.oracle} is because it is the best bound we can obtain for the $L_{\infty}$ norm of the profile score vector, which stems from the convergence rate of $\hat{\Lambda}$. Theorem \ref{thm.strong} establishes the asymptotic normality of $\hat{\bbeta}_1$ with the same additional assumption $\theta_{n}\rho^{\prime}(d_{n}^{\ast})=o_{P}(n^{-1/2})$ as in Theorem 4.6 of \cite{bradic2011regularization}. Note that we do not impose conditions on the Hessian matrix of log profile likelihood function compared to \cite{bradic2011regularization} or require a consistent estimator of the $p$-dimensional efficient information matrix in contrast with \cite{zhao2019simultaneous} and \cite{WuZhSun}. This is by virtue of the fixed nonzero coefficient vector $\bbeta_{01}$ in Condition 5. While Theorem \ref{thm.oracle} and \ref{thm.strong} only exhibit the oracle property for one of the local maximizers of the penalized log profile likelihood $\mathcal{C}(\beta)$ instead of the estimator obtained from the optimization algorithm, Proposition \ref{prop.global} explores the conditions under which a local maximizer $\hat{\bbeta}$ becomes the global optimum on $\mathbb{S}_s$, which consists of all $p$-dimensional vectors with at most $s$ nonzero components. The proposition is analogous to Proposition 3(a) of \cite{fan2011nonconcave}.

\subsection{Properties of the adaptive lasso}
The asymptotic properties of the adaptive lasso method are presented below with the proofs relegated to the supplementary material.

\begin{theorem}
\label{thm.oracle.alasso}
Under Conditions \ref{condition.D1}--\ref{condition.p} and assuming that the initial estimator $\tilde{\boldsymbol{\beta}}$ is zero-consistent with rate $r_n$, $\theta_{n}=O(n^{-1/3})$ and $\theta_n r_n\gg n^{\max\{-1/3,(\delta-1)/2\}}$, with probability tending to one,
there exists a maximizer $\hat{\bbeta}=(\hat{\bbeta}_1^{\top},\hat{\bbeta}_2^{\top})^{\top}$ of $\mathcal{C}(\bbeta)$ with the adaptive
lasso penalty $p_{\theta_{n},\alpha}(|\beta_{j}|)=\theta_{n}|\beta_{j}|/|\tilde{\beta}_{j}|$ such that $\hat{\bbeta}_2=0$ and \[
\|\hat{\bbeta}-\bbeta_{0}\|_{2}=O_{P}\left(n^{-1/2}+\theta_{n}\right).
\]

\end{theorem}

\begin{theorem} \label{thm.strong.alasso}
Under Conditions \ref{condition.D1}--\ref{condition.p} and  assuming that the initial estimator $\tilde{\boldsymbol{\beta}}$ is zero-consistent with rate $r_n$, $\theta_{n}=o(n^{-1/2})$ and $\theta_n r_n\gg n^{\max\{-1/3,(\delta-1)/2\}}$, the penalized profile likelihood estimator $\hat{\bbeta}_{1}$ from Theorem \ref{thm.oracle.alasso} satisfies
\[
\sqrt{n}(\hat{\bbeta}_{1}-\bbeta_{01})\stackrel{d}{\longrightarrow}N(0,\mathcal{I}_{11}^{-1}).
\]
\end{theorem}

Analogous to Theorems \ref{thm.oracle} and \ref{thm.strong}, Theorems \ref{thm.oracle.alasso} and \ref{thm.strong.alasso} show that the adaptive lasso estimator has the oracle property \citep{fanli2001} as long as a reasonable initial estimator $\tilde{\bbeta}$ is available. Since the $L_2$ norm is always larger than the $L_{\infty}$ norm, an $L_2$-consistent initial estimator is zero-consistent with some rate $r_n$. Thus there always exists some $\theta_n=O(n^{-1/3})$ such that $\theta_n r_n \gg n^{\max\{-1/3,(\delta-1)/2\}}$ in Theorem \ref{thm.oracle.alasso} as long as $\delta\le 1/3$. For Theorem \ref{thm.strong.alasso}, if the initial estimator has a convergence rate faster than $n^{\min\{-1/6,-\delta/2\}}$ under the $L_2$ norm, we can have an $r_n$ of a larger order than $n^{\max\{1/6,\delta/2\}}$ and choose a $\theta_n$ such that $\theta_n = o(n^{-1/2})$ and $\theta_n r_n\gg n^{\max\{-1/3,(\delta-1)/2\}}$. In the subsequent simulation studies and real application, the lasso estimator was used as the initial estimator, but we did not show that it satisfies the aforementioned conditions due to the technical challenge of no closed-form expression of the profile likelihood for the Cox model with interval censored data. This is different from the case of the Cox model under right censoring where the partial likelihood can be employed to derive oracle inequalities for the lasso estimator \citep{huang2013oracle}. To the best of our knowledge, investigating the consistency of lasso in high-dimensional Cox models with interval censored data remains an open problem.

\section{Numerical Studies}\label{sec:simul}
We conducted simulation studies to evaluate the finite sample performance of the proposed variable selection methods employing the lasso, adaptive lasso, SCAD, and MCP penalties, respectively. We also compared their performance with that of the SCAD- and MCP-penalized Cox regressions with right-censored data, which were developed based on \citet{breheny2011coordinate} and \citet{simon2011regularization} and implemented by the function ``ncvsurv" in the R package ``ncvreg". To apply the``ncvsurv" function to interval censored data, we imputed the event times using the mid-points of the event-bracketing time intervals. To demonstrate the oracle property of our adaptive lasso, MCP and SCAD methods, we ran simulations of fitting the true model to the simulated data (the oracle procedure) as well. In all the numerical experiments, we simulated SNP data as features. We considered two numbers of SNPs, 1) $p = 3,000$ and 2) $p = 10,000$, and two scenarios about the correlations between SNPs, 1) linkage equilibrium, where SNPs at different loci are independent of each other, and 2) linkage disequilibrium, where there is a correlation between SNPs at different loci. The sample size $n$ was set to $500$ and $1,000$, and thus $n < p$ for all the scenarios.

We simulated SNPs as follows. First, we generated the minor allele frequency (MAF) of each SNP from a uniform distribution, $U(0.05, 0.20)$. Then, we generated a standard normal random variable for each SNP per subject and determined the minor allele count (0, 1 or 2) of the SNP, which served as a feature, by using cutoff values to trichotomize the random variable. The cutoff values were determined to make the distribution of the minor allele count satisfy the Hardy–Weinberg equilibrium with the generated MAF. To introduce linkage disequilibrium between the loci, the normal random variables were generated from a $p-$dimensional multivariate normal distribution with a zero mean vector and a covariance matrix whose $ij$-element is $\rho^{|i-j|}$. We set $\rho = 0$ for the scenario of linkage equilibrium and $\rho = 0.8$ for the scenario of linkage disequilibrium.

The time-to-event outcome of each subject was generated under a Cox model with a Weibull-type baseline hazard and the minor allele counts of SNPs as the covariates. Specifically, $\lambda(t|\bZ) = \kappa \eta (\eta t)^{\kappa -1} \exp(\bbeta^\top \bZ)$, where $\eta = 1.2$, $\kappa = 1.5$ and $\bZ$ denotes the vector of minor allele counts. We considered two values for the  regression coefficient vector: one with six nonzero coefficients, $\bbeta^{(6)} = (-1.40, -0.83, -1.64,  0.69,  1.39,  1.65)^\top$, and the other with 12 nonzero coefficients, $\bbeta^{(12)} = (\bbeta^{(6)\top}, -0.52,  0.86, -1.23,  1.18, -1.97, -1.68)^\top$. Note that, in each case, the set of nonzero parameters consist of equal numbers of small, medium, and large effects, defined by their absolute values falling within the respective ranges of $(0.5, 1.0)$, $(1.0, 1.5)$, and $(1.5, 2.0)$. For each subject, the inspection times $V_1,\ldots,V_6$ were generated from $V_t = V_{t-1} + U(0.1, (2 + t)/10)$ with $V_0 = 0$ for $t = 1, \cdots, 6$. With these inspection times, the proportion of subjects experiencing the event after their last inspection times ranged from 0.21 to 0.34 for the scenarios with $\bbeta^{(6)}$ and from 0.11 to 0.21 for the scenarios with $\bbeta^{(12)}$. Lastly, two hundred Monte Carlos runs were carried out in each simulation setting to assess the performance of the methods.

The simulation study was conducted using R 4.0.2 in the servers of MSU High Performance Computing Center (HPCC), each equipped with a 2.4 GHz 14-core Intel Xeon E5-2680 v4 processor and 128 GB RAM. The average computation time for completing a single simulation under the most time-demanding simulation setting (12 nonzero coefficients in the scenario of linkage disequilibrium with $p = 10000$ and $n = 1000$) was approximately 24 minutes for Lasso, 25 minutes for adaptive Lasso, 16 minutes for MCP, and 15 minutes for SCAD, respectively.

\begin{table}[htbp]
\renewcommand{\arraystretch}{0.8}
  \centering
  \caption{The performance metrics of $\hat{\bbeta}$ with $\bbeta^{(6)}$ and $\rho = 0$}
  \footnotesize
    \begin{tabular}{cccccccccc}
               \hline
    Method & $L_1$ error    & $L_2$ error    & FP    & FN    &       & $L_1$ error    & $L_2$ error    & FP    & FN  \vspace{0.2cm}\\
              \hline
          & \multicolumn{4}{l}{\textit{(n = 500; p = 3,000)}} &       & \multicolumn{4}{l}{\textit{(n = 1,000; p = 3,000)}} \\
    Oracle & 0.70  & 0.35  & 0.00  & 0.00  &       & 0.49  & 0.24  & 0.00  & 0.00 \\
    Lasso & 4.26  & 1.79  & 0.74  & 0.05  &       & 3.13  & 1.32  & 0.62  & 0.00 \\
    Adaptive lasso & 0.90  & 0.44  & 0.30  & 0.05  &       & 0.58  & 0.29  & 0.27  & 0.00 \\
    MCP   & 0.81  & 0.40  & 0.21  & 0.01  &       & 0.54  & 0.27  & 0.15  & 0.00 \\
    SCAD  & 0.99  & 0.51  & 0.34  & 0.03  &       & 0.49  & 0.24  & 0.00  & 0.00 \\
    MCP+mid-point & 553.38 & 53.68 & 184.89 & 0.00  &       & 697.54 & 45.43 & 403.77 & 0.00 \\
    SCAD+mid-point & 732.23 & 70.52 & 217.56 & 0.00  &       & 995.00 & 63.29 & 474.05 & 0.00 \vspace{0.2cm}\\
          & \multicolumn{4}{l}{\textit{(n = 500; p = 10,000)}} &       & \multicolumn{4}{l}{\textit{(n = 1,000; p = 10,000)}} \\
    Oracle & 0.70  & 0.35  & 0.00  & 0.00  &       & 0.46  & 0.23  & 0.00  & 0.00 \\
    Lasso & 4.61  & 1.94  & 0.58  & 0.12  &       & 3.39  & 1.43  & 0.56  & 0.00 \\
    Adaptive lasso & 0.91  & 0.46  & 0.21  & 0.12  &       & 0.55  & 0.27  & 0.25  & 0.00 \\
    MCP   & 0.74  & 0.37  & 0.08  & 0.00  &       & 0.48  & 0.24  & 0.05  & 0.00 \\
    SCAD  & 1.18  & 0.60  & 0.54  & 0.04  &       & 0.46  & 0.23  & 0.07  & 0.00 \\
    MCP+mid-point & 443.35 & 48.65 & 150.18 & 0.01  &       & 568.95 & 43.38 & 320.08 & 0.00 \\
    SCAD+mid-point & 590.18 & 64.37 & 192.98 & 0.01  &       & 759.70 & 57.43 & 401.32 & 0.00 \\
               \hline
    \end{tabular}%
  \label{table1}%
\end{table}%

\begin{table}[htbp]
\renewcommand{\arraystretch}{0.8}
  \centering
  \caption{The performance metrics of $\hat{\bbeta}$ with $\bbeta^{(6)}$ and $\rho = 0.8$}
  \footnotesize
    \begin{tabular}{cccccccccc}
               \hline
    Method & $L_1$ error    & $L_2$ error    & FP    & FN    &       & $L_1$ error    & $L_2$ error    & FP    & FN  \vspace{0.2cm}\\
              \hline
          & \multicolumn{4}{l}{\textit{(n = 500; p = 3,000)}} &       & \multicolumn{4}{l}{\textit{(n = 1,000; p = 3,000)}} \\
    Oracle & 0.69  & 0.34  & 0.00  & 0.00  &       & 0.48  & 0.24  & 0.00  & 0.00 \\
    Lasso & 4.32  & 1.80  & 0.93  & 0.04  &       & 3.19  & 1.33  & 0.95  & 0.00 \\
    Adaptive lasso & 0.86  & 0.42  & 0.26  & 0.04  &       & 0.52  & 0.26  & 0.14  & 0.00 \\
    MCP   & 0.75  & 0.37  & 0.12  & 0.01  &       & 0.50  & 0.25  & 0.07  & 0.00 \\
    SCAD  & 0.94  & 0.48  & 0.44  & 0.02  &       & 0.47  & 0.23  & 0.00  & 0.00 \\
    MCP+mid-point & 571.60 & 53.24 & 190.67 & 0.01  &       & 586.36 & 37.42 & 384.41 & 0.00 \\
    SCAD+mid-point & 761.06 & 70.03 & 221.30 & 0.01  &       & 911.87 & 55.95 & 462.20 & 0.00 \vspace{0.2cm} \\
          & \multicolumn{4}{l}{\textit{(n = 500; p = 10,000)}} &       & \multicolumn{4}{l}{\textit{(n = 1,000; p = 10,000)}} \\
    Oracle & 0.73  & 0.36  & 0.00  & 0.00  &       & 0.49  & 0.24  & 0.00  & 0.00 \\
    Lasso & 4.65  & 1.95  & 0.72  & 0.09  &       & 3.42  & 1.44  & 0.73  & 0.00 \\
    Adaptive lasso & 0.93  & 0.46  & 0.22  & 0.09  &       & 0.54  & 0.27  & 0.15  & 0.00 \\
    MCP   & 0.80  & 0.39  & 0.14  & 0.01  &       & 0.53  & 0.26  & 0.09  & 0.00 \\
    SCAD  & 1.17  & 0.59  & 0.53  & 0.03  &       & 0.50  & 0.25  & 0.02  & 0.00 \\
    MCP+mid-point & 457.11 & 48.97 & 154.33 & 0.01  &       & 559.09 & 41.51 & 325.07 & 0.00 \\
    SCAD+mid-point & 592.61 & 63.56 & 194.42 & 0.01  &       & 735.38 & 54.15 & 399.99 & 0.00 \\
               \hline
    \end{tabular}%
  \label{table2}%
\end{table}%

\begin{table}[htbp]
\renewcommand{\arraystretch}{0.8}
  \centering
  \caption{The performance metrics of $\hat{\bbeta}$ with $\bbeta^{(12)}$ and $\rho = 0$}
  \footnotesize
    \begin{tabular}{cccccccccc}
               \hline
    Method & $L_1$ error    & $L_2$ error    & FP    & FN    &       & $L_1$ error    & $L_2$ error    & FP    & FN  \vspace{0.2cm}\\
              \hline
          & \multicolumn{4}{l}{\textit{(n = 500; p = 3,000)}} &       & \multicolumn{4}{l}{\textit{(n = 1,000; p = 3,000)}} \\
    Oracle & 1.51  & 0.55  & 0.00  & 0.00  &       & 1.00  & 0.36  & 0.00  & 0.00 \\
    Lasso & 10.12 & 3.04  & 1.68  & 0.54  &       & 7.75  & 2.33  & 1.73  & 0.01 \\
    Adaptive lasso & 2.19  & 0.79  & 0.60  & 0.57  &       & 1.24  & 0.44  & 0.50  & 0.01 \\
    MCP   & 1.72  & 0.63  & 0.33  & 0.17  &       & 1.06  & 0.39  & 0.19  & 0.00 \\
    SCAD  & 2.08  & 0.76  & 0.67  & 0.22  &       & 1.07  & 0.40  & 0.09  & 0.01 \\
    MCP+mid-point & 523.46 & 55.44 & 152.07 & 0.15  &       & 656.38 & 47.98 & 338.90 & 0.00 \\
    SCAD+mid-point & 674.02 & 70.93 & 180.19 & 0.16  &       & 913.69 & 65.76 & 399.37 & 0.01  \vspace{0.2cm}\\
          & \multicolumn{4}{l}{\textit{(n = 500; p = 10,000)}} &       & \multicolumn{4}{l}{\textit{(n = 1,000; p = 10,000)}} \\
    Oracle & 1.49  & 0.54  & 0.00  & 0.00  &       & 1.02  & 0.37  & 0.00  & 0.00 \\
    Lasso & 10.99 & 3.31  & 1.28  & 1.00  &       & 8.36  & 2.52  & 1.44  & 0.03 \\
    Adaptive lasso & 2.55  & 0.92  & 0.55  & 1.01  &       & 1.22  & 0.43  & 0.31  & 0.03 \\
    MCP   & 1.72  & 0.64  & 0.27  & 0.25  &       & 1.06  & 0.38  & 0.12  & 0.01 \\
    SCAD  & 2.36  & 0.86  & 1.02  & 0.27  &       & 1.15  & 0.43  & 0.15  & 0.03 \\
    MCP+mid-point & 430.82 & 50.87 & 124.21 & 0.21  &       & 532.49 & 45.28 & 268.58 & 0.01 \\
    SCAD+mid-point & 555.28 & 65.38 & 160.03 & 0.24  &       & 707.47 & 59.68 & 338.66 & 0.02 \\
               \hline
    \end{tabular}%
  \label{table3}%
\end{table}%

\begin{table}[htbp]
\renewcommand{\arraystretch}{0.8}
  \centering
  \caption{The performance metrics of $\hat{\bbeta}$ with $\bbeta^{(12)}$ and $\rho = 0.8$}
  \footnotesize
    \begin{tabular}{cccccccccc}
               \hline
    Method & $L_1$ error    & $L_2$ error    & FP    & FN    &       & $L_1$ error    & $L_2$ error    & FP    & FN  \vspace{0.2cm}\\
              \hline
          & \multicolumn{4}{l}{\textit{(n = 500; p = 3,000)}} &       & \multicolumn{4}{l}{\textit{(n = 1,000; p = 3,000)}} \\
    Oracle & 1.54  & 0.56  & 0.00  & 0.00  &       & 0.99  & 0.36  & 0.00  & 0.00 \\
    Lasso & 10.49 & 3.14  & 2.23  & 0.80  &       & 7.93  & 2.38  & 2.42  & 0.00 \\
    Adaptive lasso & 2.47  & 0.88  & 0.71  & 0.83  &       & 1.22  & 0.43  & 0.33  & 0.01 \\
    MCP   & 1.83  & 0.67  & 0.44  & 0.18  &       & 1.04  & 0.38  & 0.15  & 0.01 \\
    SCAD  & 2.00  & 0.74  & 0.86  & 0.18  &       & 1.03  & 0.38  & 0.09  & 0.01 \\
    MCP+mid-point & 512.00 & 52.63 & 155.49 & 0.19  &       & 592.83 & 41.71 & 332.41 & 0.01 \\
    SCAD+mid-point & 678.54 & 69.09 & 181.36 & 0.23  &       & 863.78 & 59.12 & 395.21 & 0.01 \vspace{0.2cm}\\
          & \multicolumn{4}{l}{\textit{(n = 500; p = 10,000)}} &       & \multicolumn{4}{l}{\textit{(n = 1,000; p = 10,000)}} \\
    Oracle & 1.47  & 0.54  & 0.00  & 0.00  &       & 0.99  & 0.36  & 0.00  & 0.00 \\
    Lasso & 11.28 & 3.39  & 1.51  & 1.33  &       & 8.53  & 2.56  & 1.84  & 0.02 \\
    Adaptive lasso & 2.91  & 1.04  & 0.48  & 1.34  &       & 1.27  & 0.45  & 0.34  & 0.02 \\
    MCP   & 1.81  & 0.67  & 0.36  & 0.28  &       & 1.05  & 0.39  & 0.14  & 0.02 \\
    SCAD  & 2.56  & 0.93  & 1.13  & 0.32  &       & 1.12  & 0.42  & 0.19  & 0.03 \\
    MCP+mid-point & 431.73 & 50.03 & 127.02 & 0.30  &       & 529.11 & 43.79 & 272.98 & 0.02 \\
    SCAD+mid-point & 572.58 & 66.26 & 162.13 & 0.30  &       & 695.59 & 57.01 & 340.31 & 0.02 \\
               \hline
    \end{tabular}%
  \label{table4}%
\end{table}%

Tables \ref{table1} to \ref{table4} summarize the performance in estimating $\bbeta$ of the proposed methods as well as the MCP- and SCAD-penalized Cox regressions with mid-point imputation. We calculated four performance metrics: the estimation accuracy was evaluated using the $L_1 \equiv E(\lVert \hat{\bbeta} - \bbeta \rVert_1)$ and $L_2 \equiv E(\lVert \hat{\bbeta} - \bbeta \rVert_2)$ errors, and the variable selection performance was evaluated using the expected numbers of false positives (FP) and  false negatives (FN), where ``positives'' and ``negatives'' mean nonzero and zero estimated coefficients, respectively.

The penalized Cox regressions for right censored data performed the worst across all the simulation scenarios, having a much larger estimation error and a much higher number of false positives than the other methods, probably due to the mid-point imputation of even times. The proposed methods except the one with the lasso penalty performed very well in terms of variable selection. Their expected numbers of false positives and false negatives are below one except that the adaptive lasso's FN and the SCAD's FP are a bit above one in the setting with the smallest $n$, the largest $p$ and the lowest sparsity. In addition, the coefficient estimation errors of the adaptive lasso, MCP and SCAD were comparable to that of the oracle procedure, especially when $n=1000$. The MCP performed better in terms of estimation error and variable selection accuracy than the adaptive lasso and SCAD when $n=500$, but there was no significant performance difference between the three methods when $n=1000$, except that the adpative lasso had a bit higher FP. The lasso had much larger estimation errors than the adaptive lasso, MCP and SCAD, but the errors decreased as the sample size went up, implying estimation consistency. Interestingly, the lasso's FP did not decrease with the sample size except in the setting with the highest sparsity and independent features, although that FP was already very low. The lasso's FN decreased as the sample size grew. The lower sparsity worsened every performance metric of the proposed methods, but the higher correlation between the features did not make everything worse.




The estimation results of individual nonzero coefficients and the baseline cumulative hazard function are relegated to Supplementary Material. The mean estimates obtained by the proposed methods with all the penalties except lasso were close to the true values. The empirical standard error of every coefficient estimator obtained by the adaptive lasso, MCP and SCAD decreased as the sample size increased, and it became similar to that of the oracle estimator when $n = 1,000$. Furthermore, the average estimates of $\Lambda(t)$ are reasonably close to the true baseline cumulative hazard function for the methods that possess the oracle property. The asymptotic normality of the estimators for the non-zero coefficients was demonstrated by the Normal Q-Q plots in Supplementary Material.



\section{Application}\label{sec:appl}
We applied the proposed methods to a dbGaP data set (dbGaP accession: phs000095.v3.p1) titled Dental Caries: Whole Genome Association and Gene x Environment Studies, which contains caries assessments and whole-genome genotyping information of 5418 subjects across four sites (PITT, DRDR, IOWA and GEIRS) from University of Pittsburgh and University of Iowa. Each subject's caries assessment was from only one time point. We aimed to identify the typed SNPs that are independently associated with age to early child caries (ECC). Thus, we chose the subjects who were under 6 years of age at the caries assessment as the analysis cohort.
The phenotype, age to ECC, is subject to case 1 interval censoring due to the single caries assessment.

We performed both the SNP- and the subject-level filtering of the genotype data using PLINK 1.9. Specifically, a SNP was excluded if it met any of the following three criteria: (1) its minor allele frequency (MAF) was less than 0.01, (2) it had a p-value below $10^{-6}$ in the Hardy-Weinberg Equilibrium (HWE) test, or (3) the missing data rate exceeded 2$\%$. A subject was removed from the analysis if his/her genotype missing rate was above $2\%$. The quality control led to a final analysis data set of 1118 subjects across three study sites (PITT, IOWA, and GEIRS). The study site DRDR was excluded from the analysis data set because all the participants at that site were over the age of 6. Each subject has 539,436 SNPs in the final data set. The missing genotypes of a SNP were imputed by sampling from a binomial distribution, $Bin(2,\widehat{\mbox{MAF}})$, where $\widehat{\mbox{MAF}}$ was the sample minor allele frequency estimated using the non-missing genotypes.

It is a formidably time-consuming task to apply the proposed variable selection methods directly to the 539,436 SNPs to identify the ones that are independently associated with age to ECC. To reduce the computational burden, we first selected the top 10,000 SNPs in terms of the strength of association with the phenotype. The association strength of a SNP was assessed by the p-value of the Wald test for the SNP effect based on a Cox model of age to ECC adjusted for sex, study site, self-reported race, and the top 20 principal components of the SNP data. Then, we applied the proposed variable selection methods to the 10,000 SNPs and the interval-censored age-to-ECC data to search for the SNPs independently associated with age to ECC. In the Cox model for variable selection, we also considered sex and study site as unpenalized covariates to adjust for potential confounding.

\begin{threeparttable}[H]
  \centering
  \scriptsize
  \caption{The SNPs identified to be independently associated with age to ECC by the proposed variable selection methods and the corresponding interrogation results from IPA.}
    \begin{tabular}{cccc}
    \toprule
     SNP &  Gene & Dental/oral diseases or functions & Coefficient estimate \\
    \midrule
    rs6000495 & CSF2RB &    -   & 0.77(MCP) \\
    rs7959625 & CUX2  &     -  & -0.55(MCP) \\
    rs7643642 & GK5   & oral squamous cell carcinoma & -0.73(MCP) \\
    rs17826057 & IRAK3 &    -   & 0.87(MCP) \\
    rs2638500 & KRT79 &    -   & -0.3(MCP) \\
    rs2325610 & LOC105370259 &     -  & -0.36(MCP) \\
    rs980561 & MPPED2 & oral squamous cell carcinoma & 0.42(MCP) \\
rs10987080\tnote{$\dagger$} & PBX3  &   oral squamous cell carcinoma (tongue)    & 0.01(Lasso), -0.35(Adaptive lasso), -0.56(MCP) \\
    rs4624889 & PHACTR2 &    -   & -0.52(MCP) \\
    rs11868735 & RAB11FIP4 &    -   & 0.62(MCP) \\
    rs6657332 & RYR2  & oral squamous cell carcinoma (tongue) & -0.43(MCP) \\
    rs7224155 & SKAP1 &   -    & -0.47(MCP) \\
    rs2331766 & SLC7A7 &  -     & 2(MCP) \\
    rs2291306 & TTN   & oral squamous cell carcinoma & -0.56(MCP) \\
    rs515028 &    -   &    -   & 0.61(MCP) \\
    rs1019595 &    -   &   -    & 0.99(MCP) \\
    rs1126967 &    -   &    -   & -0.48(MCP) \\
    rs1862986 &    -   &   -    & 0.89(MCP) \\
    rs10800364 &    -   &  -     & -0.52(MCP) \\
    rs11670677 &   -    &  -     & 0.46(MCP) \\
    rs2928719 &   -    &   -    & -0.46(MCP) \\
    rs7012400 &   -    &   -    & 0.54(MCP) \\
    rs1032673 &   -    &   -    & 0.59(MCP) \\
    rs6561237 &   -    &   -    & -0.51(MCP) \\
    rs9320811 &   -    &   -    & 1.61(MCP) \\
    \bottomrule
    \end{tabular}%
  \begin{tablenotes}
    \item[$\dagger$] This SNP was linked to PBX3 in IPA according to \citet{wang2013association}; however, it is no longer present in the current version of IPA (Version 111725566).
  \end{tablenotes}
  \label{analysis}
\end{threeparttable}

Table \ref{analysis} shows the SNPs that were identified to be independently associated with age to ECC by any of the proposed variable selection methods. Both lasso and adaptive lasso identified only one SNP, while MCP found a total of 25 SNPs. SCAD, however, did not detect any. We used the Ingenuity Pathway Analysis software (IPA, Version 111725566)
to map the detected SNPs to genes. Additionally, we acquired the dental- and oral-related phenotype annotations of those genes by searching the IPA using ``dental'' and ``oral'' respectively as the keyword in the ``Diseases and Functions" category. The SNP identified by lasso, adaptive lasso, and MCP, \textit{rs10987080}, was reported in another genetic study of primary tooth caries \citep{wang2013association}. They mapped this SNP to the \textit{PBX3} gene using IPA, although this mapping is no longer present in the current version of the software. According to IPA, among the genes listed in Table \ref{analysis}, \textit{GK5}, \textit{MPPED2}, \textit{PBX3}, \textit{RYR2} and \textit{TTN} have been previously reported in oral disease genetics literature \citep{hedberg2016genetic, HAYES2016106, cui2022construction}. All the five genes were found to be associated with squamous cell carcinoma in human oral cavity. Particularly, \textit{PBX3} and \textit{RYR2} were associated with oral squamous cell carcinoma of tongue \citep{HAYES2016106}. In addition, the \textit{RYR2} gene was known to be mutated during transformation of oral dysplasia to cancer \citep{singh2020study}, and it emerged as one of the most commonly mutated genes in a whole-exome sequencing study that characterized the mutational patterns of oral squamous cell carcinoma tumors \citep{patel2021whole}. The \textit{MPPED2} gene was previously reported to have a possible association with the age to ECC in \cite{wu2021}, where a multi‐marker genetic association test for interval-censored data was developed and applied to the same data set we analyzed. The associations between \textit{MPPED2} and other ECC phenotypeslike dft have also been implicated in  \citet{shaffer2011genome} and \citet{stanley2014genetic}, although its functional role in caries etiology remains unclear.

\section{Discussion}\label{sec:disc}
We developed a set of penalized variable selection methods for high-dimensional Cox models with interval-censored data. The methods with adaptive lasso or folded concave penalties were proven to enjoy the oracle property, even when the dimensionality grows exponentially with the sample size.

The proposed variable selection methods can be extended to survival data subject to both interval censoring and left truncation. The left truncation problem arises when a random sample is collected from the individuals whose failure times are greater than their study entry times. Let $V_{i0}$ be the study entry time, also called left truncation time, for the $i$-th subject. Then, the log likelihood function of left-truncated and interval-censored data under the Cox model is
\begin{align}\label{loglikLTIC}
l_{n}(\bbeta,\Lambda)=\sum_{i=1}^{n}\log\left[\exp\{-\Lambda(L_i|\bZ_i; \bbeta)+\Lambda(V_{i0}|\bZ_i; \bbeta)\}-\exp\{-\Lambda(R_i|\bZ_i; \bbeta)+\Lambda(V_{i0}|\bZ_i; \bbeta)\}\right].
\end{align}
The associated NPMLE $\tilde{\Lambda}(t)$ of the cumulative baseline hazard function increases only within a new set of maximal intersections, which are disjoint intervals of the form $(l,u]$, where $l \in \lbrace L_i: i = 1, \ldots, n\rbrace$, $u \in \lbrace R_i, V_{i0}: i = 1, \ldots, n, R_i < \infty\rbrace$, and $(l, u)$ contains no $L_i$, $R_i$ or $V_{i0}$ for all $i$'s \citep{alioum1996proportional}. Moreover, the log likelihood \eqref{loglikLTIC} is indifferent to how $\tilde{\Lambda}(t)$ increases within those intervals, as only the overall jump sizes of $\tilde{\Lambda}(t)$ over $(l,u]$'s with $l>0$ and $u<\infty$ affect it. Write the intervals $(l,u]$'s with $l>0$ and $u<\infty$ as $(l_1,u_1],\ldots,(l_m,u_m]$, where $u_k<u_{k+1}$ $(k=1,\ldots,m-1)$. One can perform variable selection with left-truncated and interval-censored data using the same EM algorithm as for interval censored data, except with a modification of substituting the summation index $u_k \le L_i$ with $V_{i0} < u_k \le L_i$ throughout the algorithm.

Four future research topics are worth pursuing. First, it is desirable to extend the proposed methods to other survival models, since the proportion hazards assumption made by the Cox model might not hold in some real applications, especially those with long follow-ups. \citet{zeng2016maximum} developed a nonparametric maximum likelihood estimation method for a class of semiparametric transformation models, which includes the Cox model as a member, under interval censoring. Their method uses the same EM algorithm as ours except without any penalty and having an additional latent variable that induces the transformation model. So, computationally, it is straightforward to extending our methods to the semiparametric transformation models. We conjecture that the theoretical results and their proofs in this paper are also generalizable to those models. The other research direction meriting investigation is grouped feature selection with interval-censored data. This is motivated by our real application, where the SNPs can be naturally grouped into genes. To perform grouped feature selection, one can change the penalty in our methods to group lasso \citep{YuLi2006} for group selection or group bridge \citep{Huangetal2009} for bi-level selection. Other group penalties like group SCAD and group MCP \citep{breheny2015group} can also be used. The penalized estimation can still be carried out under the EM framework, but the M step needs to use a group descent algorithm \citep{simon2012standardization,breheny2015group} instead of the regular coordinate descent, and one needs to orthonormalize the groups upfront for better group identification \citep{simon2012standardization,breheny2015group}. The theory of the grouped feature selection is expected to require more effort to derive. The other two topics are to allow the number of nonzero coefficients to increase with $n$ and to prove that the lasso estimator, as the initial estimator for the adaptive lasso, satisfies the conditions in Theorem \ref{thm.strong.alasso} (or find such an initial estimator), respectively.

	\section*{Supplementary Material}
	
	The online Supplementary Material contains proofs of all the theorems and propositions as well as additional simulation results.
	\par
	\section*{Acknowledgements}
	
	This work was partially supported by the National Institute of Dental and Craniofacial Research (Award No. R56DE030437 and R03DE032357), Natural Science Foundation (DMS-1915099), and the National Research Foundation of Korea (NRF) grant (2021R1G1A1009269). Funding support for the study entitled Dental Caries: Whole Genome Association and Gene x Environment
Studies was provided by the National Institute of Dental and Craniofacial Research (NIDCR, grant number
U01-DE018903). This genome-wide association study is part of the Gene Environment Association Studies
(GENEVA) program of the trans-NIH Genes, Environment and Health Initiative (GEI). Genotyping services were
provided by the Center for Inherited Disease Research (CIDR). CIDR is fully funded through a federal contract
from the National Institutes of Health (NIH) to The Johns Hopkins University, contract number
HHSN268200782096C. Funds for this project's genotyping were provided by the NIDCR through CIDRs NIH
contract. Assistance with phenotype harmonization and genotype cleaning, as well as with general study
coordination, was provided by the GENEVA Coordinating Center (U01-HG004446) and by the National Center for
Biotechnology Information (NCBI). Data and samples were provided by: (1) the Center for Oral Health Research
in Appalachia (a collaboration of the University of Pittsburgh and West Virginia University funded by NIDCR
R01-DE 014899); (2) the University of Pittsburgh School of Dental Medicine (SDM) DNA Bank and Research
Registry, supported by the SDM and by the University of Pittsburgh Clinical and Translational Sciences Institute
(funded by NIH/NCRR/CTSA Grant UL1-RR024153); (3) the Iowa Fluoride Study and the Iowa Bone
Development Study, funded by NIDCR (R01-DE09551and R01-DE12101, respectively); and (4) the Iowa
Comprehensive Program to Investigate Craniofacial and Dental Anomalies (funded by NIDCR, P60-DE-013076).
The datasets used for the analyses described in this manuscript were obtained from dbGaP at
\url{https://www.ncbi.nlm.nih.gov/projects/gap/cgi-bin/study.cgi?study_id=phs000095.v3.p1}    through dbGaP accession
number phs000095.v3.p1.
	\par
	

\bibhang=1.7pc
\bibsep=2pt
\fontsize{9}{14pt plus.8pt minus .6pt}\selectfont
\renewcommand\bibname{\large \bf References}
\expandafter\ifx\csname
natexlab\endcsname\relax\def\natexlab#1{#1}\fi
\expandafter\ifx\csname url\endcsname\relax
  \def\url#1{\texttt{#1}}\fi
\expandafter\ifx\csname urlprefix\endcsname\relax\def\urlprefix{URL}\fi

\bibliographystyle{chicago}      
\bibliography{ref}   

\vskip .65cm
\noindent
Daewoo Pak\\
Division of Data Science, Yonsei University, Wonju, Korea.
\vskip 2pt
\noindent
E-mail: dpak@yonsei.ac.kr
\vskip 2pt

\noindent
Jianrui Zhang \\
Department of Statistics and Probability, Michigan State University, East Lansing, MI 48824, USA.
\vskip 2pt
\noindent
E-mail: zhan1783@msu.edu

\vskip 2pt

\noindent
Di Wu \\
Department of Epidemiology and Biostatistics, Michigan State University, East Lansing, MI 48824, USA.
\vskip 2pt
\noindent
E-mail: wudi14@msu.edu

\vskip 2pt

\noindent
Haolei Weng \\
Department of Statistics and Probability, Michigan State University, East Lansing, MI 48824, USA.
\vskip 2pt
\noindent
E-mail: wenghaol@msu.edu

\vskip 2pt

\noindent
Chenxi Li \\
Department of Epidemiology and Biostatistics, Michigan State University, East Lansing, MI 48824, USA.
\vskip 2pt
\noindent
E-mail: cli@msu.edu

\end{document}


	
	\renewcommand{\baselinestretch}{2}
	
	\markright{ \hbox{\footnotesize\rm Statistica Sinica
		}\hfill\\[-13pt]
		\hbox{\footnotesize\rm
		}\hfill }
	
	
	\renewcommand{\thefootnote}{}
	$\ $\par
	
	
	\fontsize{12}{14pt plus.8pt minus .6pt}\selectfont \vspace{0.8pc}
	\centerline{\large\bf Supplementary Material to ``Variable Selection in Ultra-high }
	\vspace{2pt}
	\centerline{\large\bf  Dimensional Feature Space for the Cox Model}
 \vspace{2pt}
        \centerline{\large\bf with Interval-Censored Data''}


\section{Proofs of the Theoretical Results}
We define some notations used throughout this section here. We use the empirical process notation. Let $\mathbb{P}_{n}$ be the empirical measure for $n$ independent subjects and $P$ be the true probability measure. The corresponding empirical process is $\mathbb{G}_{n}=\sqrt{n}(\mathbb{P}_{n}-P)$. Also define
\[
\dot{l}_{pn}(\mathbf{\boldsymbol{\beta}})=\frac{\partial l_{pn}(\mathbf{\boldsymbol{\beta}})}{\partial\mathbf{\boldsymbol{\beta}}}\quad\text{and}\quad\ddot{l}_{pn}(\mathbf{\boldsymbol{\beta}})=\frac{\partial^{2}l_{pn}(\mathbf{\boldsymbol{\beta}})}{\partial\mathbf{\boldsymbol{\beta}}\partial\mathbf{\boldsymbol{\beta}}^{\top}},
\]
and we can write
\[
\dot{l}_{pn}(\mathbf{\boldsymbol{\beta}})=\begin{pmatrix}\dot{l}_{pn1}(\mathbf{\boldsymbol{\beta}})\\
\dot{l}_{pn2}(\mathbf{\boldsymbol{\beta}})
\end{pmatrix}\quad\text{and}\quad\ddot{l}_{pn}(\mathbf{\boldsymbol{\beta}})=\begin{pmatrix}\ddot{l}_{pn11}(\mathbf{\boldsymbol{\beta}}) & \ddot{l}_{pn12}(\mathbf{\boldsymbol{\beta}})\\
\ddot{l}_{pn21}(\mathbf{\boldsymbol{\beta}}) & \ddot{l}_{pn22}(\mathbf{\boldsymbol{\beta}})
\end{pmatrix},
\]
where $\dot{l}_{pn1}(\mathbf{\boldsymbol{\beta}})$ is an $s\times1$ vector and $\ddot{l}_{pn11}(\mathbf{\boldsymbol{\beta}})$ is an $s\times s$ matrix.

Let $l(\mathbf{\boldsymbol{\beta}},\Lambda)$ be the log likelihood for a single subject and $l_{\boldsymbol{\beta}}(\mathbf{\boldsymbol{\beta}},\Lambda)$ denote the score function of $l(\mathbf{\boldsymbol{\beta}},\Lambda)$ with respect to $\mathbf{\boldsymbol{\beta}}$. To obtain the score function of $l(\mathbf{\boldsymbol{\beta}},\Lambda)$ with respect to $\Lambda$, we define $l_{\Lambda}(\mathbf{\boldsymbol{\beta}},\Lambda)(h)=\partial l(\mathbf{\boldsymbol{\beta}},\Lambda_{\epsilon,h})/\partial\epsilon|_{\epsilon=0}$, where $d\Lambda_{\epsilon,h}=(1+\epsilon h)d\Lambda$ is the parametric submodels for $\Lambda$ utilized by \cite{zeng2016maximum}
with $h$ running through the space of bounded functions. The second-order derivatives evaluated at $(\mathbf{\boldsymbol{\beta}}_{0},\Lambda_{0})$ can be defined similarly as in \cite{zeng2016maximum}. We write
\[
l_{\bbeta}(\bbeta,\Lambda)=\begin{pmatrix}l_{\bbeta,1}(\bbeta,\Lambda)\\
l_{\bbeta,2}(\bbeta,\Lambda)
\end{pmatrix}
\]
where $l_{\bbeta,1}$ is a $s\times 1$ vector corresponding to $\bbeta_{01}$. Let $\boldsymbol{h}_{01}$ denote the $s$-dimensional least favorable direction \citep{murphy2000profile} for the Oracle $l_n(\bbeta_1,\Lambda)$ at $(\mathbf{\boldsymbol{\beta}}_{01},\Lambda_{0})$, which exists under Conditions 2-4 \citep[see the proof of Theorem 2 in][]{zeng2016maximum}. Then the efficient score function for $\bbeta_1$ in the Oracle is
\[
\tilde{l}_1=l_{\boldsymbol{\beta},1}(\mathbf{\boldsymbol{\beta}}_{0},\Lambda_{0})-l_{\Lambda}(\mathbf{\boldsymbol{\beta}}_{0},\Lambda_{0})(\boldsymbol{h}_{01}).
\]

\subsection{Proof of Theorem 1}

Theorem 1 is a direct result of the following Lemma \ref{lemma.eb} and Lemma \ref{lemma.oracle}.

\begin{lemma} \label{lemma.eb}
Define a penalized oracle estimator $\hat{\bbeta}^{o}=(\hat{\bbeta}_{1}^{o\top},\bzero^{\top})^{\top}$ as a local maximizer of $\mathcal{C}(\bbeta)=\mathcal{C}(\bbeta_{1},\bzero)$ where $\text{dim}(\bbeta_{1})=s$. Under Conditions 1-5 and assuming that $\theta_{n}\rho^{\prime}( d_{n}^{\ast})=o(1)$, with probability tending to one, there exists a penalized oracle estimator $\hat{\bbeta}^{o}$ such that
\[
\|\hat{\bbeta}^{o}-\bbeta_0\|_{2}=O_{P}\left(n^{-1/2}+\theta_{n}\rho^{\prime}( d_{n}^{\ast})\right).
\]
\end{lemma}

\begin{proof}
We follow the proof of Theorem 4.2 in \cite{bradic2011regularization} and the proof of Theorem 1 in \cite{li2020adaptive}. It suffices to show that, for any $\epsilon>0$, there exists a large constant $B$ and
$\gamma_{n}=B\left(n^{-1/2}+\theta_{n}\rho^{\prime}( d_{n}^{\ast})\right)$ such that for sufficiently large $n$,
\[
P\left(\sup_{\|\boldsymbol{u}\|_{2}=1}\mathcal{C}(\bbeta_{01}+\gamma_{n}\boldsymbol{u},\bzero)<\mathcal{C}(\bbeta_{01},\bzero)\right)\ge1-\epsilon.
\]

Let $\tilde{\bbeta}_{n1}=\argmax_{\bbeta_1\in\{\bbeta_{01}+\gamma_{n}\boldsymbol{u}:\|\boldsymbol{u}\|_{2}=1\}}l_{pn}(\bbeta_1,\bzero)$. Then $\tilde{\bbeta}_{n1}=\bbeta_{01}+o_{P}(1)$ due to Condition 5, and
\[
\mathcal{C}(\bbeta_{01}+\gamma_{n}\boldsymbol{u},\bzero)-\mathcal{C}(\bbeta_{01},\bzero)\le l_{pn}(\tilde{\bbeta}_{n1},\bzero)- l_{pn}(\bbeta_{01},\bzero)-\bone^{\top}p_{\theta_{n},\alpha}(\bbeta_{01}+\gamma_{n}\boldsymbol{u},\bzero)+\bone^{\top}p_{\theta_{n},\alpha}(\bbeta_{01},\bzero).
\]
Note that the effective dimensionality $s$ is fixed. So, by remark A1 in \cite{zeng2016maximum}, Theorem 1 in \cite{murphy2000profile} is applicable to the Oracle $l_n(\bbeta_1,\Lambda)$, which implies that
\begin{align*}
& l_{pn}(\tilde{\bbeta}_{n1},\bzero)- l_{pn}(\bbeta_{01},\bzero)   \\
= &~(\tilde{\bbeta}_{n1}-\bbeta_{01})\mathbb{P}_{n}\tilde{l}_1-\frac{1}{2}(\tilde{\bbeta}_{n1}-\bbeta_{01})^{\top}\mathcal{I}_{11}(\tilde{\bbeta}_{n1}-\bbeta_{01})+o_{P}\left\{\left(\|\tilde{\bbeta}_{n1}-\bbeta_{01}\|+n^{-1/2}\right)^{2}\right\}.
\end{align*}
It is obvious that
\begin{align*}
|(\tilde{\bbeta}_{n1}-\bbeta_{01})\mathbb{P}_{n}\tilde{l}_1| & =O_{P}(n^{-1/2}\gamma_{n}),\\
o_{P}\left\{\left(\|\tilde{\bbeta}_{n1}-\bbeta_{01}\|+n^{-1/2}\right)^{2}\right\} & =o_{P}(\gamma_{n}^{2}).
\end{align*}
Since $\mathcal{I}_{11}$ does not change when $n$ and $p$ increase,
\[
\frac{1}{2}(\tilde{\bbeta}_{n1}-\bbeta_{01})^{\top}\mathcal{I}_{11}(\tilde{\bbeta}_{n1}-\bbeta_{01})\ge\frac{1}{2}\lambda_{\min}(\mathcal{I}_{11})\gamma_{n}^{2},
\]
where $\lambda_{\min}(\mathcal{I}_{11})>0$. By similar arguments to the proof of Theorem 4.2 in \cite{bradic2011regularization},
\[
|\bone^{\top}p_{\theta_{n},\alpha}(\bbeta_{01}+\gamma_{n}\boldsymbol{u},\bzero)-\bone^{\top}p_{\theta_{n},\alpha}(\bbeta_{01},\bzero)|\le \sqrt{s}\theta_{n}\gamma_{n}\rho^{\prime}( d_{n}^{\ast}).
\]
Combining these results leads to
\begin{equation*}
   \mathcal{C}(\bbeta_{01}+\gamma_{n}\boldsymbol{u},\bzero)-\mathcal{C}(\bbeta_{01},\bzero)\le \gamma_{n}\left[O_{P}\left(n^{-1/2}+\theta_{n}\rho^{\prime}( d_{n}^{\ast})\right)-\frac{1}{2}\lambda_{\min}(\mathcal{I}_{11})\gamma_{n}\left(1+o_{P}(1)\right)\right].
\end{equation*}
With probability tending to one, the right-hand side of the above inequality is smaller than zero when $\gamma_{n}=B\left(n^{-1/2}+\theta_{n}\rho^{\prime}( d_{n}^{\ast})\right)$ for a sufficiently large $B$. This completes the proof.
\end{proof}

\begin{lemma}  \label{lemma.oracle}
Under Conditions 1-6 and assuming that $\theta_{n}>>n^{\max\{-1/3,(\delta-1)/2\}}$ and $\theta_{n}\rho^{\prime}( d_{n}^{\ast})=O(n^{-1/3})$, with probability tending to one, there exists a local maximizer $\hat{\bbeta}$ of $\mathcal{C}(\bbeta)$ over $\mathbb{R}^p$ such that $\hat{\bbeta}=\hat{\bbeta}^{o}$, where $\hat{\bbeta}^{o}$ is the penalized oracle estimator in Lemma \ref{lemma.eb}.
\end{lemma}

\begin{proof}
We follow the proof of Theorem 4.3 in \cite{bradic2011regularization}. We first note that under Condition 1, an estimate $\hat{\bbeta}\in\mathbb{R}^{p}$ with $\mbox{supp}(\hat{\bbeta})=\mathcal{M}^{\ast}$ is a strict local maximizer of the nonconcave penalized log profile likelihood $\mathcal{C}(\bbeta)$ if
\begin{align}
\dot{l}_{pn1}(\hat{\bbeta})-\theta_{n}\rho^{\prime}(|\hat{\bbeta}_{1}|)\circ\text{sgn}(\hat{\bbeta}_{1}) & =0,\label{eq.lemma.max.1}\\
\|\dot{l}_{pn2}(\hat{\bbeta})\|_{\infty} & <\theta_{n}\rho^{\prime}(0+),\\
\lambda_{\min}\left(-\ddot{\ell}_{n11}(\hat{\bbeta})\right) & >\theta_{n}\kappa(\rho,\hat{\bbeta}_{1}),\label{eq.lemma.max.3}
\end{align}
where $\circ$ is the Hadamard product. This result comes from similar arguments to the proof of Theorem 2.1 in \cite{bradic2011regularization} and the proof of Theorem 1 in \cite{fan2011nonconcave}.

Since equation (\ref{eq.lemma.max.1}) is  satisfied by $\hat{\bbeta}^{o}$ in Lemma \ref{lemma.eb}, it suffices to show that with probability tending to one,
\begin{align}
\|\dot{l}_{pn2}(\hat{\bbeta}^{o})\|_{\infty} & <\theta_{n}\rho^{\prime}(0+) \label{eq.thm.oracle.1} \\
\lambda_{\min}\left(-\ddot{l}_{pn11}(\hat{\bbeta}^{o})\right) & >\theta_{n}\kappa(\rho,\hat{\bbeta}_{1}^{o}). \label{eq.thm.oracle.2}
\end{align}

To check that (\ref{eq.thm.oracle.2}) holds with probability tending to one, first note that with probability tending to one, $\hat{\bbeta}_{1}^{o}\in\mathcal{N}_{0}$ as $n\to\infty$ by Condition 6(i) and Lemma \ref{lemma.eb}. Thus, $\theta_{n}\kappa(\rho,\hat{\bbeta}_{1}^{o})\le\theta_{n}\kappa_{0}<\lambda_{\min}(\mathcal{I}_{11})$ with probability tending to one. By Lemma \ref{lemma.Hessian} below and similar arguments to the proof of Theorem 4.3 in \cite{bradic2011regularization},
\[
\lambda_{\min}\left(-\ddot{l}_{pn11}(\hat{\bbeta}^{o})\right)=\lambda_{\min}(\mathcal{I}_{11})+o_{P}(1).
\]
So with probability tending to one,
\[
\lambda_{\min}\left(-\ddot{l}_{pn11}(\hat{\bbeta}^{o})\right)>\theta_{n}\kappa(\rho,\hat{\bbeta}_{1}^{o}).
\]

To check (\ref{eq.thm.oracle.1}), we first let
\[
q(t,\bZ_{i};\bbeta,\Lambda)=\Lambda(t)\exp\left(\bbeta^{\top}\bZ_{i}\right)\quad \mbox{and} \quad Q(t,\bZ_{i};\bbeta,\Lambda)=\exp\left\{ -q(t,\bZ_{i};\bbeta,\Lambda)\right\} .
\]
Then for the $i$th subject, its log-likelihood is
\[
l(\bZ_{i},L_{i},R_{i};\bbeta,\Lambda)=\log\left\{ Q(L_{i},\bZ_{i};\bbeta,\Lambda)-Q(R_{i},\bZ_{i};\bbeta,\Lambda)\right\}
\]
and the score function with respect to $\bbeta$ is
\begin{align*}
l_{\bbeta}(\bZ_{i},L_{i},R_{i};\bbeta,\Lambda) & =a_{i}(\bbeta,\Lambda)\bZ_{i},
\end{align*}
where
\[
a_{i}(\bbeta,\Lambda)=\frac{Q(R_{i},\bZ_{i};\bbeta,\Lambda)\cdot q(R_{i},\bZ_{i};\bbeta,\Lambda)-Q(L_{i},\bZ_{i};\bbeta,\Lambda)\cdot q(L_{i},\bZ_{i};\bbeta,\Lambda)}{Q(L_{i},\bZ_{i};\bbeta,\Lambda)-Q(R_{i},\bZ_{i};\bbeta,\Lambda)}.
\]
In addition,
\begin{align*}
\dot{l}_{pn}(\bbeta) & =\frac{1}{n}\sum_{i=1}^{n}l_{\bbeta}(\bZ_{i},L_{i},R_{i};\bbeta,\hat{\Lambda}_{\bbeta})=\frac{1}{n}\sum_{i=1}^{n}a_{i}(\bbeta,\hat{\Lambda}_{\bbeta})\bZ_{i},
\end{align*}
where $\hat{\Lambda}_{\bbeta}=\argmax_{\Lambda}l_{n}(\bbeta,\Lambda)$. We  write $l_{\bbeta}(\bbeta,\Lambda)$ and $l_{\bbeta}(\bZ_{i},L_{i},R_{i};\bbeta,\Lambda)$ as
\[
l_{\bbeta}(\bbeta,\Lambda)=\begin{pmatrix}l_{\bbeta1}(\bbeta,\Lambda)\\
l_{\bbeta2}(\bbeta,\Lambda)
\end{pmatrix}\quad \mbox{and} \quad l_{\bbeta}(\bZ_{i},L_{i},R_{i};\bbeta,\Lambda)=\begin{pmatrix}l_{\bbeta1}(\bZ_{i},L_{i},R_{i};\bbeta,\Lambda)\\
l_{\bbeta2}(\bZ_{i},L_{i},R_{i};\bbeta,\Lambda)
\end{pmatrix}
\]
such that both $l_{\bbeta1}(\bbeta,\Lambda)$ and $l_{\bbeta1}(\bZ_{i},L_{i},R_{i};\bbeta,\Lambda)$ are $s\times1$ vectors. Noting that
\begin{align*}
\|\dot{l}_{pn2}(\hat{\bbeta}^{o})\|_{\infty} & =\|\frac{1}{n}\sum_{i=1}^{n}l_{\bbeta2}(\bZ_{i},L_{i},R_{i};\hat{\bbeta}^{o},\hat{\Lambda}_{\hat{\bbeta}^{o}})\|_{\infty}\\
 & \le\|\frac{1}{n}\sum_{i=1}^{n}l_{\bbeta2}(\bZ_{i},L_{i},R_{i};\bbeta_{0},\Lambda_{0})\|_{\infty}\\
 & \quad \quad+\|\frac{1}{n}\sum_{i=1}^{n}l_{\bbeta2}(\bZ_{i},L_{i},R_{i};\hat{\bbeta}^{o},\hat{\Lambda}_{\hat{\bbeta}^{o}})-\sum_{i=1}^{n}l_{\bbeta2}(\bZ_{i},L_{i},R_{i};\bbeta_{0},\Lambda_{0})\|_{\infty},
\end{align*}
we bound the two terms on the right side separately.

To bound the first term, we derive concentration inequalities for
\[
\frac{1}{n}\sum_{i=1}^{n}a_{i}(\bbeta_0,\Lambda_0)Z_{ij},\quad j=s+1,\dots,p.
\]
Under Conditions 2--4, there exists a positive constant $A$ such that $|a_{i}(\bbeta_0,\Lambda_0)|\le A$ for all $i=1,\dots,n$, where $A$ only depends on $(C_{\bbeta},C_{\bZ},C_{\Lambda},\eta)$. Thus, $|a_{i}(\bbeta_0,\Lambda_0)Z_{ij}|\le AC_{\bZ}$ almost surely. By Hoeffding's inequality for bounded independent random variables, for all $t>0$,
\[
\text{P}\left(\frac{1}{n}\sum_{i=1}^{n}a_{i}(\bbeta_0,\Lambda_0)Z_{ij}>n^{-1/2}t\right)\le2\exp\left(-\frac{2(n^{-1/2}t)^{2}}{4\sum_{i=1}^{n}(AC_{\bZ}/n)^{2}}\right)=2\exp\left(-\frac{t^{2}}{2A^{2}C_{\bZ}^{2}}\right).
\]
Let $u_{n}=n^{(\delta+1)/4}\sqrt{\theta_{n}}>>\sqrt{\log p}$ and define
\[
\Omega_{n}=\left\{ \|\frac{1}{n}\sum_{i=1}^{n}l_{\bbeta2}(\bZ_{i},L_{i},R_{i};\bbeta_0,\Lambda_0)\|_{\infty}\le n^{-1/2}u_{n}\right\} .
\]
By Bonferroni's inequality,
\[
\text{P}(\Omega_{n})\ge1-\sum_{j=s+1}^{p}P\left(\frac{1}{n}\sum_{i=1}^{n}a_{i}(\bbeta_0,\Lambda_0)Z_{ij}>n^{-1/2}u_{n}\right)\ge1-2(p-s)\exp\left(-\frac{u_{n}^{2}}{4A^{2}C_{\bZ}^{2}}\right)\longrightarrow1,
\]
and on $\Omega_{n}$, we have
\[
\theta_{n}^{-1}\|\frac{1}{n}\sum_{i=1}^{n}l_{\bbeta2}(\bZ_{i},L_{i},R_{i};\bbeta_0,\Lambda_0)\|_{\infty}\le\theta_{n}^{-1}\cdot n^{-1/2}\cdot n^{(\delta+1)/4}\sqrt{\theta_{n}}=\frac{n^{(\delta-1)/4}}{\sqrt{\theta_{n}}}\longrightarrow0.
\]

For the second term, we aim to show that
\[
\|\frac{1}{n}\sum_{i=1}^{n}l_{\bbeta2}(\bZ_{i},L_{i},R_{i};\hat{\bbeta}^{o},\hat{\Lambda}_{\hat{\bbeta}^{o}})-\sum_{i=1}^{n}l_{\bbeta2}(\bZ_{i},L_{i},R_{i};\bbeta_0,\Lambda_0)\|_{\infty}=O_{P}(n^{-1/3}),
\]
since $\theta_n \gg n^{-1/3}$. Note that
\begin{align*}
\frac{1}{n}\sum_{i=1}^{n}l_{\bbeta2}(\bZ_{i},L_{i},R_{i};\hat{\bbeta}^{o},\hat{\Lambda}_{\hat{\bbeta}^{o}}) & -\sum_{i=1}^{n}l_{\bbeta2}(\bZ_{i},L_{i},R_{i};\bbeta_{0},\Lambda_{0})\\
 & =\frac{1}{n}\begin{pmatrix}Z_{1,s+1} & Z_{2,s+1} & \dots & Z_{n,s+1}\\
Z_{1,s+2} & Z_{2,s+2} & \dots & Z_{n,s+2}\\
\dots & \dots & \dots & \dots\\
Z_{1,p} & Z_{2,p} & \dots & Z_{n,p}
\end{pmatrix}\begin{pmatrix}a_{1}(\hat{\bbeta}^{o},\hat{\Lambda}_{\hat{\bbeta}^{o}})-a_{1}(\bbeta_{0},\Lambda_{0})\\
a_{2}(\hat{\bbeta}^{o},\hat{\Lambda}_{\hat{\bbeta}^{o}})-a_{2}(\bbeta_{0},\Lambda_{0})\\
\dots\\
a_{n}(\hat{\bbeta}^{o},\hat{\Lambda}_{\hat{\bbeta}^{o}})-a_{n}(\bbeta_{0},\Lambda_{0})
\end{pmatrix},
\end{align*}
where $Z_{i,j}$ refers to the $j$th component of $\bZ_{i}$ and $|Z_{i,j}|\le C_{\bZ}$. So it suffices to show that
\begin{equation} \label{eq.S6}
 \frac{1}{n}\sum_{i=1}^{n}|a_{i}(\hat{\bbeta}^{o},\hat{\Lambda}_{\hat{\bbeta}^{o}})-a_{i}(\bbeta_0,\Lambda_0)|=O_{P}(n^{-1/3}).
\end{equation}

As shown in Lemma \ref{lemma.eb} and Lemma \ref{lemma.Hessian}, $(\hat{\bbeta}^{o},\hat{\Lambda}_{\hat{\bbeta}^{o}})$ is consistent for $(\bbeta_0,\Lambda_0)$. Under Conditions 1--4, $q(L_{i},\bZ_{i};\hat{\bbeta}^{o},\hat{\Lambda}_{\hat{\bbeta}^{o}}),q(R_{i},\bZ_{i};\hat{\bbeta}^{o},\hat{\Lambda}_{\hat{\bbeta}^{o}}),q(L_{i},\bZ_{i};\bbeta_0,\Lambda_0),q(R_{i},\bZ_{i};\bbeta_0,\Lambda_0)$ are bounded below from $0$ and bounded above uniformly for all $i=1,\dots,n$, and the denominators of $a_{i}(\hat{\bbeta}^{o},\hat{\Lambda}_{\hat{\bbeta}^{o}})$ and $a_{i}(\bbeta_0,\Lambda_0)$ are bounded below from $0$ uniformly for all $i=1,\dots,n$. So there exists a constant $C$ such that for all $i=1,\dots,n$,
\begin{align*}
|a_{i}(\hat{\bbeta}^{o},\hat{\Lambda}_{\hat{\bbeta}^{o}})-a_{i}(\bbeta_0,\Lambda_0)| \le & C|q(L_{i},\bZ_{i};\hat{\bbeta}^{o},\hat{\Lambda}_{\hat{\bbeta}^{o}})-q(L_{i},\bZ_{i};\bbeta_0,\Lambda_0)| \\
&~+C|q(R_{i},\bZ_{i};\hat{\bbeta}^{o},\hat{\Lambda}_{\hat{\bbeta}^{o}})-q(R_{i},\bZ_{i};\bbeta_0,\Lambda_0)|.
\end{align*}
Let
\begin{align*}
B(\bbeta,\Lambda;L_{i},R_{i},\bZ_{i})=&~|q(L_{i},\bZ_{i};\bbeta,\Lambda)-q(L_{i},\bZ_{i};\bbeta_0,\Lambda_0)| \\
&~+|q(R_{i},\bZ_{i};\bbeta,\Lambda)-q(R_{i},\bZ_{i};\bbeta_0,\Lambda_0)|.
\end{align*}
To show (\ref{eq.S6}), it suffices to show that
\[
\mathbb{P}_{n}B(\hat{\bbeta}^{o},\hat{\Lambda}_{\hat{\bbeta}^{o}})=O_{P}(n^{-1/3}).
\]
Now note that since $\theta_{n}\rho^{\prime}( d_{n}^{\ast})=O(n^{-1/3})$,
\begin{align*}
\mathbb{P}_{n}l(\hat{\bbeta}^{o},\hat{\Lambda}_{\hat{\bbeta}^{o}})-\mathbb{P}_{n}l(\bbeta_{0},\Lambda_{0}) & \ge \bone^{\top}p_{\theta_{n},\alpha}(\hat{\bbeta}^{o})-\bone^{\top}p_{\theta_{n},\alpha}(\bbeta_{0})\\
 & =O_{P}\left[\left(n^{-1/2}+\theta_{n}\rho^{\prime}( d_{n}^{\ast})\right)\theta_{n}\rho^{\prime}( d_{n}^{\ast})\right]=O_{P}(n^{-2/3}).
\end{align*}
By similar arguments to the proof of Lemma A1 in \cite{zeng2016maximum}, the Hellinger distance between $(\hat{\bbeta}^{o},\hat{\Lambda}_{\hat{\bbeta}^{o}})$ and $(\bbeta_0,\Lambda_0)$ is $O_{P}(n^{-1/3})$. Then applying the mean value theorem, we have
\[
E\left[\sum_{j=1}^{K}\left\{ q(V_{j},\bZ;\hat{\bbeta}^{o},\hat{\Lambda}_{\hat{\bbeta}^{o}})-q(V_{j},\bZ;\bbeta_0,\Lambda_0)\right\} ^{2}\right]=O_{P}(n^{-2/3}),
\]
where we use $\bZ$ and $(V_{1},\dots,V_{K})$ to denote the covariates and inspection times for a general subject. This implies
\[
PB(\hat{\bbeta}^{o},\hat{\Lambda}_{\hat{\bbeta}^{o}})=O_{P}(n^{-1/3}).
\]
Since $q(L_{i},\bZ_{i};\hat{\bbeta}^{o},\hat{\Lambda}_{\hat{\bbeta}^{o}})$ and $q(R_{i},\bZ_{i};\hat{\bbeta}^{o},\hat{\Lambda}_{\hat{\bbeta}^{o}})$ belong to a Donsker class as shown in proof of Theorem 1 of \cite{zeng2016maximum}, $B(\hat{\bbeta}^{o},\hat{\Lambda}_{\hat{\bbeta}^{o}})$ also belongs to a Donsker class. By the asymptotic continuity of empirical process indexed by a Donsker class of functions \citep[Lemma 19.24,][]{van2000asymptotic}, $\sqrt{n}(\mathbb{P}_{n}-P)B(\hat{\bbeta}^{o},\hat{\Lambda}_{\hat{\bbeta}^{o}})=o_{P}(1)$. So $\mathbb{P}_{n}B(\hat{\bbeta}^{o},\hat{\Lambda}_{\hat{\bbeta}^{o}})=O_{P}(n^{-1/3})$.
\end{proof}

\begin{lemma} \label{lemma.Hessian}
Under Conditions 1--5 and assuming that $\theta_{n}\rho^{\prime}( d_{n}^{\ast})=o(1)$, $\sup_{t\in[\zeta,\tau]}|\hat{\Lambda}_{\hat{\bbeta}^{o}}(t)-\Lambda_0(t)|\to0$ almost surely and $-\ddot{l}_{pn11}(\hat{\bbeta}^{o})$ is consistent for $\mathcal{I}_{11}$, where $\hat{\bbeta}^{o}$ is a local maximizer of $\mathcal{C}(\beta)$ given by Lemma \ref{lemma.eb} and $\hat{\Lambda}_{\bbeta}=\argmax_{\Lambda}l_{n}(\bbeta,\Lambda)$.
\end{lemma}

\begin{proof}
Let $\dot{\Lambda}_{\tilde{\bbeta}}=\partial(d\hat{\Lambda}_{\bbeta})/\partial\bbeta_1\big|_{\bbeta=\tilde{\bbeta}} (d\hat{\Lambda}_{\tilde{\bbeta}})^{-1} $. We only consider the derivative of $d\hat{\Lambda}_{\bbeta}$ with respect to $\bbeta_{1}$ since  the $j$th component of $\hat{\bbeta}^{o}$ is always $0$ for those $j\notin\mathcal{M}_{\ast}$. By Lemma \ref{lemma.eb}, the assumption that $\theta_{n}\rho^{\prime}( d_{n}^{\ast})=o(1)$, and Condition 5, $\hat{\bbeta}^{o}$ is consistent for $\bbeta_0$ and
\begin{align*}
\mathbb{P}_{n}l(\hat{\bbeta}^{o},\hat{\Lambda}_{\hat{\bbeta}^{o}})-\mathbb{P}_{n}l(\bbeta_0,\Lambda_0)&\ge \bone^{\top}p_{\theta_{n},\alpha}(\hat{\bbeta}^{o})-\bone^{\top}p_{\theta_{n},\alpha}(\bbeta_{0}) \\
&=O_{P}\left[\left(n^{-1/2}+\theta_{n}\rho^{\prime}( d_{n}^{\ast})\right)\theta_{n}\rho^{\prime}( d_{n}^{\ast})\right]=o_{P}(1).
\end{align*}
Then by similar arguments to the proofs of Theorem 1 and Theorem 3 in \cite{zeng2017maximum} respectively, $\sup_{t\in[\zeta,\tau]}|\hat{\Lambda}_{\hat{\bbeta}^{o}}(t)-\Lambda_0(t)|\to0$ almost surely and $\|\dot{\Lambda}_{\hat{\bbeta}^{o}}+\boldsymbol{h}_{01}\|_{L_{2}(\Lambda_0)}=o_{P}(1)$, where $\|\boldsymbol{h}\|_{L_2(\Lambda_0)}^2=\sum_{j=1}^s\int_{\zeta}^{\tau}h_j^{2}(t)d\Lambda_0(t)$.

The following proof of the consistency of $-\ddot{l}_{pn11}(\hat{\bbeta}^{o})$ is similar to the proof of Theorem 4 in \cite{zhang2023post}.

By the definitions of $l_{pn}(\bbeta)$ and $\hat{\Lambda}_{\hat{\bbeta}^{o}}$,
\[
\dot{l}_{pn1}(\hat{\bbeta}^{o})=\mathbb{P}_{n}l_{\bbeta,1}(\hat{\bbeta}^{o},\hat{\Lambda}_{\hat{\bbeta}^{o}})+\mathbb{P}_{n}l_{\Lambda}(\hat{\bbeta}^{o},\hat{\Lambda}_{\hat{\bbeta}^{o}})(\dot{\Lambda}_{\hat{\bbeta}^{o}})
\]
and
\[
\mathbb{P}_{n}l_{\Lambda}(\hat{\bbeta}^{o},\hat{\Lambda}_{\hat{\bbeta}^{o}})(\boldsymbol{h}_{1})=0
\]
for any $\boldsymbol{h}_1$. So
\[
\ddot{l}_{pn11}(\hat{\bbeta}^{o})=\mathbb{P}_{n}\left\{ l_{\bbeta\bbeta,1}(\hat{\bbeta}^{o},\hat{\Lambda}_{\hat{\bbeta}^{o}})+l_{\bbeta\Lambda,1}(\hat{\bbeta}^{o},\hat{\Lambda}_{\hat{\bbeta}^{o}})(\dot{\Lambda}_{\hat{\bbeta}^{o}})\right\} ,
\]
where $l_{\bbeta\bbeta,1}$ is the
derivative of $l_{\bbeta,1}$ with respect to $\bbeta_{1}$, and $l_{\bbeta\Lambda,1}$ is the derivative of $l_{\bbeta,1}$ along the submodel $d\Lambda_{\epsilon,\boldsymbol{h}}=(1+\epsilon \boldsymbol{h})d\Lambda$. Due to the boundedness of $\hat{\bbeta}^{o}$ and $\hat{\Lambda}_{\hat{\bbeta}^{o}}$, it is easy to show that $l_{\bbeta\bbeta,1}(\hat{\bbeta}^{o},\hat{\Lambda}_{\hat{\bbeta}^{o}})$
and $l_{\bbeta\Lambda,1}(\hat{\bbeta}^{o},\hat{\Lambda}_{\hat{\bbeta}^{o}})(\dot{\Lambda}_{\hat{\bbeta}^{o}})$ belong to two Glivenko-Cantelli classes respectively by similar arguments to the proof of Lemma 1 in \cite{zeng2017maximum} (see also remark A1 in \cite{zeng2016maximum}). Thus,
\begin{align*}
-\ddot{l}_{pn11}(\hat{\bbeta}^{o}) & =-P\left\{ l_{\bbeta\bbeta,1}(\hat{\bbeta}^{o},\hat{\Lambda}_{\hat{\bbeta}^{o}})+l_{\bbeta\Lambda,1}(\hat{\bbeta}^{o},\hat{\Lambda}_{\hat{\bbeta}^{o}})(\dot{\Lambda}_{\hat{\bbeta}^{o}})\right\} +o_{P}(1)\\
 & =-P\left\{ l_{\bbeta\bbeta,1}(\hat{\bbeta}^{o},\hat{\Lambda}_{\hat{\bbeta}^{o}})-l_{\bbeta\Lambda,1}(\hat{\bbeta}^{o},\hat{\Lambda}_{\hat{\bbeta}^{o}})(\boldsymbol{h}_{01})\right\} +o_{P}(1)\\
 & =P\left[\left\{ l_{\bbeta,1}(\bbeta_0,\Lambda_0)-l_{\Lambda}(\bbeta_0,\Lambda_0)(\boldsymbol{h}_{01})\right\} ^{\otimes2}\right]+o_{P}(1)\\
 & =\mathcal{I}_{11}+o_{P}(1).
\end{align*}
Here we have used
\[
P\left\{ l_{\bbeta,1}(\bbeta_0,\Lambda_0)l_{\Lambda}(\bbeta_0,\Lambda_0)(\boldsymbol{h}_{01})^{\top}\right\} =P\left[\left\{ l_{\Lambda}(\bbeta_0,\Lambda_0)(\boldsymbol{h}_{01})\right\} ^{\otimes2}\right]
\]
from the proof of Theorem 2 in \cite{zeng2016maximum}.
\end{proof}

\subsection{Proof of Theorem 2}
We first show that
\begin{equation}
\mathbb{G}_{n}\tilde{l}_{1}=\sqrt{n}\theta_{n}\rho^{\prime}(|\hat{\bbeta}_{1}^{o}|)\circ\text{sgn}(\hat{\bbeta}_{1}^{o})+\sqrt{n}\mathcal{I}_{11}(\hat{\bbeta}_{1}^{o}-\bbeta_{01})+o_{P}(1) \label{lemma.oracle.2}
\end{equation}
by following the proof of Theorem 2 in \cite{zeng2016maximum} and the proof of Lemma 2 in \cite{zhang2023post}. By similar arguments to the proof of Lemma 2 of \cite{zeng2017maximum}, $l_{\bbeta,1}(\hat{\bbeta}^{o},\hat{\Lambda}_{\hat{\bbeta}^{o}})$ and $l_{\Lambda}(\hat{\bbeta}^{o},\hat{\Lambda}_{\hat{\bbeta}^{o}})(\dot{\Lambda}_{\hat{\bbeta}^{o}})$ belong to two Donsker classes respectively.

According to Lemmas \ref{lemma.eb} and \ref{lemma.Hessian}, $(\hat{\bbeta}^{o},\hat{\Lambda}_{\hat{\bbeta}^{o}})$ is consistent for $(\bbeta_0,\Lambda_0)$. In the proof of Lemma \ref{lemma.Hessian}, we have also shown $\|\dot{\Lambda}_{\hat{\bbeta}^{o}}+\boldsymbol{h}_{01}\|_{L_{2}(\Lambda_0)}=o_{P}(1)$. Then, by the asymptotic continuity of empirical process indexed by a Donsker class of functions \citep[Lemma 19.24,][]{van2000asymptotic},
\[
\mathbb{G}_{n}l_{\Lambda}(\hat{\bbeta}^{o},\hat{\Lambda}_{\hat{\bbeta}^{o}})(\dot{\Lambda}_{\hat{\bbeta}^{o}})=\mathbb{G}_{n}l_{\Lambda}(\bbeta_0,\Lambda_0)(-\boldsymbol{h}_{01})+o_{P}(1)=\mathbb{G}_{n}l_{\Lambda}(\hat{\bbeta}^{o},\hat{\Lambda}_{\hat{\bbeta}^{o}})(-\boldsymbol{h}_{01})+o_{P}(1),
\]
Also we have
\begin{align*}
\mathbb{G}_{n}l_{\bbeta,1}(\hat{\bbeta}^{o},\hat{\Lambda}_{\hat{\bbeta}^{o}}) & =n^{1/2}\mathbb{P}_{n}l_{\bbeta,1}(\hat{\bbeta}^{o},\hat{\Lambda}_{\hat{\bbeta}^{o}})-n^{1/2}\left\{ Pl_{\bbeta,1}(\hat{\bbeta}^{o},\hat{\Lambda}_{\hat{\bbeta}^{o}})-Pl_{\bbeta,1}(\bbeta_0,\Lambda_0)\right\} ,\\
\mathbb{G}_{n}l_{\Lambda}(\hat{\bbeta}^{o},\hat{\Lambda}_{\hat{\bbeta}^{o}})(-\boldsymbol{h}_{01}) & =-n^{1/2}\left\{ Pl_{\Lambda}(\hat{\bbeta}^{o},\hat{\Lambda}_{\hat{\bbeta}^{o}})(-\boldsymbol{h}_{01})-Pl_{\Lambda}(\bbeta_0,\Lambda_0)(-\boldsymbol{h}_{01})\right\} .
\end{align*}
Since $\theta_{n}\rho^{\prime}( d_{n}^{\ast})=O(n^{-1/3})$,
\begin{align*}
\mathbb{P}_{n}l(\hat{\bbeta}^{o},\hat{\Lambda}_{\hat{\bbeta}^{o}})-\mathbb{P}_{n}l(\bbeta_{0},\Lambda_{0}) & \ge \bone^{\top}p_{\theta_{n},\alpha}(\hat{\bbeta}^{o})-\bone^{\top}p_{\theta_{n},\alpha}(\bbeta_{0})\\
 & =O_{P}\left[\left(n^{-1/2}+\theta_{n}\rho^{\prime}( d_{n}^{\ast})\right)\theta_{n}\rho^{\prime}( d_{n}^{\ast})\right]=O_{P}(n^{-2/3}).
\end{align*}
Then, by the proof of Theorem 2 in \cite{zeng2016maximum}, the second-order terms of the Taylor expansions of $-n^{1/2}\left\{ Pl_{\bbeta,1}(\hat{\bbeta}^{o},\hat{\Lambda}_{\hat{\bbeta}^{o}})-Pl_{\bbeta,1}(\bbeta_0,\Lambda_0)\right\}$ and\\ $-n^{1/2}\left\{ Pl_{\Lambda}(\hat{\bbeta}^{o},\hat{\Lambda}_{\hat{\bbeta}^{o}})(-\boldsymbol{h}_{01})-Pl_{\Lambda}(\bbeta_0,\Lambda_0)(-\boldsymbol{h}_{01})\right\}$ are both $o_{P}(1)$. So
\begin{align*}
&~\mathbb{G}_{n}l_{\bbeta,1}(\hat{\bbeta}^{o},\hat{\Lambda}_{\hat{\bbeta}^{o}}) \\
=&~ \sqrt{n}\theta_{n}\rho^{\prime}(\hat{\bbeta}_{1}^{o})\circ\text{sgn}(\hat{\bbeta}_{1}^{o})-n^{1/2}\left[Pl_{\beta\beta,1}(\hat{\bbeta}^{o}-\bbeta_0)+Pl_{\beta\Lambda,1}\left\{ d(\hat{\Lambda}_{\hat{\bbeta}^{o}}-\Lambda_0)/d\Lambda_0\right\} \right] +o_{P}(1),\\
&~\mathbb{G}_{n}l_{\Lambda} (\hat{\bbeta}^{o},\hat{\Lambda}_{\hat{\bbeta}^{o}})(\hat{\Lambda}_{\hat{\bbeta}^{o}}) \\
=&~-n^{1/2}\left[Pl_{\Lambda\bbeta,1}(-\boldsymbol{h}_{01})(\hat{\bbeta}^{o}-\bbeta_0)+Pl_{\Lambda\Lambda}\left\{ -\boldsymbol{h}_{01},d(\hat{\Lambda}_{\hat{\bbeta}^{o}}-\Lambda_0)/d\Lambda_0\right\} \right] +o_{P}(1),
\end{align*}
where $l_{\Lambda\bbeta,1}(\boldsymbol{h}_1)$ is the derivative of $l_{\Lambda}(\boldsymbol{h}_1)$ with respect to $\bbeta_{1}$, and $l_{\Lambda\Lambda}\{\boldsymbol{h}_1,d(\hat{\Lambda}_{\hat{\bbeta}^{o}}-\Lambda_0)/d\Lambda_0\}$ is the derivative of $\ell_{\Lambda}(\boldsymbol{h}_1)$ along the submodel $\Lambda_{\epsilon}=\Lambda_0+\epsilon(\hat{\Lambda}_{\hat{\bbeta}^{o}}-\Lambda_0)$. By similar arguments to the proof of Theorem 2 of \cite{zeng2016maximum}, we have
\[
\mathbb{G}_{n}\left[l_{\bbeta,1}(\hat{\bbeta}^{o},\hat{\Lambda}_{\hat{\bbeta}^{o}})+l_{\Lambda}(\hat{\bbeta}^{o},\hat{\Lambda}_{\hat{\bbeta}^{o}})(\dot{\Lambda}_{\hat{\bbeta}^{o}})\right]=\sqrt{n}\theta_{n}\rho^{\prime}(\hat{\bbeta}_{1}^{o})\circ\text{sgn}(\hat{\bbeta}_{1}^{o})+\sqrt{n}\mathcal{I}_{11}(\hat{\bbeta}_{1}^{o}-\bbeta_{01})+o_{P}(1),
\]
By the asymptotic continuity,
\[
\mathbb{G}_{n}\left[l_{\bbeta,1}(\hat{\bbeta}^{o},\hat{\Lambda}_{\hat{\bbeta}^{o}})+l_{\Lambda}(\hat{\bbeta}^{o},\hat{\Lambda}_{\hat{\bbeta}^{o}})(\hat{\Lambda}_{\hat{\bbeta}^{o}})\right]=\mathbb{G}_{n}\tilde{l}_{1}+o_{P}(1).
\]
This completes the proof of equation (\ref{lemma.oracle.2}).

Thus,
\[
\sqrt{n}(\hat{\bbeta}_{1}^{o}-\bbeta_{01})=\mathcal{I}_{11}^{-1}\mathbb{G}_{n}\tilde{l}_{1}-\sqrt{n}\theta_{n}\mathcal{I}_{11}^{-1}\rho^{\prime}(|\hat{\bbeta}_{1}^{o}|)\circ\text{sgn}(\hat{\bbeta}_{1}^{o})+o_{P}(1),
\]
where $\circ$ is the Hadamard product. Since $\theta_{n}\rho^{\prime}( d_{n}^{\ast})=o(n^{-1/2})$,
\[
\sqrt{n}(\hat{\bbeta}_{1}^{o}-\bbeta_{01})=\mathcal{I}_{11}^{-1}\mathbb{G}_{n}\tilde{l}_{1}+o_{P}(1).
\]
The result follows from the central limit theorem.

\subsection{Proof of Theorem 3}
Same as in the proof of Theorem 1, we define a penalized oracle estimator $\hat{\bbeta}^{o}=(\hat{\bbeta}_{1}^{o\top},\bzero^{\top})^{\top}$ to be a local maximizer of $\mathcal{C}(\bbeta)$ constrained on the $s$-dimensional subspace $E=\left\{ \bbeta\in\mathbb{R}^{p}:\bbeta_{2}=\bzero\right\} $. For the adaptive lasso, $\rho^{\prime}( \hat{\beta}^{o}_{1j})=O_{P}(1)$ $(j=1,\ldots,s)$ by the condition that the initial estimator $\tilde{\bbeta}$ is zero-consistent with rate $r_n$. So, by similar arguments to the proof of Lemma \ref{lemma.eb}, there exists a penalized oracle estimator $\hat{\bbeta}^{o}$ such that $\|\hat{\bbeta}^{o}-\bbeta_0\|_{2}=O_{P}(n^{-1/2}+\theta_{n})$. Following the proof of Lemma \ref{lemma.oracle} and taking advantage of the convexity of adaptive lasso penalty, $\hat{\bbeta}^{o}$ is a strict local maximizer of $\mathcal{C}(\beta)$ on $\mathbb{R}^p$ as long as
\begin{equation} \label{eq.alasso.1}
|\dot{l}_{pn,j}(\hat{\bbeta}^{o})| <\theta_{n}/|\tilde{\beta}_j|,\quad \text{for all } j\in \mathcal{M}_{\ast}^c
\end{equation}
where $\dot{l}_{pn,j}(\hat{\bbeta}^{o})$ is the $j$th component of $\dot{l}_{pn}(\hat{\bbeta}^{o})$.

Since $\theta_{n}=O(n^{-1/3})$, $\theta_n r_n\gg n^{\max\{-1/3,(\delta-1)/2\}}$ and $r_{n}\max_{j\in\mathcal{M}_{\ast}^{c}}|\tilde{\beta}_{j}|=O_{p}(1)$, by the same arguments as in the proof of Lemma \ref{lemma.oracle} except choosing $u_n$ such that $n^{\delta/2}\ll u_{n}\ll n^{1/2}\theta_n r_{n}$, we can show that
\[
\frac{|\dot{l}_{pn,j}(\hat{\bbeta}^{o})|}{\theta_{n}/|\tilde{\beta}_{j}|}\lesssim\frac{\|\dot{l}_{pn2}(\hat{\bbeta}^{o})\|_{\infty}}{\theta_{n}r_{n}}=o_{p}(1).
\]
(\ref{eq.alasso.1}) then holds with probability tending to one.

\subsection{Proof of Theorem 4}
The result comes from similar arguments to the proof of Theorem 2.

\subsection{Proof of Proposition 1}
The profile log likelihood $l_{pn}(\bbeta)$ has been shown to be concave in the proof of Theorem 1 in \cite{li2020adaptive}. Then the proof of Proposition 1 can follow that of Proposition 3(a) in \cite{fan2011nonconcave}.

\section{Additional simulation results}\label{A1}

\subsection{Estimation results of individual nonzero coefficients}

\begin{table}[H]
\renewcommand{\arraystretch}{0.8}
\centering
\caption{Empirical mean estimates of $\bbeta^{(6)}$ with $p = 3000$ and $\rho = 0$. The values in parentheses are the empirical standard errors.}
\footnotesize
\begin{tabular}{cccccccc}
\toprule
Method & $\beta_1$    & $\beta_2$    & $\beta_3$    & $\beta_4$    & $\beta_5$   & $\beta_6$ \\
\midrule

Truth & -1.40  & -0.83 & -1.64 & 0.69  & 1.39  & 1.65 \vspace{0.2cm} \\
\multicolumn{7}{l}{\textit{(n = 500; p = 3,000)}} \\
    Oracle &  -1.42(0.17) &  -0.84(0.12) &  -1.69(0.18) &   0.70(0.10) &   1.42(0.13) &   1.69(0.19) \\
    Lasso &  -0.65(0.13) &  -0.30(0.10) &  -0.77(0.13) &   0.21(0.09) &   0.69(0.11) &   0.75(0.15) \\
    Adaptive Lasso &  -1.40(0.18) &  -0.81(0.14) &  -1.66(0.18) &   0.65(0.17) &   1.41(0.14) &   1.66(0.20) \\
    MCP   &  -1.43(0.18) &  -0.84(0.12) &  -1.69(0.18) &   0.70(0.11) &   1.43(0.13) &   1.69(0.19) \\
    SCAD  &  -1.38(0.17) &  -0.73(0.25) &  -1.64(0.19) &   0.52(0.28) &   1.38(0.14) &   1.62(0.22) \\
    MCP+mid-point & -16.55(4.04) &  -9.33(2.62) & -19.28(4.32) &   7.66(2.31) &  14.37(2.68) &  16.69(4.00) \\
    SCAD+mid-point & -21.17(4.94) & -12.27(3.69) & -25.00(5.93) &   9.89(3.18) &  18.89(3.82) &  21.43(5.12) \vspace{0.2cm}\\
\multicolumn{7}{l}{\textit{(n = 1000; p = 3,000)}} \\
    Oracle &  -1.40(0.11) &  -0.84(0.08) &  -1.67(0.14) &   0.71(0.07) &   1.41(0.09) &   1.68(0.11) \\
    Lasso &  -0.82(0.08) &  -0.44(0.06) &  -0.98(0.11) &   0.35(0.06) &   0.88(0.08) &   1.01(0.09) \\
    Adaptive Lasso &  -1.39(0.11) &  -0.83(0.08) &  -1.66(0.14) &   0.70(0.07) &   1.41(0.10) &   1.67(0.11) \\
    MCP   &  -1.40(0.11) &  -0.84(0.08) &  -1.67(0.14) &   0.70(0.07) &   1.41(0.09) &   1.68(0.11) \\
    SCAD  &  -1.39(0.11) &  -0.83(0.08) &  -1.66(0.14) &   0.70(0.08) &   1.41(0.09) &   1.67(0.11) \\
    MCP+mid-point & -13.11(2.70) &  -7.54(1.88) & -15.49(3.09) &   5.93(1.60) &  11.57(2.08) &  13.28(2.60) \\
    SCAD+mid-point & -17.58(3.38) & -10.27(2.47) & -20.91(4.12) &   8.19(2.12) &  15.76(3.06) &  18.15(3.78) \\
\bottomrule
\end{tabular}%
\label{tab:addlabel}%
\end{table}%

\begin{table}[H]
\renewcommand{\arraystretch}{0.8}
\centering
\caption{Empirical mean estimates of $\bbeta^{(6)}$ with $p = 10000$ and $\rho = 0$. The values in parentheses are the empirical standard errors.}
\footnotesize
\begin{tabular}{cccccccc}
\toprule
Method & $\beta_1$    & $\beta_2$    & $\beta_3$    & $\beta_4$    & $\beta_5$   & $\beta_6$ \\
\midrule
Truth & -1.40  & -0.83 & -1.64 & 0.69  & 1.39  & 1.65 \vspace{0.2cm} \\
\multicolumn{7}{l}{\textit{(n = 500; p = 10,000)}} \\
    Oracle &  -1.42(0.15) &  -0.84(0.12) &  -1.67(0.19) &   0.70(0.10) &   1.43(0.14) &   1.69(0.17) \\
    Lasso &  -0.59(0.11) &  -0.26(0.10) &  -0.68(0.13) &   0.17(0.09) &   0.64(0.12) &   0.68(0.17) \\
    Adaptive Lasso &  -1.39(0.16) &  -0.81(0.15) &  -1.63(0.20) &   0.62(0.24) &   1.41(0.14) &   1.66(0.17) \\
    MCP   &  -1.42(0.15) &  -0.84(0.13) &  -1.67(0.19) &   0.70(0.10) &   1.43(0.14) &   1.69(0.17) \\
    SCAD  &  -1.34(0.18) &  -0.65(0.31) &  -1.58(0.21) &   0.43(0.30) &   1.36(0.15) &   1.60(0.20) \\
    MCP+mid-point & -15.64(3.67) &  -8.89(2.67) & -18.31(4.08) &   7.30(2.34) &  14.32(3.04) &  16.33(3.53) \\
    SCAD+mid-point & -20.60(4.87) & -11.66(3.52) & -23.80(5.72) &   9.49(2.81) &  18.46(4.12) &  21.24(4.76) \vspace{0.2cm}\\
\multicolumn{7}{l}{\textit{(n = 1000; p = 10,000)}} \\
    Oracle &  -1.42(0.10) &  -0.84(0.07) &  -1.64(0.13) &   0.70(0.07) &   1.40(0.08) &   1.67(0.11) \\
    Lasso &  -0.80(0.08) &  -0.41(0.06) &  -0.91(0.10) &   0.32(0.06) &   0.83(0.07) &   0.95(0.09) \\
    Adaptive Lasso &  -1.42(0.10) &  -0.83(0.07) &  -1.63(0.13) &   0.69(0.07) &   1.40(0.09) &   1.66(0.11) \\
    MCP   &  -1.42(0.11) &  -0.84(0.07) &  -1.64(0.13) &   0.70(0.07) &   1.40(0.08) &   1.67(0.11) \\
    SCAD  &  -1.42(0.10) &  -0.83(0.07) &  -1.63(0.13) &   0.69(0.08) &   1.39(0.08) &   1.66(0.11) \\
    MCP+mid-point & -13.55(2.40) &  -8.01(1.76) & -15.86(2.95) &   6.23(1.35) &  12.09(2.08) &  13.91(2.47) \\
    SCAD+mid-point & -17.85(3.61) & -10.68(2.57) & -20.85(4.43) &   8.24(2.07) &  15.93(3.30) &  18.55(3.91) \\
\bottomrule
\end{tabular}%
\label{tab:addlabel}%
\end{table}%

\begin{table}[H]
\renewcommand{\arraystretch}{0.8}
\centering
\caption{Empirical mean estimates of $\bbeta^{(6)}$ with $p = 3000$ and $\rho = 0.8$. The values in parentheses are the empirical standard errors.}
\footnotesize
\begin{tabular}{cccccccc}
\toprule
Method & $\beta_1$    & $\beta_2$    & $\beta_3$    & $\beta_4$    & $\beta_5$   & $\beta_6$ \\
\midrule
Truth & -1.40  & -0.83 & -1.64 & 0.69  & 1.39  & 1.65 \vspace{0.2cm} \\
\multicolumn{7}{l}{\textit{(n = 500; p = 3,000)}} \\
    Oracle &  -1.41(0.16) &  -0.84(0.11) &  -1.69(0.18) &   0.71(0.11) &   1.41(0.12) &   1.71(0.17) \\
    Lasso &  -0.62(0.12) &  -0.29(0.09) &  -0.76(0.14) &   0.21(0.09) &   0.67(0.11) &   0.77(0.14) \\
    Adaptive Lasso &  -1.37(0.17) &  -0.80(0.13) &  -1.64(0.20) &   0.66(0.16) &   1.38(0.13) &   1.67(0.19) \\
    MCP   &  -1.41(0.16) &  -0.84(0.11) &  -1.69(0.18) &   0.71(0.12) &   1.41(0.12) &   1.71(0.17) \\
    SCAD  &  -1.37(0.18) &  -0.74(0.24) &  -1.64(0.20) &   0.56(0.27) &   1.36(0.13) &   1.66(0.18) \\
    MCP+mid-point & -14.86(3.80) &  -8.69(2.74) & -17.49(4.42) &   6.89(2.42) &  13.32(2.95) &  15.38(3.50) \\
    SCAD+mid-point & -19.10(4.65) & -11.31(3.42) & -23.08(5.50) &   9.10(3.16) &  17.31(3.76) &  20.19(4.87) \vspace{0.2cm}\\
\multicolumn{7}{l}{\textit{(n = 1000; p = 3,000)}} \\
    Oracle &  -1.43(0.11) &  -0.84(0.08) &  -1.66(0.12) &   0.71(0.07) &   1.41(0.10) &   1.68(0.11) \\
    Lasso &  -0.83(0.09) &  -0.43(0.07) &  -0.96(0.09) &   0.35(0.06) &   0.87(0.09) &   1.00(0.10) \\
    Adaptive Lasso &  -1.41(0.11) &  -0.82(0.08) &  -1.64(0.12) &   0.69(0.07) &   1.39(0.10) &   1.65(0.11) \\
    MCP   &  -1.43(0.11) &  -0.84(0.08) &  -1.66(0.12) &   0.71(0.07) &   1.41(0.10) &   1.68(0.11) \\
    SCAD  &  -1.42(0.11) &  -0.84(0.08) &  -1.66(0.11) &   0.70(0.07) &   1.41(0.10) &   1.67(0.11) \\
    MCP+mid-point & -10.01(1.98) &  -5.81(1.55) & -11.67(2.55) &   4.56(1.32) &   8.93(1.83) &  10.13(2.17) \\
    SCAD+mid-point & -14.31(3.03) &  -8.21(2.11) & -16.60(3.52) &   6.49(1.99) &  12.60(2.74) &  14.15(2.85) \\
\bottomrule
\end{tabular}%
\label{tab:addlabel}%
\end{table}%

\begin{table}[H]
\renewcommand{\arraystretch}{0.8}
\centering
\caption{Empirical mean estimates of $\bbeta^{(6)}$ with $p = 10000$ and $\rho = 0.8$. The values in parentheses are the empirical standard errors.}
\footnotesize
\begin{tabular}{cccccccc}
\toprule
Method & $\beta_1$    & $\beta_2$    & $\beta_3$    & $\beta_4$    & $\beta_5$   & $\beta_6$ \\
\midrule
Truth & -1.40  & -0.83 & -1.64 & 0.69  & 1.39  & 1.65 \vspace{0.2cm} \\
\multicolumn{7}{l}{\textit{(n = 500; p = 10,000)}} \\
   Oracle &  -1.44(0.17) &  -0.84(0.12) &  -1.68(0.19) &   0.71(0.11) &   1.43(0.13) &   1.71(0.16) \\
    Lasso &  -0.58(0.12) &  -0.25(0.10) &  -0.68(0.14) &   0.17(0.09) &   0.63(0.11) &   0.68(0.15) \\
    Adaptive Lasso &  -1.40(0.19) &  -0.80(0.14) &  -1.63(0.20) &   0.63(0.22) &   1.40(0.14) &   1.66(0.18) \\
    MCP   &  -1.44(0.17) &  -0.84(0.13) &  -1.68(0.19) &   0.71(0.12) &   1.43(0.13) &   1.71(0.16) \\
    SCAD  &  -1.36(0.20) &  -0.64(0.31) &  -1.59(0.21) &   0.46(0.31) &   1.36(0.14) &   1.62(0.17) \\
    MCP+mid-point & -15.50(3.64) &  -8.86(2.63) & -17.92(4.26) &   7.18(2.10) &  13.78(2.75) &  15.84(3.43) \\
    SCAD+mid-point & -19.92(4.96) & -11.20(3.18) & -23.20(5.57) &   9.00(2.88) &  17.70(3.81) &  20.19(4.59) \vspace{0.2cm}\\
\multicolumn{7}{l}{\textit{(n = 1000; p = 10,000)}} \\
    Oracle &  -1.42(0.11) &  -0.84(0.10) &  -1.65(0.13) &   0.70(0.07) &   1.40(0.09) &   1.68(0.12) \\
    Lasso &  -0.79(0.08) &  -0.41(0.08) &  -0.91(0.10) &   0.32(0.06) &   0.83(0.07) &   0.95(0.10) \\
    Adaptive Lasso &  -1.41(0.11) &  -0.82(0.10) &  -1.64(0.13) &   0.69(0.07) &   1.39(0.10) &   1.66(0.12) \\
    MCP   &  -1.42(0.10) &  -0.84(0.10) &  -1.66(0.13) &   0.70(0.07) &   1.40(0.09) &   1.68(0.12) \\
    SCAD  &  -1.41(0.11) &  -0.83(0.11) &  -1.65(0.13) &   0.69(0.09) &   1.39(0.09) &   1.67(0.12) \\
    MCP+mid-point & -12.80(2.78) &  -7.49(1.82) & -14.90(2.90) &   5.85(1.31) &  11.39(2.05) &  12.83(2.40) \\
    SCAD+mid-point & -16.06(3.01) &  -9.43(2.26) & -19.26(3.85) &   7.48(1.69) &  14.32(2.52) &  16.41(3.13) \\
\bottomrule
\end{tabular}%
\label{tab:addlabel}%
\end{table}%

\begin{table}[H]
\renewcommand{\arraystretch}{0.8}
\centering
\caption{Empirical mean estimates of $\bbeta^{(12)}$ with $p = 3000$, $n = 500$ and $\rho = 0$. The values in parentheses are the empirical standard errors.}
\footnotesize
\begin{tabular}{ccccccc}
\toprule
Method & $\beta_1$ & $\beta_2$ & $\beta_3$ & $\beta_4$ & $\beta_5$ & $\beta_6$ \\
\midrule
Truth & -1.4  & -0.83 & -1.64 & 0.69  & 1.39  & 1.65 \\
    Oracle &  -1.45(0.17) &  -0.84(0.12) &  -1.71(0.21) &   0.72(0.10) &   1.44(0.12) &   1.71(0.18) \\
    Lasso &  -0.48(0.12) &  -0.22(0.08) &  -0.59(0.13) &   0.15(0.09) &   0.56(0.10) &   0.58(0.14) \\
    Adaptive Lasso &  -1.35(0.19) &  -0.76(0.15) &  -1.60(0.22) &   0.59(0.24) &   1.35(0.15) &   1.59(0.23) \\
    MCP   &  -1.44(0.17) &  -0.83(0.12) &  -1.69(0.20) &   0.71(0.11) &   1.42(0.13) &   1.68(0.17) \\
    SCAD  &  -1.36(0.19) &  -0.76(0.20) &  -1.61(0.20) &   0.53(0.29) &   1.34(0.13) &   1.58(0.18) \\
    MCP+mid-point & -13.38(3.30) &  -7.70(2.33) & -15.99(4.01) &   6.36(2.14) &  12.60(2.58) &  14.89(3.38) \\
    SCAD+mid-point & -16.97(4.14) &  -9.86(2.79) & -20.10(4.82) &   8.17(2.67) &  16.16(3.29) &  18.94(3.88) \\
&       &       &       &       &       &  \\
& $\beta_7$ & $\beta_8$ & $\beta_9$ & $\beta_{10}$ & $\beta_{11}$ & $\beta_{12}$ \\
\midrule
Truth & -0.52 & 0.86  & -1.23 & 1.18  & -1.97 & -1.68 \\
    Oracle &  -0.53(0.12) &   0.90(0.12) &  -1.27(0.13) &   1.22(0.13) &  -2.03(0.22) &  -1.71(0.22) \\
    Lasso &  -0.06(0.06) &   0.24(0.10) &  -0.44(0.09) &   0.43(0.10) &  -0.69(0.15) &  -0.54(0.15) \\
    Adaptive Lasso &  -0.30(0.27) &   0.82(0.17) &  -1.19(0.15) &   1.14(0.14) &  -1.90(0.24) &  -1.58(0.28) \\
    MCP   &  -0.46(0.23) &   0.89(0.12) &  -1.26(0.13) &   1.21(0.12) &  -2.02(0.22) &  -1.69(0.23) \\
    SCAD  &  -0.21(0.23) &   0.80(0.22) &  -1.18(0.13) &   1.13(0.13) &  -1.91(0.22) &  -1.59(0.26) \\
    MCP+mid-point &  -4.36(2.56) &   7.94(2.00) & -11.99(2.80) &  10.78(2.49) & -19.43(5.07) & -16.29(4.38) \\
    SCAD+mid-point &  -5.68(3.31) &  10.17(2.56) & -14.86(3.41) &  13.66(2.93) & -24.30(5.65) & -20.56(5.33) \\
\bottomrule
\end{tabular}%
\label{tab:addlabel}%
\end{table}%

\begin{table}[H]
\renewcommand{\arraystretch}{0.8}
\centering
\caption{Empirical mean estimates of $\bbeta^{(12)}$ with $p = 3000$, $n = 1000$ and $\rho = 0$. The values in parentheses are the empirical standard errors.}
\footnotesize
\begin{tabular}{ccccccc}
\toprule
Method & $\beta_1$ & $\beta_2$ & $\beta_3$ & $\beta_4$ & $\beta_5$ & $\beta_6$ \\
\midrule
Truth & -1.4  & -0.83 & -1.64 & 0.69  & 1.39  & 1.65 \\
    Oracle &  -1.40(0.11) &  -0.85(0.08) &  -1.67(0.14) &   0.71(0.08) &   1.41(0.09) &   1.68(0.11) \\
    Lasso &  -0.68(0.08) &  -0.37(0.06) &  -0.83(0.10) &   0.28(0.07) &   0.75(0.07) &   0.85(0.09) \\
    Adaptive Lasso &  -1.37(0.12) &  -0.82(0.09) &  -1.62(0.15) &   0.68(0.09) &   1.37(0.09) &   1.62(0.12) \\
    MCP   &  -1.40(0.11) &  -0.85(0.08) &  -1.66(0.14) &   0.70(0.08) &   1.41(0.09) &   1.67(0.11) \\
    SCAD  &  -1.37(0.11) &  -0.83(0.08) &  -1.64(0.14) &   0.69(0.08) &   1.38(0.09) &   1.63(0.11) \\
    MCP+mid-point & -11.60(2.28) &  -6.62(1.50) & -13.49(2.61) &   5.41(1.26) &  10.81(2.01) &  12.38(2.26) \\
    SCAD+mid-point & -15.16(3.03) &  -8.82(2.27) & -17.75(3.63) &   7.26(1.80) &  14.30(2.60) &  16.41(3.30) \\
&       &       &       &       &       &  \\
& $\beta_7$ & $\beta_8$ & $\beta_9$ & $\beta_{10}$ & $\beta_{11}$ & $\beta_{12}$ \\
\midrule
Truth & -0.52 & 0.86  & -1.23 & 1.18  & -1.97 & -1.68 \\
    Oracle &  -0.52(0.08) &   0.88(0.07) &  -1.24(0.09) &   1.19(0.08) &  -1.98(0.16) &  -1.70(0.14) \\
    Lasso &  -0.16(0.06) &   0.40(0.06) &  -0.62(0.07) &   0.61(0.06) &  -0.97(0.11) &  -0.81(0.10) \\
    Adaptive Lasso &  -0.49(0.10) &   0.85(0.08) &  -1.21(0.09) &   1.15(0.09) &  -1.93(0.17) &  -1.66(0.15) \\
    MCP   &  -0.52(0.08) &   0.88(0.07) &  -1.24(0.09) &   1.18(0.08) &  -1.98(0.16) &  -1.70(0.14) \\
    SCAD  &  -0.43(0.18) &   0.86(0.07) &  -1.22(0.09) &   1.16(0.08) &  -1.94(0.16) &  -1.67(0.14) \\
    MCP+mid-point &  -4.01(1.32) &   6.74(1.32) & -10.19(1.86) &   9.17(1.92) & -16.32(3.08) & -13.74(2.73) \\
    SCAD+mid-point &  -5.50(1.76) &   8.97(2.13) & -13.44(2.50) &  12.10(2.43) & -21.63(4.26) & -18.25(3.70) \\
\bottomrule
\end{tabular}%
\label{tab:addlabel}%
\end{table}%

\begin{table}[H]
\renewcommand{\arraystretch}{0.8}
\centering
\caption{Empirical mean estimates of $\bbeta^{(12)}$ with $p = 10000$, $n = 500$ and $\rho = 0$. The values in parentheses are the empirical standard errors.}
\footnotesize
\begin{tabular}{ccccccc}
\toprule
Method & $\beta_1$ & $\beta_2$ & $\beta_3$ & $\beta_4$ & $\beta_5$ & $\beta_6$ \\
\midrule
Truth & -1.4  & -0.83 & -1.64 & 0.69  & 1.39  & 1.65 \\
    Oracle &  -1.45(0.18) &  -0.86(0.12) &  -1.70(0.21) &   0.73(0.11) &   1.47(0.13) &   1.71(0.16) \\
    Lasso &  -0.40(0.12) &  -0.17(0.08) &  -0.48(0.13) &   0.10(0.08) &   0.50(0.10) &   0.47(0.15) \\
    Adaptive Lasso &  -1.33(0.20) &  -0.74(0.24) &  -1.56(0.23) &   0.51(0.32) &   1.35(0.16) &   1.55(0.26) \\
    MCP   &  -1.43(0.17) &  -0.84(0.12) &  -1.67(0.21) &   0.70(0.16) &   1.44(0.13) &   1.67(0.17) \\
    SCAD  &  -1.33(0.19) &  -0.72(0.24) &  -1.55(0.21) &   0.43(0.32) &   1.33(0.14) &   1.53(0.20) \\
    MCP+mid-point & -13.10(3.17) &  -7.61(2.10) & -15.26(3.51) &   6.29(2.06) &  12.31(2.42) &  14.43(3.04) \\
    SCAD+mid-point & -16.26(3.85) &  -9.45(2.64) & -19.24(5.06) &   7.65(2.56) &  15.50(3.23) &  17.98(4.00) \\
&       &       &       &       &       &  \\
& $\beta_7$ & $\beta_8$ & $\beta_9$ & $\beta_{10}$ & $\beta_{11}$ & $\beta_{12}$ \\
\midrule
Truth & -0.52 & 0.86  & -1.23 & 1.18  & -1.97 & -1.68 \\
    Oracle &  -0.55(0.11) &   0.89(0.12) &  -1.28(0.13) &   1.22(0.12) &  -2.03(0.22) &  -1.75(0.21) \\
    Lasso &  -0.04(0.05) &   0.17(0.10) &  -0.37(0.09) &   0.36(0.10) &  -0.58(0.13) &  -0.47(0.13) \\
    Adaptive Lasso &  -0.25(0.29) &   0.73(0.30) &  -1.17(0.16) &   1.11(0.16) &  -1.86(0.26) &  -1.60(0.24) \\
    MCP   &  -0.45(0.25) &   0.88(0.13) &  -1.26(0.13) &   1.20(0.12) &  -2.00(0.21) &  -1.72(0.21) \\
    SCAD  &  -0.16(0.20) &   0.73(0.27) &  -1.16(0.14) &   1.10(0.12) &  -1.87(0.22) &  -1.60(0.21) \\
    MCP+mid-point &  -4.29(2.60) &   7.72(2.10) & -11.50(2.58) &  10.53(2.29) & -18.25(4.18) & -15.88(3.85) \\
    SCAD+mid-point &  -5.26(3.35) &   9.85(2.78) & -14.54(3.14) &  13.12(2.91) & -23.03(5.41) & -20.05(5.24) \\
\bottomrule
\end{tabular}%
\label{tab:addlabel}%
\end{table}%

\begin{table}[H]
\renewcommand{\arraystretch}{0.8}
\centering
\caption{Empirical mean estimates of $\bbeta^{(12)}$ with $p = 10000$, $n = 1000$ and $\rho = 0$. The values in parentheses are the empirical standard errors.}
\footnotesize
\begin{tabular}{ccccccc}
\toprule
Method & $\beta_1$ & $\beta_2$ & $\beta_3$ & $\beta_4$ & $\beta_5$ & $\beta_6$ \\
\midrule
Truth & -1.4  & -0.83 & -1.64 & 0.69  & 1.39  & 1.65 \\
    Oracle &  -1.43(0.12) &  -0.84(0.07) &  -1.65(0.13) &   0.70(0.08) &   1.40(0.09) &   1.68(0.13) \\
    Lasso &  -0.64(0.08) &  -0.33(0.06) &  -0.75(0.09) &   0.25(0.07) &   0.69(0.07) &   0.78(0.11) \\
    Adaptive Lasso &  -1.39(0.12) &  -0.81(0.08) &  -1.60(0.14) &   0.67(0.09) &   1.36(0.09) &   1.62(0.14) \\
    MCP   &  -1.42(0.12) &  -0.84(0.08) &  -1.65(0.13) &   0.70(0.08) &   1.40(0.09) &   1.67(0.13) \\
    SCAD  &  -1.39(0.12) &  -0.82(0.08) &  -1.61(0.13) &   0.68(0.10) &   1.36(0.09) &   1.62(0.13) \\
    MCP+mid-point & -11.71(2.35) &  -6.78(1.55) & -13.55(2.56) &   5.49(1.27) &  10.87(1.83) &  12.39(2.28) \\
    SCAD+mid-point & -15.01(3.01) &  -8.61(1.92) & -17.41(3.61) &   7.03(1.75) &  13.90(2.49) &  16.10(3.22) \\
&       &       &       &       &       &  \\
& $\beta_7$ & $\beta_8$ & $\beta_9$ & $\beta_{10}$ & $\beta_{11}$ & $\beta_{12}$ \\
\midrule
Truth & -0.52 & 0.86  & -1.23 & 1.18  & -1.97 & -1.68 \\
    Oracle &  -0.52(0.08) &   0.88(0.08) &  -1.25(0.10) &   1.20(0.09) &  -2.00(0.15) &  -1.70(0.13) \\
    Lasso &  -0.13(0.06) &   0.35(0.07) &  -0.58(0.07) &   0.57(0.07) &  -0.90(0.11) &  -0.74(0.09) \\
    Adaptive Lasso &  -0.47(0.13) &   0.84(0.09) &  -1.21(0.11) &   1.16(0.10) &  -1.94(0.16) &  -1.65(0.15) \\
    MCP   &  -0.52(0.09) &   0.87(0.08) &  -1.24(0.10) &   1.20(0.09) &  -1.99(0.15) &  -1.70(0.13) \\
    SCAD  &  -0.39(0.21) &   0.85(0.08) &  -1.22(0.10) &   1.16(0.09) &  -1.95(0.15) &  -1.66(0.13) \\
    MCP+mid-point &  -4.18(1.32) &   6.86(1.56) & -10.09(1.92) &   9.40(1.72) & -16.20(3.31) & -14.00(2.75) \\
    SCAD+mid-point &  -5.39(1.74) &   8.77(2.00) & -13.04(2.61) &  12.12(2.30) & -20.95(4.35) & -17.97(3.84) \\
\bottomrule
\end{tabular}%
\label{tab:addlabel}%
\end{table}%

\begin{table}[H]
\renewcommand{\arraystretch}{0.8}
\centering
\caption{Empirical mean estimates of $\bbeta^{(12)}$ with $p = 3000$, $n = 500$ and $\rho = 0.8$. The values in parentheses are the empirical standard errors.}
\footnotesize
\begin{tabular}{ccccccc}
\toprule
Method & $\beta_1$ & $\beta_2$ & $\beta_3$ & $\beta_4$ & $\beta_5$ & $\beta_6$ \\
\midrule
Truth & -1.4  & -0.83 & -1.64 & 0.69  & 1.39  & 1.65 \\
    Oracle &  -1.46(0.17) &  -0.85(0.12) &  -1.69(0.19) &   0.71(0.11) &   1.46(0.13) &   1.73(0.17) \\
    Lasso &  -0.45(0.13) &  -0.20(0.09) &  -0.55(0.11) &   0.12(0.09) &   0.53(0.11) &   0.54(0.16) \\
    Adaptive Lasso &  -1.31(0.20) &  -0.72(0.20) &  -1.53(0.22) &   0.51(0.28) &   1.33(0.15) &   1.54(0.26) \\
    MCP   &  -1.45(0.17) &  -0.84(0.12) &  -1.68(0.20) &   0.69(0.14) &   1.45(0.13) &   1.71(0.17) \\
    SCAD  &  -1.37(0.17) &  -0.76(0.18) &  -1.59(0.20) &   0.54(0.28) &   1.36(0.13) &   1.60(0.17) \\
    MCP+mid-point & -12.52(3.18) &  -7.27(2.30) & -14.83(3.47) &   5.82(2.00) &  11.80(2.54) &  13.83(3.56) \\
    SCAD+mid-point & -15.91(3.96) &  -9.18(2.97) & -18.87(4.73) &   7.48(2.75) &  15.41(3.56) &  17.80(4.31) \\
&       &       &       &       &       &  \\
& $\beta_7$ & $\beta_8$ & $\beta_9$ & $\beta_{10}$ & $\beta_{11}$ & $\beta_{12}$ \\
\midrule
Truth & -0.52 & 0.86  & -1.23 & 1.18  & -1.97 & -1.68 \\
    Oracle &  -0.54(0.11) &   0.91(0.11) &  -1.29(0.13) &   1.21(0.12) &  -2.08(0.25) &  -1.74(0.21) \\
    Lasso &  -0.05(0.06) &   0.22(0.10) &  -0.41(0.10) &   0.40(0.10) &  -0.67(0.16) &  -0.51(0.14) \\
    Adaptive Lasso &  -0.24(0.27) &   0.78(0.21) &  -1.16(0.15) &   1.09(0.13) &  -1.88(0.26) &  -1.57(0.23) \\
    MCP   &  -0.47(0.22) &   0.90(0.11) &  -1.28(0.14) &   1.20(0.11) &  -2.07(0.25) &  -1.74(0.22) \\
    SCAD  &  -0.23(0.24) &   0.82(0.19) &  -1.21(0.14) &   1.13(0.11) &  -1.96(0.25) &  -1.64(0.21) \\
    MCP+mid-point &  -3.96(2.44) &   7.41(2.30) & -10.90(2.82) &  10.07(2.47) & -17.77(4.33) & -15.14(4.16) \\
    SCAD+mid-point &  -5.00(3.36) &   9.47(2.86) & -13.93(3.38) &  12.99(3.18) & -23.32(5.53) & -19.32(5.19) \\
\bottomrule
\end{tabular}%
\label{tab:addlabel}%
\end{table}%

\begin{table}[H]
\renewcommand{\arraystretch}{0.8}
\centering
\caption{Empirical mean estimates of $\bbeta^{(12)}$ with $p = 3000$, $n = 1000$ and $\rho = 0.8$. The values in parentheses are the empirical standard errors.}
\footnotesize
\begin{tabular}{ccccccc}
\toprule
Method & $\beta_1$ & $\beta_2$ & $\beta_3$ & $\beta_4$ & $\beta_5$ & $\beta_6$ \\
\midrule
Truth & -1.4  & -0.83 & -1.64 & 0.69  & 1.39  & 1.65 \\
    Oracle &  -1.42(0.11) &  -0.84(0.08) &  -1.64(0.12) &   0.71(0.08) &   1.40(0.09) &   1.68(0.12) \\
    Lasso &  -0.68(0.08) &  -0.37(0.06) &  -0.80(0.09) &   0.28(0.06) &   0.74(0.07) &   0.82(0.10) \\
    Adaptive Lasso &  -1.36(0.12) &  -0.80(0.09) &  -1.58(0.12) &   0.66(0.09) &   1.34(0.10) &   1.59(0.12) \\
    MCP   &  -1.42(0.11) &  -0.84(0.08) &  -1.64(0.12) &   0.70(0.08) &   1.40(0.09) &   1.67(0.11) \\
    SCAD  &  -1.40(0.11) &  -0.83(0.08) &  -1.62(0.12) &   0.69(0.08) &   1.37(0.09) &   1.63(0.12) \\
    MCP+mid-point &  -9.38(2.10) &  -5.56(1.47) & -10.86(2.25) &   4.49(1.35) &   8.68(1.69) &  10.16(2.01) \\
    SCAD+mid-point & -12.85(2.78) &  -7.48(1.94) & -14.84(3.27) &   6.23(1.85) &  11.78(2.16) &  13.93(3.14) \\
&       &       &       &       &       &  \\
& $\beta_7$ & $\beta_8$ & $\beta_9$ & $\beta_{10}$ & $\beta_{11}$ & $\beta_{12}$ \\
\midrule
Truth & -0.52 & 0.86  & -1.23 & 1.18  & -1.97 & -1.68 \\
    Oracle &  -0.53(0.07) &   0.87(0.07) &  -1.24(0.10) &   1.19(0.08) &  -1.99(0.16) &  -1.70(0.14) \\
    Lasso &  -0.16(0.05) &   0.38(0.06) &  -0.60(0.07) &   0.60(0.06) &  -0.96(0.11) &  -0.80(0.10) \\
    Adaptive Lasso &  -0.48(0.09) &   0.82(0.08) &  -1.19(0.11) &   1.14(0.08) &  -1.91(0.16) &  -1.62(0.15) \\
    MCP   &  -0.53(0.09) &   0.87(0.07) &  -1.24(0.10) &   1.19(0.08) &  -1.99(0.16) &  -1.70(0.14) \\
    SCAD  &  -0.47(0.16) &   0.85(0.08) &  -1.22(0.10) &   1.17(0.08) &  -1.96(0.15) &  -1.67(0.14) \\
    MCP+mid-point &  -3.45(1.27) &   5.56(1.46) &  -8.17(1.84) &   7.57(1.62) & -13.30(2.89) & -10.94(2.55) \\
    SCAD+mid-point &  -4.70(1.55) &   7.51(1.93) & -11.01(2.26) &  10.48(2.15) & -18.41(3.91) & -15.19(3.43) \\
\bottomrule
\end{tabular}%
\label{tab:addlabel}%
\end{table}%

\begin{table}[H]
\renewcommand{\arraystretch}{0.8}
\centering
\caption{Empirical mean estimates of $\bbeta^{(12)}$ with $p = 10000$, $n = 500$ and $\rho = 0.8$. The values in parentheses are the empirical standard errors.}
\footnotesize
\begin{tabular}{ccccccc}
\toprule
Method & $\beta_1$ & $\beta_2$ & $\beta_3$ & $\beta_4$ & $\beta_5$ & $\beta_6$ \\
\midrule
Truth & -1.4  & -0.83 & -1.64 & 0.69  & 1.39  & 1.65 \\
    Oracle &  -1.43(0.17) &  -0.85(0.12) &  -1.70(0.18) &   0.71(0.12) &   1.43(0.13) &   1.71(0.19) \\
    Lasso &  -0.36(0.11) &  -0.16(0.09) &  -0.48(0.13) &   0.08(0.08) &   0.46(0.12) &   0.43(0.16) \\
    Adaptive Lasso &  -1.26(0.21) &  -0.68(0.27) &  -1.52(0.21) &   0.43(0.33) &   1.27(0.16) &   1.48(0.30) \\
    MCP   &  -1.41(0.17) &  -0.84(0.14) &  -1.68(0.18) &   0.67(0.19) &   1.40(0.14) &   1.67(0.19) \\
    SCAD  &  -1.30(0.18) &  -0.69(0.26) &  -1.54(0.19) &   0.38(0.32) &   1.28(0.14) &   1.51(0.23) \\
    MCP+mid-point & -12.69(3.18) &  -7.54(2.28) & -15.00(3.32) &   5.73(2.11) &  11.84(2.47) &  13.67(2.99) \\
    SCAD+mid-point & -16.08(4.31) &  -9.60(2.98) & -19.45(4.60) &   7.64(2.99) &  15.42(3.47) &  17.81(3.89) \\
&       &       &       &       &       &  \\
& $\beta_7$ & $\beta_8$ & $\beta_9$ & $\beta_{10}$ & $\beta_{11}$ & $\beta_{12}$ \\
\midrule
Truth & -0.52 & 0.86  & -1.23 & 1.18  & -1.97 & -1.68 \\
    Oracle &  -0.54(0.11) &   0.89(0.12) &  -1.27(0.15) &   1.23(0.11) &  -2.03(0.23) &  -1.72(0.20) \\
    Lasso &  -0.03(0.04) &   0.15(0.10) &  -0.34(0.11) &   0.35(0.10) &  -0.56(0.14) &  -0.42(0.13) \\
    Adaptive Lasso &  -0.16(0.26) &   0.67(0.31) &  -1.13(0.19) &   1.09(0.15) &  -1.81(0.24) &  -1.52(0.25) \\
    MCP   &  -0.44(0.25) &   0.87(0.12) &  -1.25(0.15) &   1.20(0.12) &  -2.01(0.24) &  -1.70(0.20) \\
    SCAD  &  -0.15(0.19) &   0.69(0.29) &  -1.14(0.16) &   1.10(0.12) &  -1.84(0.23) &  -1.55(0.26) \\
    MCP+mid-point &  -3.74(2.65) &   7.44(1.90) & -11.12(2.44) &  10.34(2.13) & -17.98(4.18) & -15.23(3.65) \\
    SCAD+mid-point &  -5.04(3.36) &   9.62(2.66) & -14.54(3.56) &  13.61(3.29) & -23.56(5.58) & -19.47(4.88) \\
\bottomrule
\end{tabular}%
\label{tab:addlabel}%
\end{table}%

\begin{table}[H]
\renewcommand{\arraystretch}{0.8}
\centering
\caption{Empirical mean estimates of $\bbeta^{(12)}$ with $p = 10000$, $n = 1000$ and $\rho = 0.8$. The values in parentheses are the empirical standard errors.}
\footnotesize
\begin{tabular}{ccccccc}
\toprule
Method & $\beta_1$ & $\beta_2$ & $\beta_3$ & $\beta_4$ & $\beta_5$ & $\beta_6$ \\
\midrule
Truth & -1.4  & -0.83 & -1.64 & 0.69  & 1.39  & 1.65 \\
    Oracle &  -1.42(0.12) &  -0.84(0.08) &  -1.65(0.12) &   0.70(0.07) &   1.41(0.09) &   1.67(0.11) \\
    Lasso &  -0.62(0.09) &  -0.32(0.06) &  -0.74(0.08) &   0.24(0.06) &   0.69(0.07) &   0.74(0.09) \\
    Adaptive Lasso &  -1.36(0.13) &  -0.80(0.09) &  -1.58(0.13) &   0.66(0.08) &   1.35(0.10) &   1.59(0.13) \\
    MCP   &  -1.41(0.12) &  -0.84(0.08) &  -1.64(0.12) &   0.70(0.07) &   1.40(0.08) &   1.65(0.11) \\
    SCAD  &  -1.38(0.12) &  -0.82(0.08) &  -1.60(0.12) &   0.67(0.08) &   1.37(0.08) &   1.61(0.11) \\
    MCP+mid-point & -11.13(2.24) &  -6.48(1.51) & -12.79(2.29) &   5.30(1.44) &  10.21(1.81) &  12.01(2.13) \\
    SCAD+mid-point & -13.97(3.10) &  -8.21(1.99) & -16.11(3.38) &   6.68(1.66) &  13.01(2.64) &  14.91(3.08) \\
&       &       &       &       &       &  \\
& $\beta_7$ & $\beta_8$ & $\beta_9$ & $\beta_{10}$ & $\beta_{11}$ & $\beta_{12}$ \\
\midrule
Truth & -0.52 & 0.86  & -1.23 & 1.18  & -1.97 & -1.68 \\
    Oracle &  -0.52(0.08) &   0.88(0.08) &  -1.25(0.09) &   1.20(0.09) &  -2.00(0.16) &  -1.71(0.15) \\
    Lasso &  -0.12(0.06) &   0.34(0.07) &  -0.56(0.07) &   0.56(0.07) &  -0.89(0.11) &  -0.74(0.10) \\
    Adaptive Lasso &  -0.46(0.12) &   0.83(0.09) &  -1.20(0.09) &   1.15(0.10) &  -1.93(0.17) &  -1.65(0.16) \\
    MCP   &  -0.50(0.11) &   0.87(0.08) &  -1.24(0.09) &   1.19(0.08) &  -2.00(0.16) &  -1.71(0.15) \\
    SCAD  &  -0.38(0.21) &   0.85(0.08) &  -1.21(0.09) &   1.16(0.08) &  -1.95(0.16) &  -1.67(0.15) \\
    MCP+mid-point &  -3.95(1.38) &   6.47(1.48) &  -9.70(1.98) &   8.81(1.48) & -15.66(3.04) & -13.40(2.69) \\
    SCAD+mid-point &  -5.09(1.66) &   8.13(2.03) & -12.22(2.66) &  11.23(2.44) & -20.05(4.37) & -17.06(3.65) \\
\bottomrule
\end{tabular}%
\label{tab:addlabel}%
\end{table}%

\subsection{Estimation results of baseline cumulative hazard functions}
\begin{figure}[H]
     \centering
 \includegraphics[width=\textwidth]{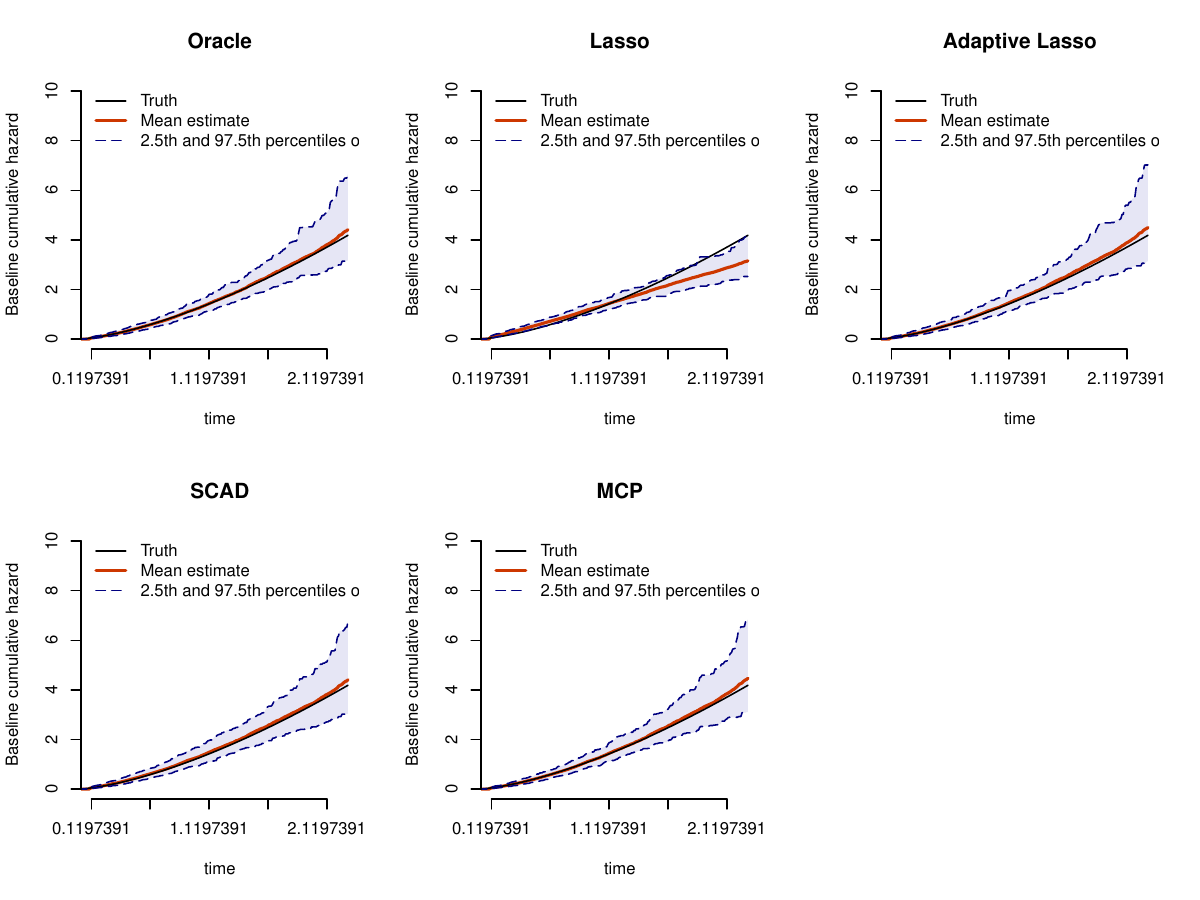}
 \caption{Empirical mean estimates for the baseline cumulative hazard function in the scenario of six nonzero coefficients with $p = 3000$, $n = 500$, and $\rho = 0$} \label{fig:2}
 \end{figure}

\begin{figure}[H]
     \centering
 \includegraphics[width=\textwidth]{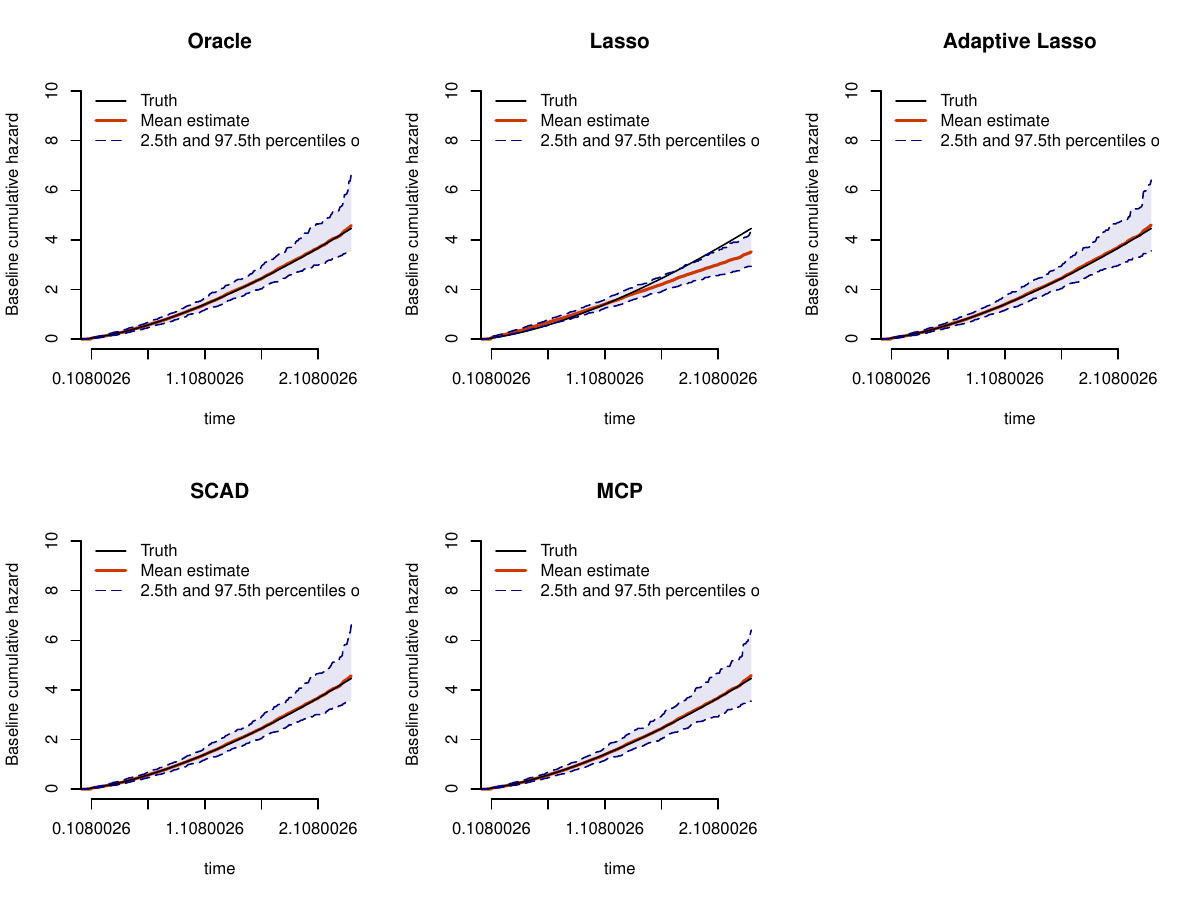}
 \caption{Empirical mean estimates for the baseline cumulative hazard function in the scenario of six nonzero coefficients with $p = 3000$, $n = 1000$, and $\rho = 0$} \label{fig:2}
 \end{figure}

\begin{figure}[H]
     \centering
 \includegraphics[width=\textwidth]{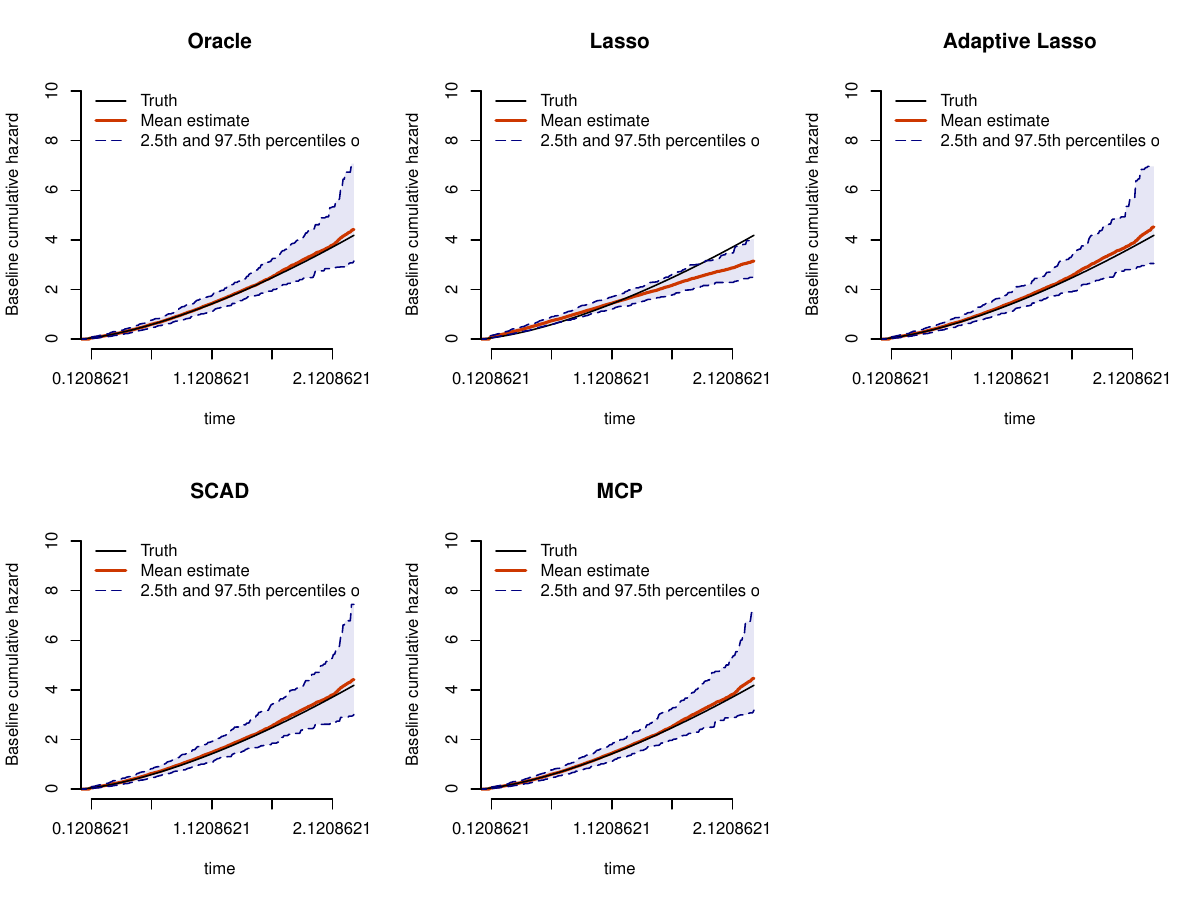}
 \caption{Empirical mean estimates for the baseline cumulative hazard function in the scenario of six nonzero coefficients with $p = 3000$, $n = 500$, and $\rho = 0.8$} \label{fig:2}
 \end{figure}

\begin{figure}[H]
     \centering
 \includegraphics[width=\textwidth]{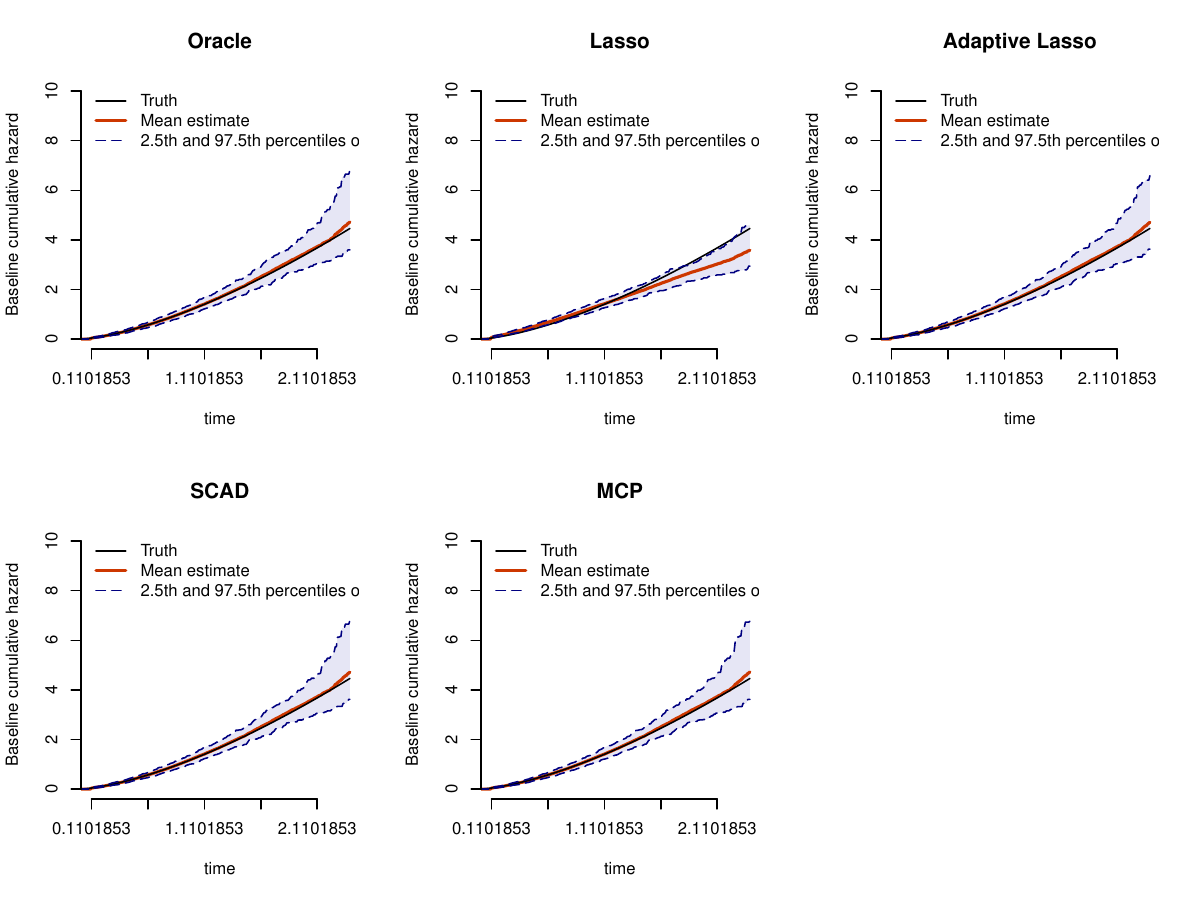}
 \caption{Empirical mean estimates for the baseline cumulative hazard function in the scenario of six nonzero coefficients with $p = 3000$, $n = 1000$, and $\rho = 0.8$} \label{fig:2}
 \end{figure}

\begin{figure}[H]
 \includegraphics[width=\textwidth]{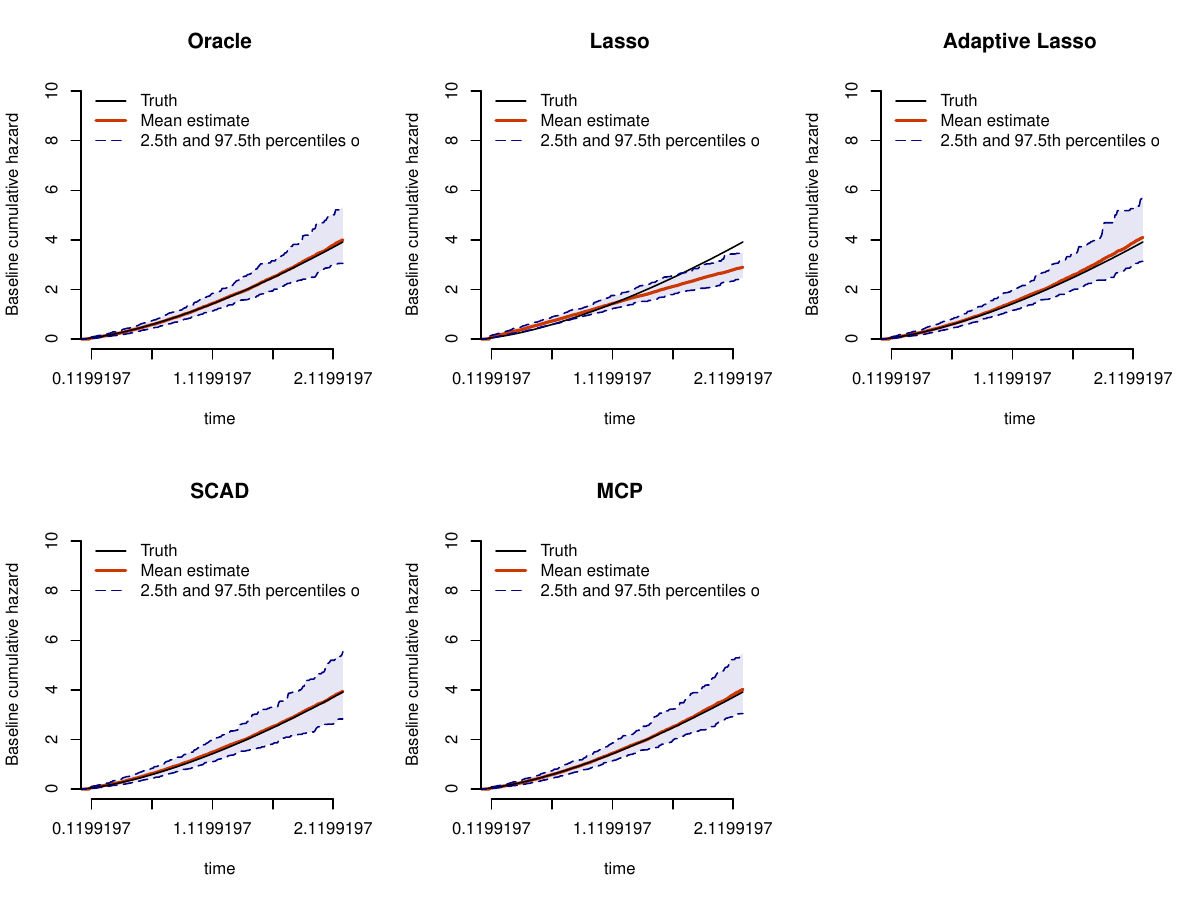}
 \caption{Empirical mean estimates for the baseline cumulative hazard function in the scenario of six nonzero coefficients with $p = 10000$, $n = 500$, and $\rho = 0$} \label{fig:2}
 \end{figure}

\begin{figure}[H]
     \centering
 \includegraphics[width=\textwidth]{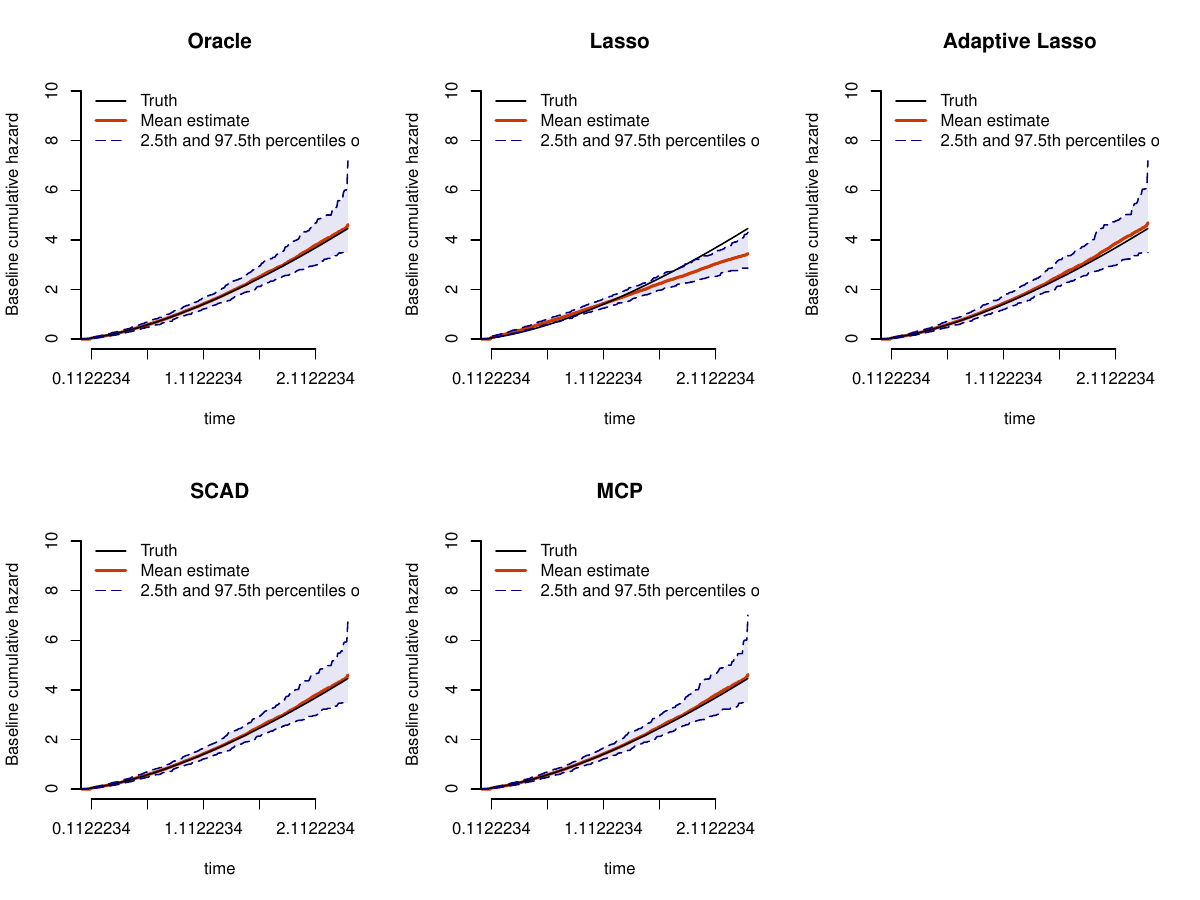}
 \caption{Empirical mean estimates for the baseline cumulative hazard function in the scenario of six nonzero coefficients with $p = 10000$, $n = 1000$, and $\rho = 0$} \label{fig:2}
 \end{figure}

\begin{figure}[H]
     \centering
 \includegraphics[width=\textwidth]{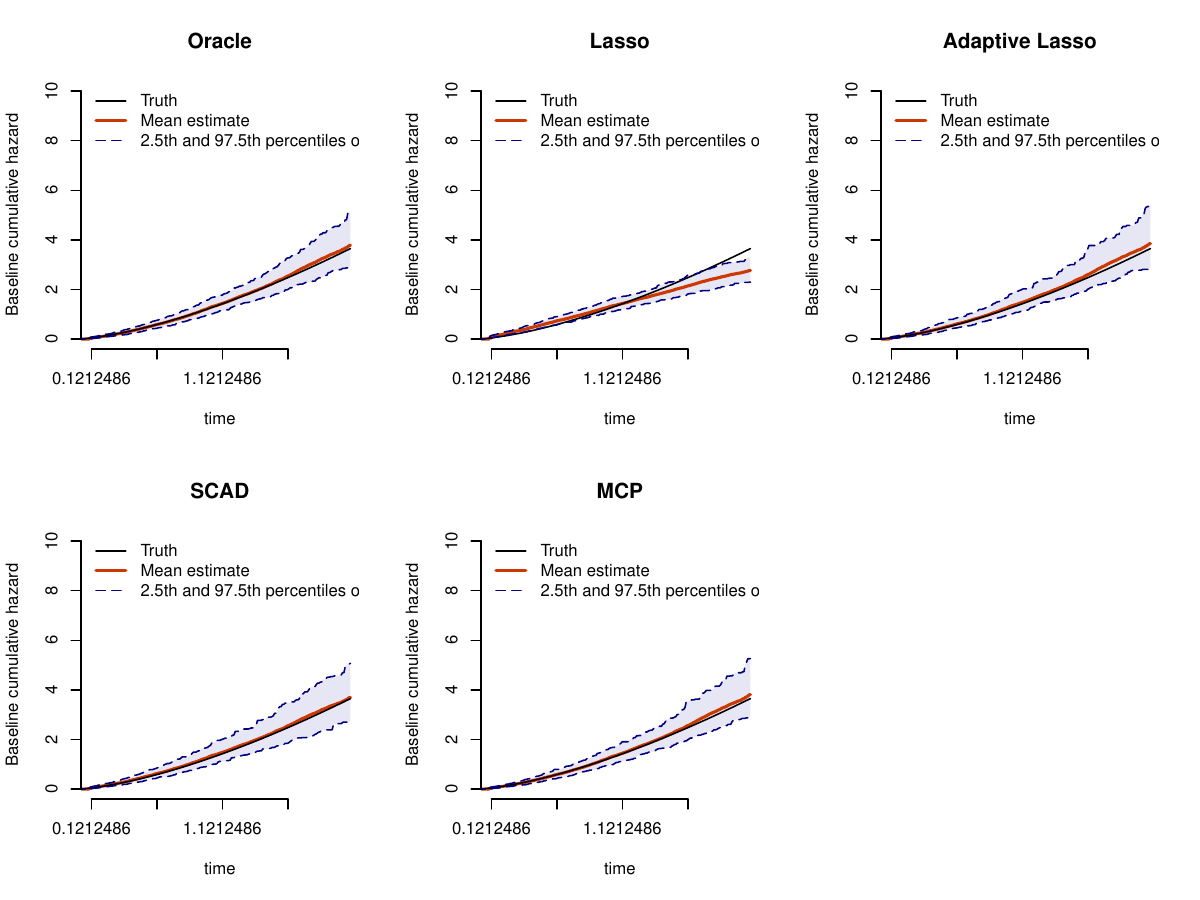}
 \caption{Empirical mean estimates for the baseline cumulative hazard function in the scenario of six nonzero coefficients with $p = 10000$, $n = 500$, and $\rho = 0.8$} \label{fig:2}
 \end{figure}

\begin{figure}[H]
     \centering
 \includegraphics[width=\textwidth]{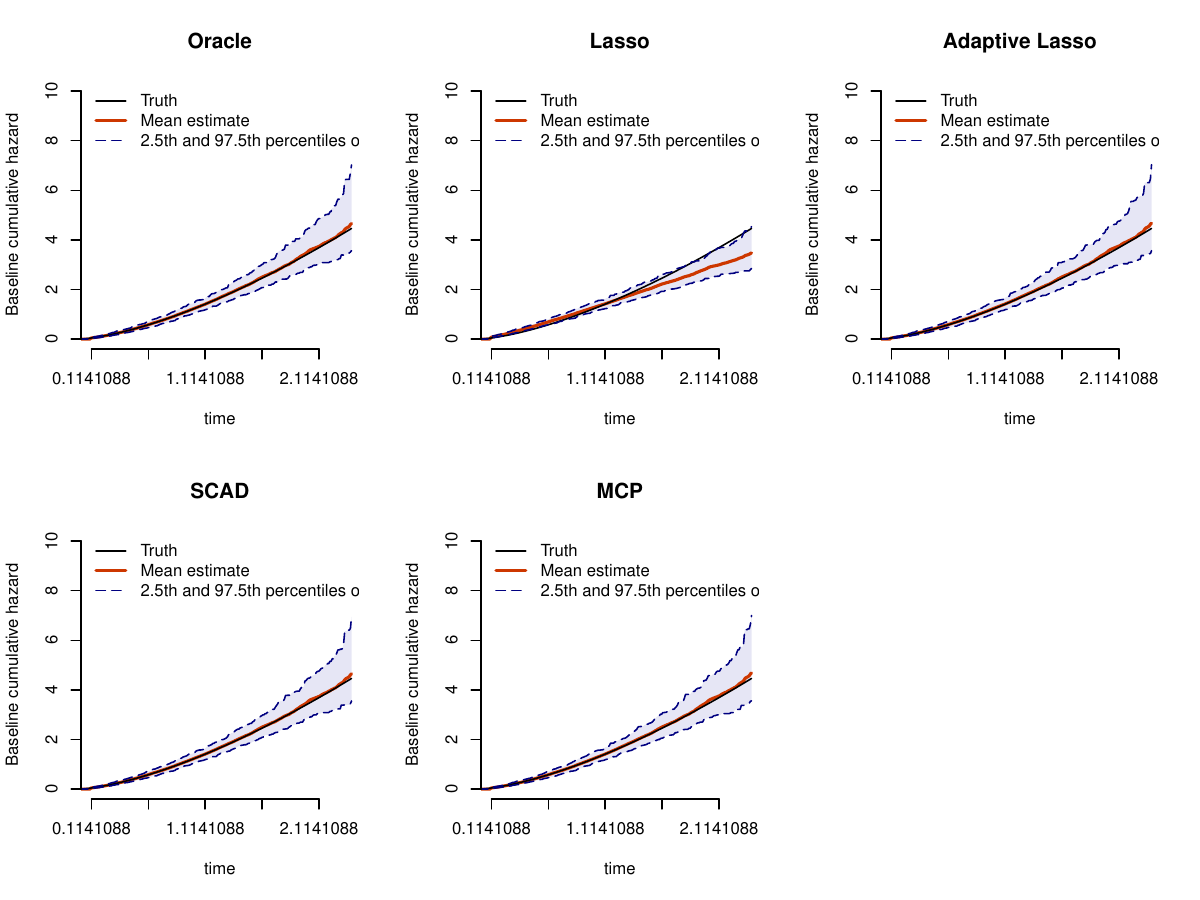}
 \caption{Empirical mean estimates for the baseline cumulative hazard function in the scenario of six nonzero coefficients with $p = 10000$, $n = 1000$, and $\rho = 0.8$} \label{fig:2}
 \end{figure}

\begin{figure}[H]
     \centering
 \includegraphics[width=\textwidth]{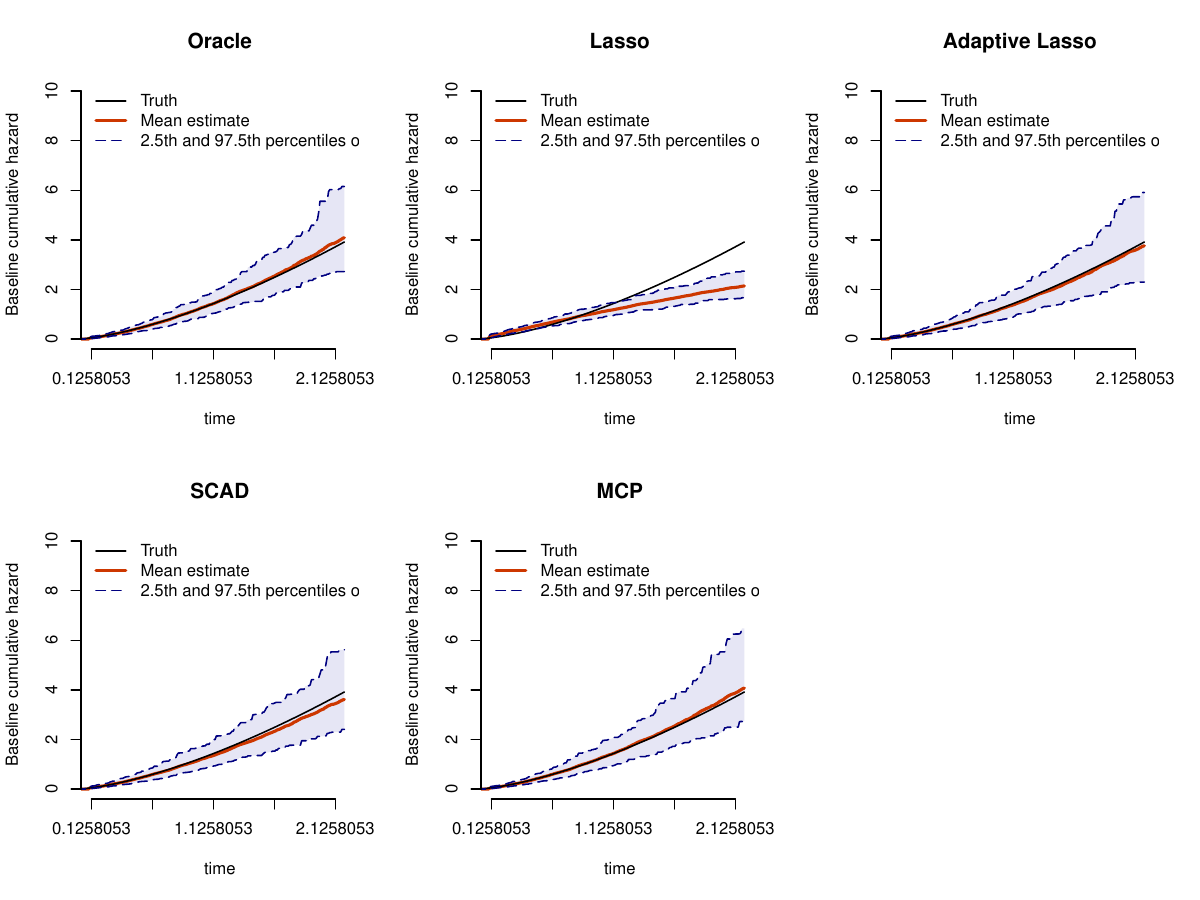}
 \caption{Empirical mean estimates for the baseline cumulative hazard function in the scenario of twelve nonzero coefficients with $p = 3000$, $n = 500$, and $\rho = 0$} \label{fig:2}
 \end{figure}

\begin{figure}[H]
     \centering
 \includegraphics[width=\textwidth]{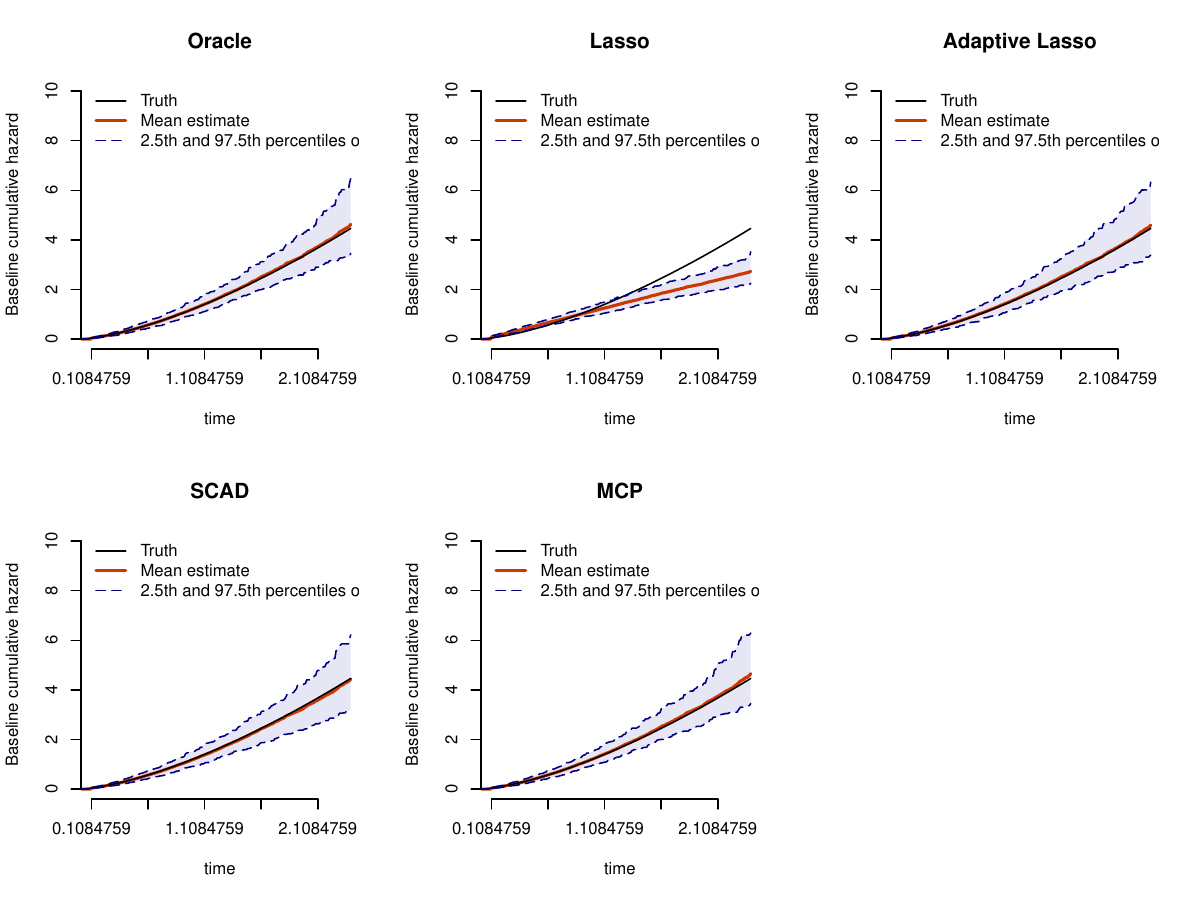}
 \caption{Empirical mean estimates for the baseline cumulative hazard function in the scenario of twelve nonzero coefficients with $p = 3000$, $n = 1000$, and $\rho = 0$} \label{fig:2}
 \end{figure}

\begin{figure}[H]
     \centering
 \includegraphics[width=\textwidth]{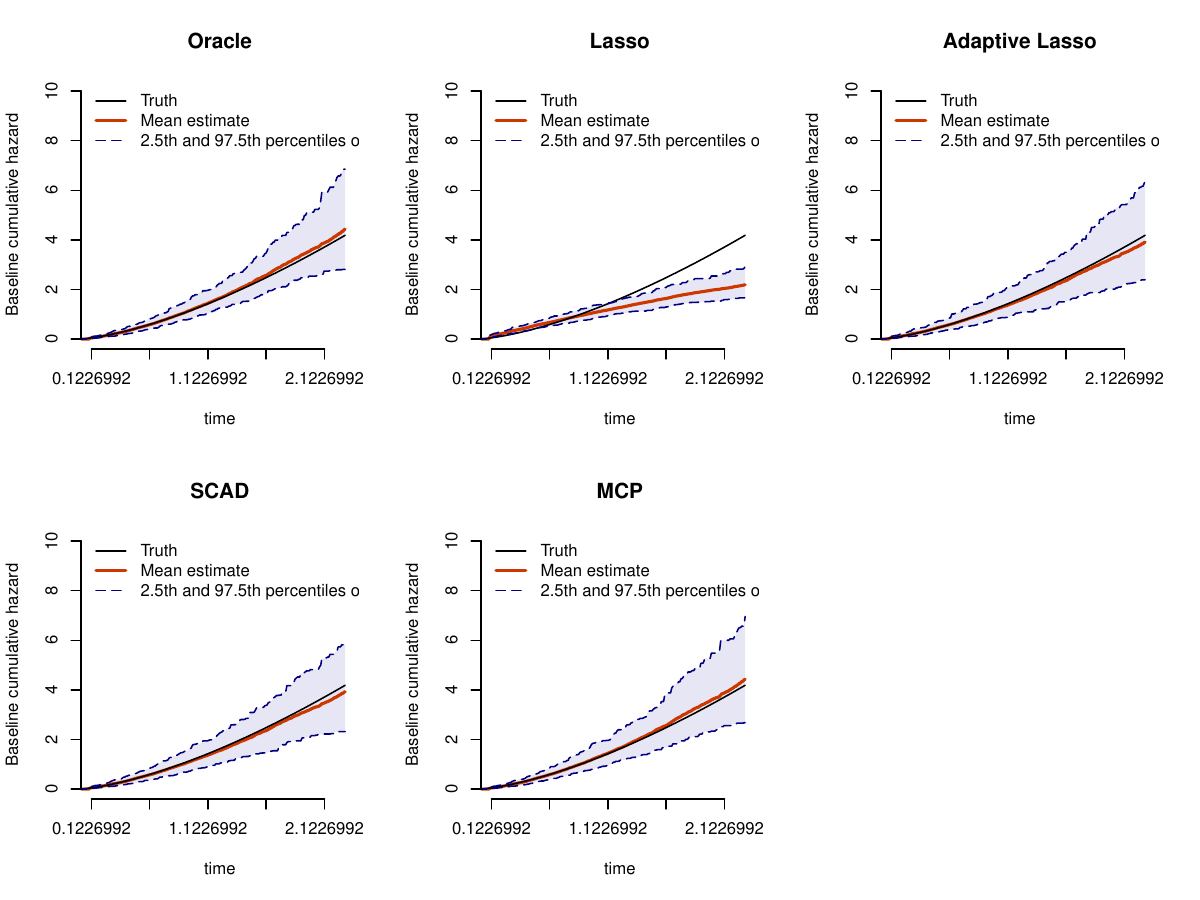}
 \caption{Empirical mean estimates for the baseline cumulative hazard function in the scenario of twelve nonzero coefficients with $p = 3000$, $n = 500$, and $\rho = 0.8$} \label{fig:2}
 \end{figure}

\begin{figure}[H]
     \centering
 \includegraphics[width=\textwidth]{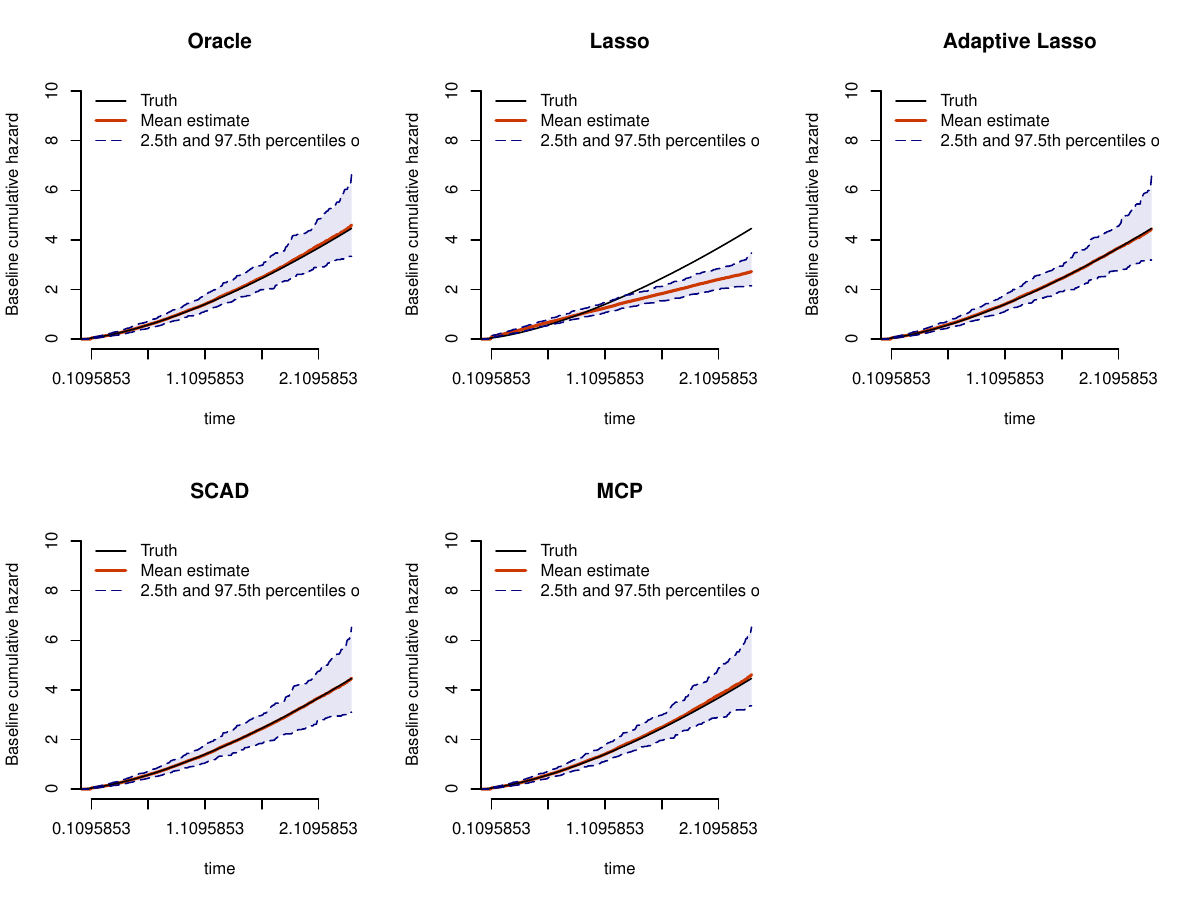}
 \caption{Empirical mean estimates for the baseline cumulative hazard function in the scenario of twelve nonzero coefficients with $p = 3000$, $n = 1000$, and $\rho = 0.8$} \label{fig:2}
 \end{figure}

\begin{figure}[H]
 \includegraphics[width=\textwidth]{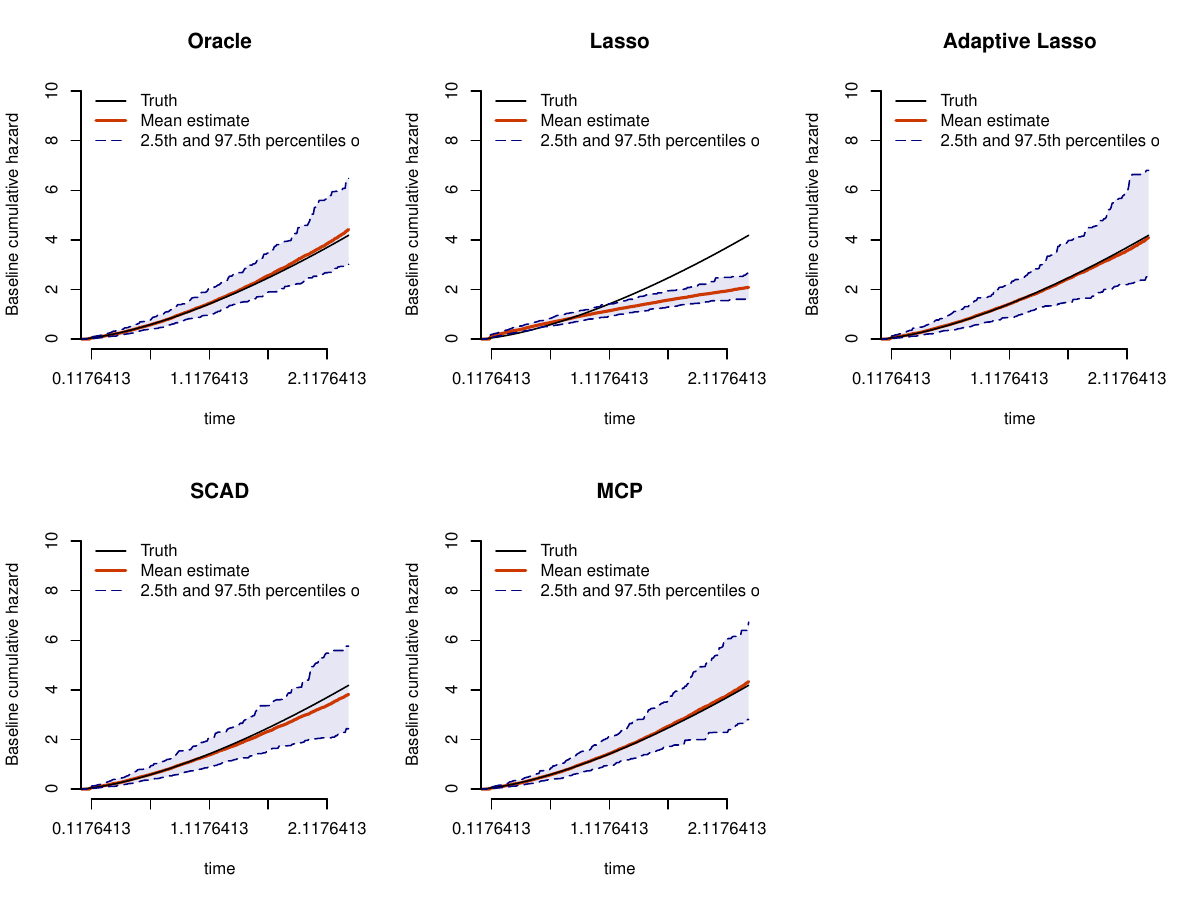}
 \caption{Empirical mean estimates for the baseline cumulative hazard function in the scenario of twelve nonzero coefficients with $p = 10000$, $n = 500$, and $\rho = 0$} \label{fig:2}
 \end{figure}

\begin{figure}[H]
     \centering
 \includegraphics[width=\textwidth]{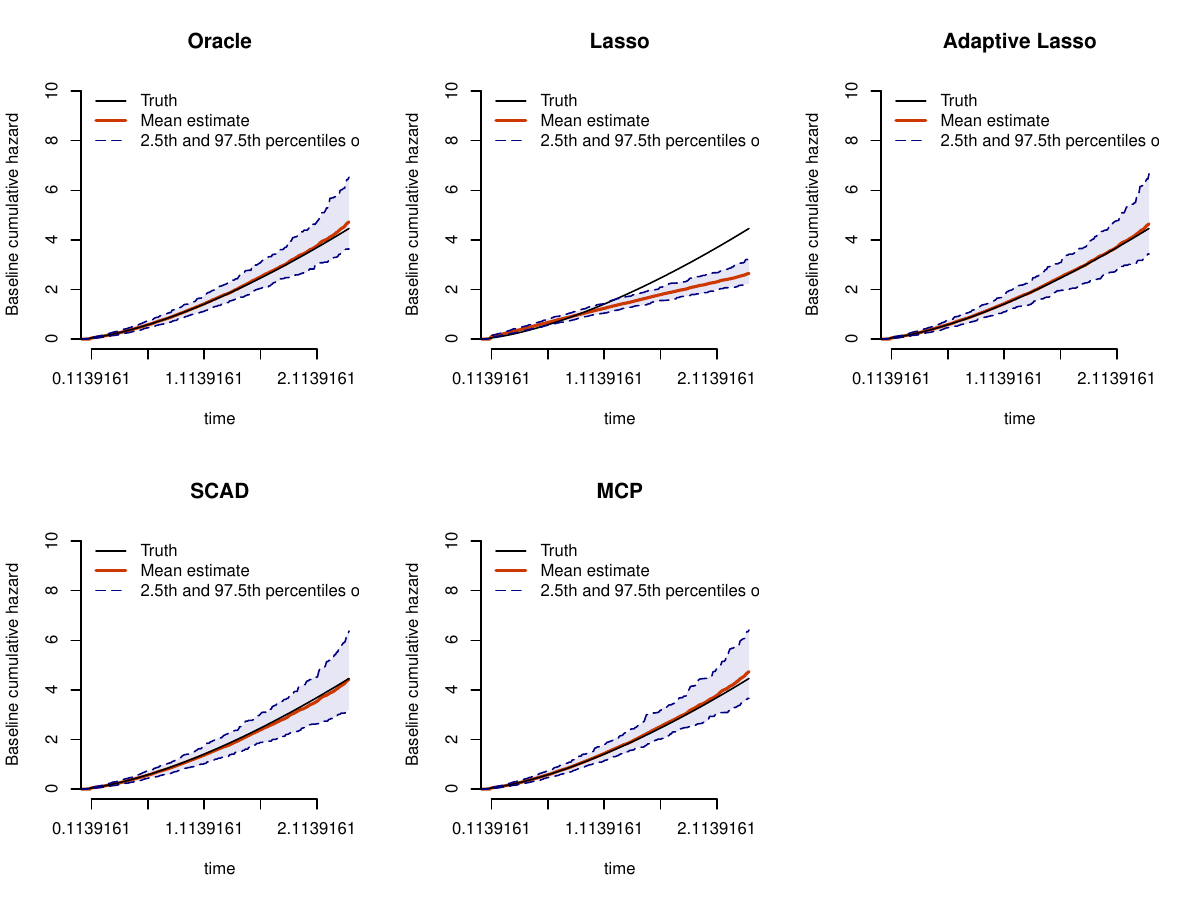}
 \caption{Empirical mean estimates for the baseline cumulative hazard function in the scenario of twelve nonzero coefficients with $p = 10000$, $n = 1000$, and $\rho = 0$} \label{fig:2}
 \end{figure}

\begin{figure}[H]
     \centering
 \includegraphics[width=\textwidth]{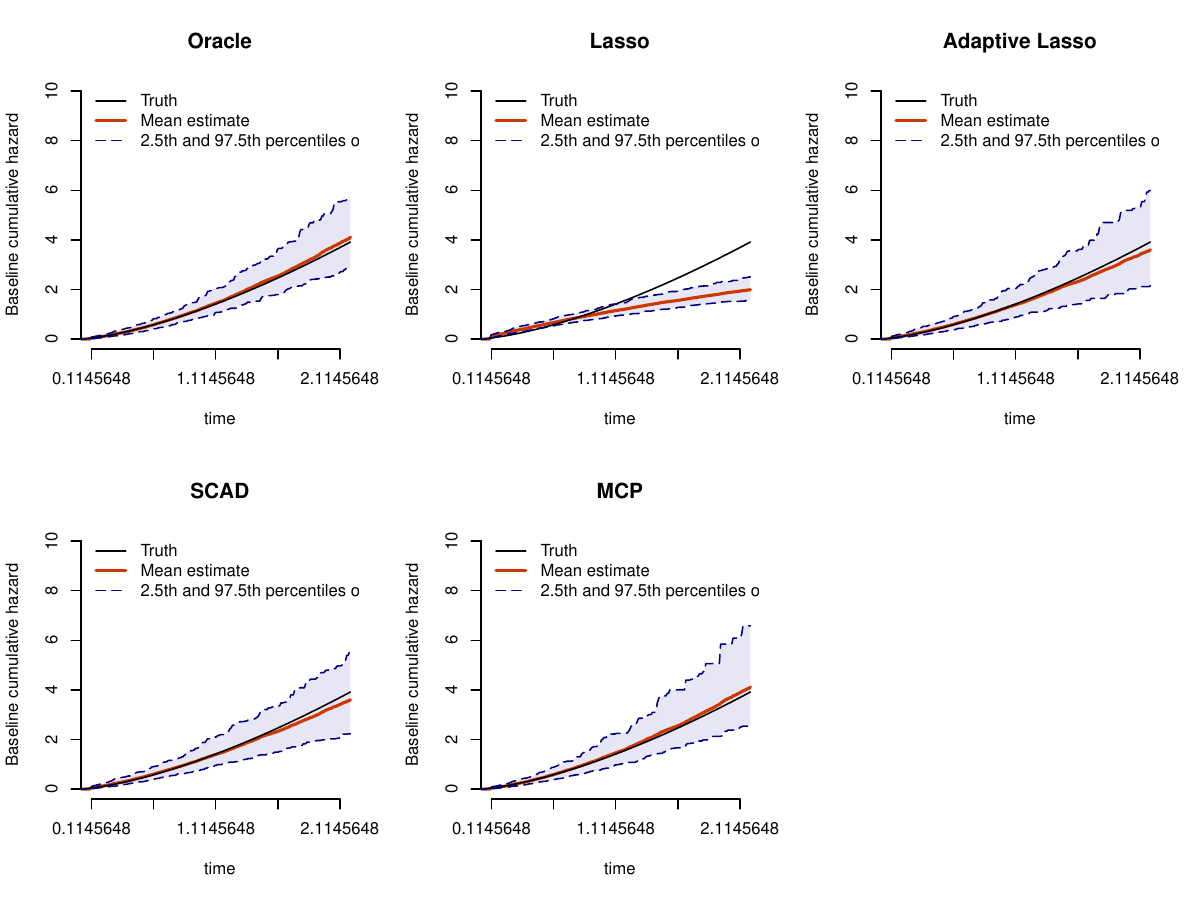}
 \caption{Empirical mean estimates for the baseline cumulative hazard function in the scenario of twelve nonzero coefficients with $p = 10000$, $n = 500$, and $\rho = 0.8$} \label{fig:2}
 \end{figure}

\begin{figure}[H]
     \centering
 \includegraphics[width=\textwidth]{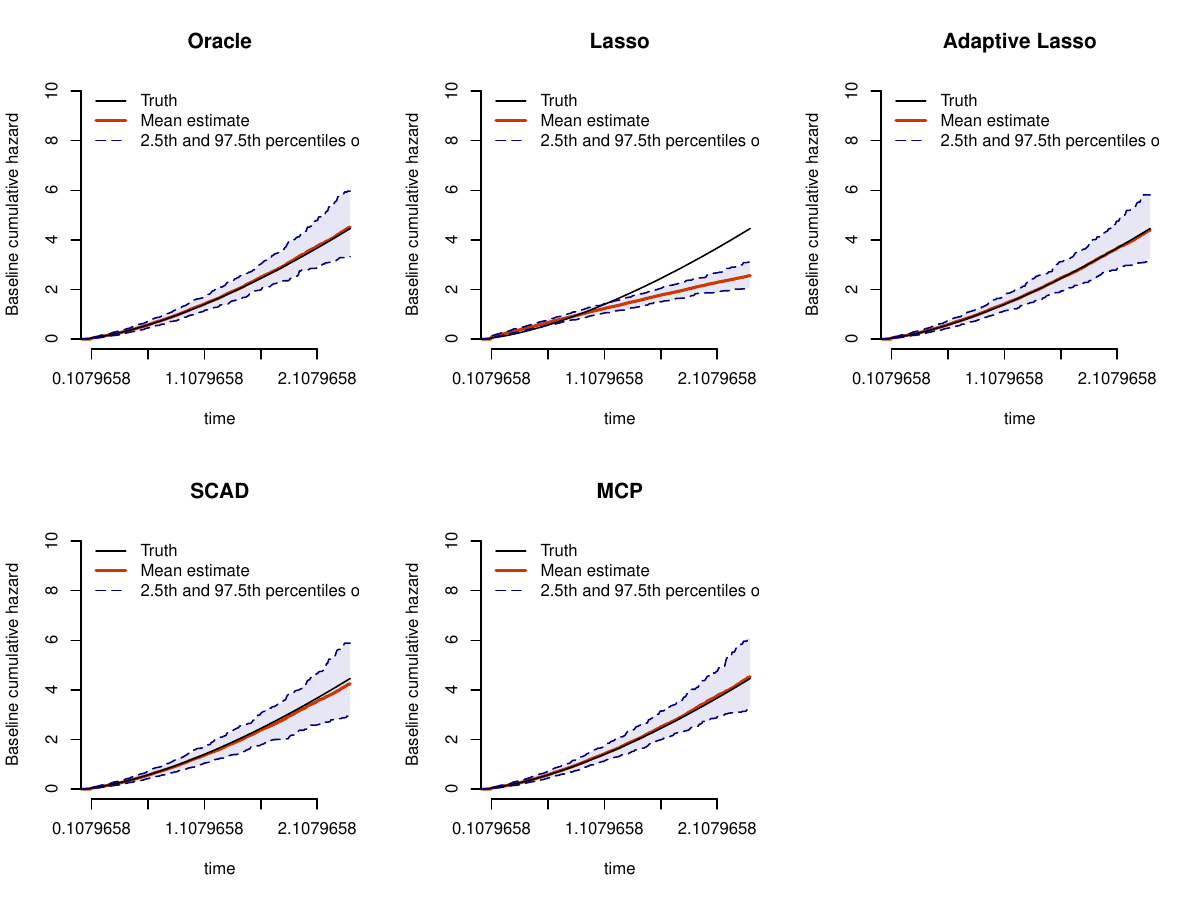}
 \caption{Empirical mean estimates for the baseline cumulative hazard function in the scenario of twelve nonzero coefficients with $p = 10000$, $n = 1000$, and $\rho = 0.8$} \label{fig:2}
 \end{figure}

 \subsection{Normal Q-Q plots of nonzero coefficient estimates}

\begin{figure}[H]\centering
\subfloat[Lasso]
        {\includegraphics[width=0.5\textwidth]{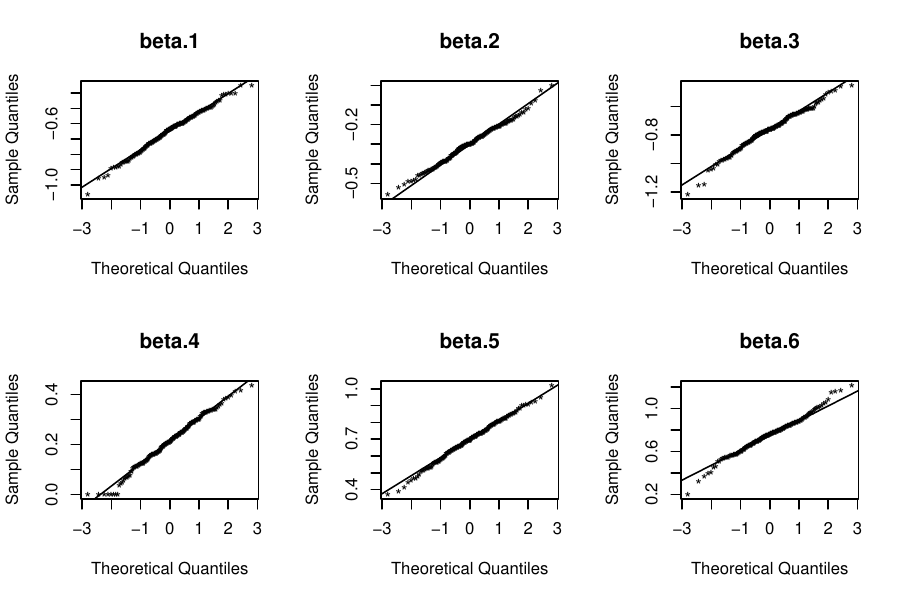}}
\subfloat[Adaptive Lasso \label{fig:mean and std of net24}]
         {\includegraphics[width=0.5\textwidth]{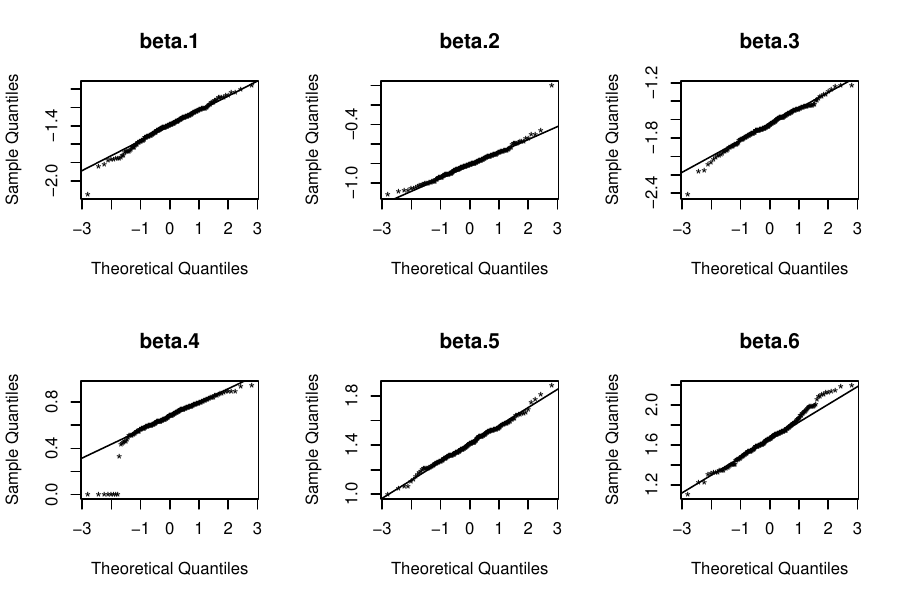}}
             \hfill
\subfloat[MCP]
         {\includegraphics[width=0.5\textwidth]{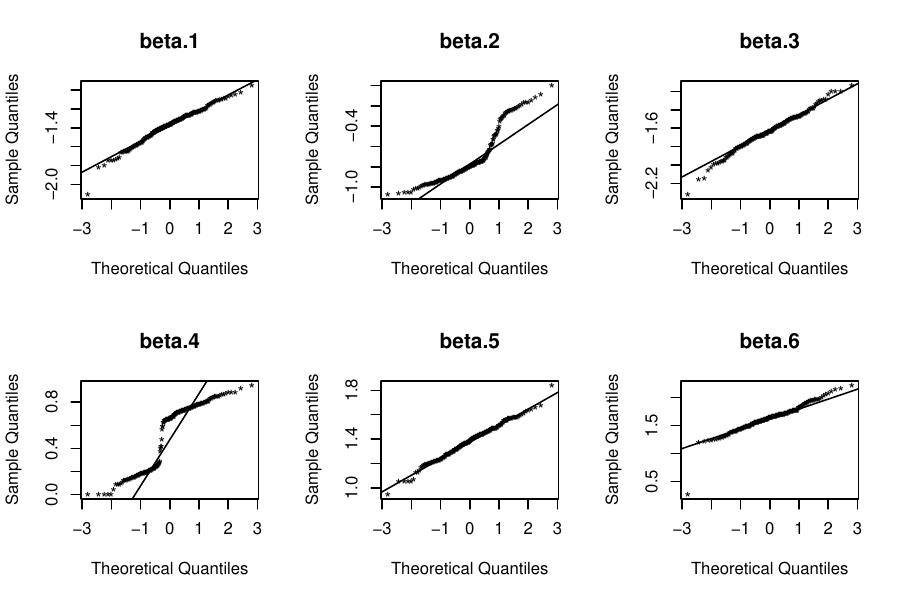}}
\subfloat[SCAD]
        {\includegraphics[width=0.5\textwidth]{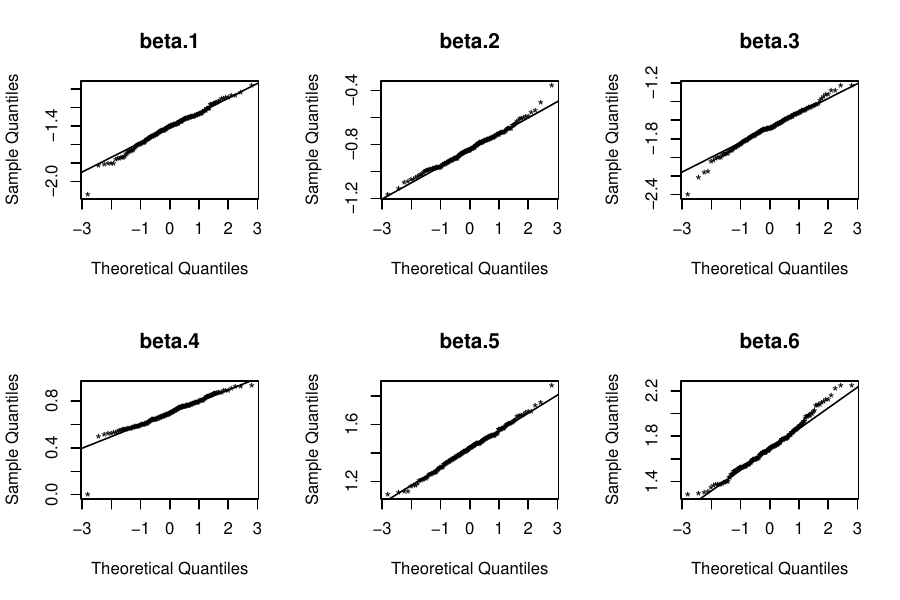}}
\caption[QQplot]
        {Normal Q-Q plots of the estimates of the six nonzero coefficients in the scenario with $p = 3000$, $n = 500$, and $\rho = 0$}
    \end{figure}

\begin{figure*}\centering
\subfloat[Lasso]
        {\includegraphics[width=0.5\textwidth]{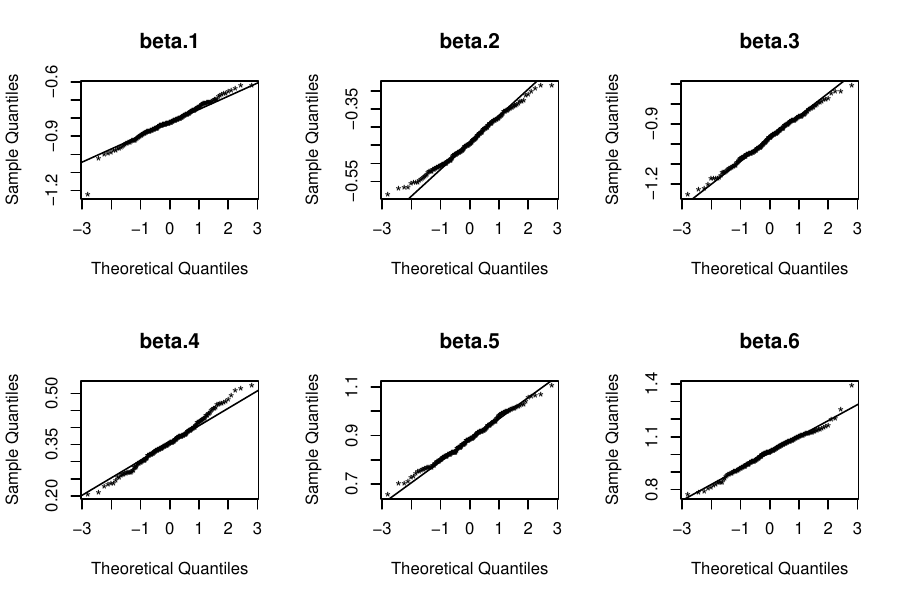}}
\subfloat[Adaptive Lasso \label{fig:mean and std of net24}]
         {\includegraphics[width=0.5\textwidth]{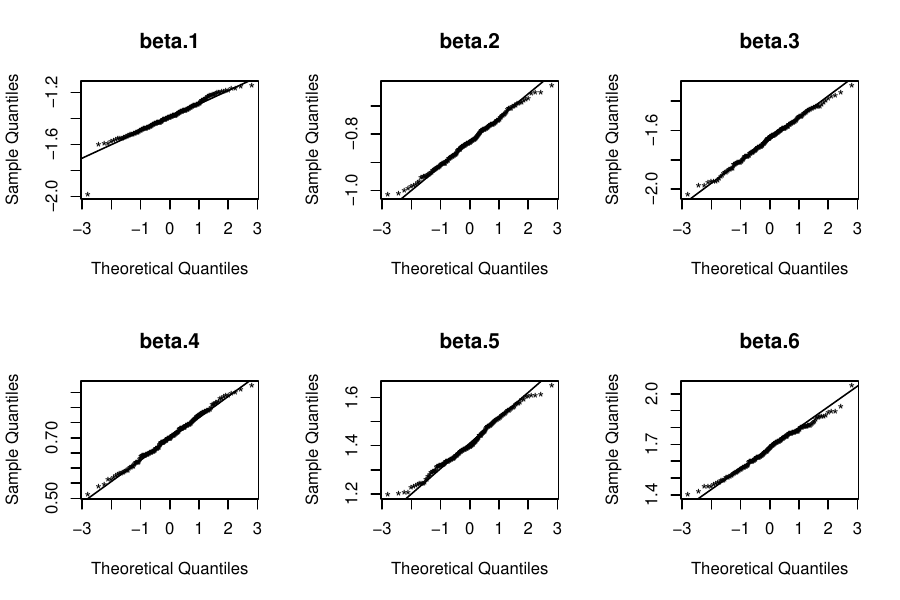}}
             \hfill
\subfloat[MCP]
         {\includegraphics[width=0.5\textwidth]{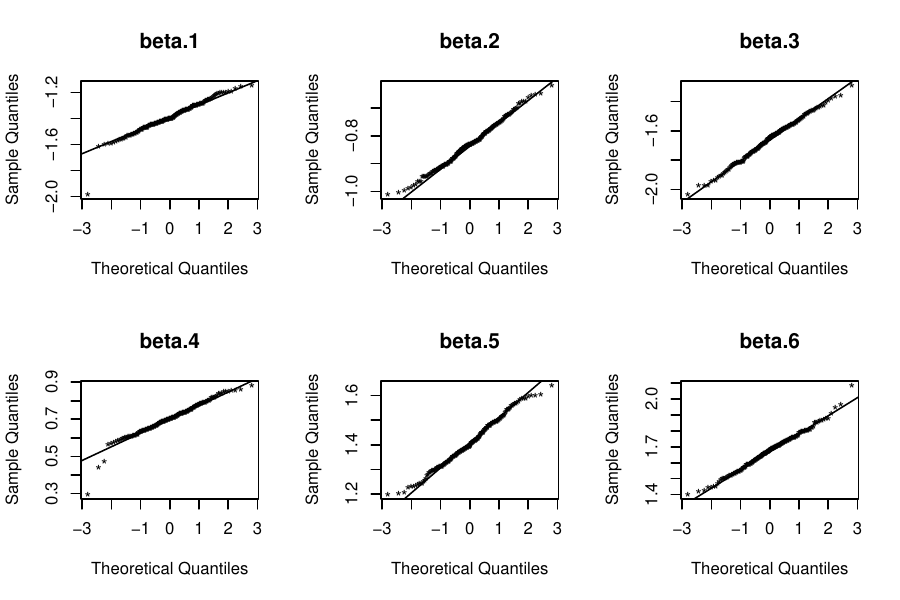}}
\subfloat[SCAD]
        {\includegraphics[width=0.5\textwidth]{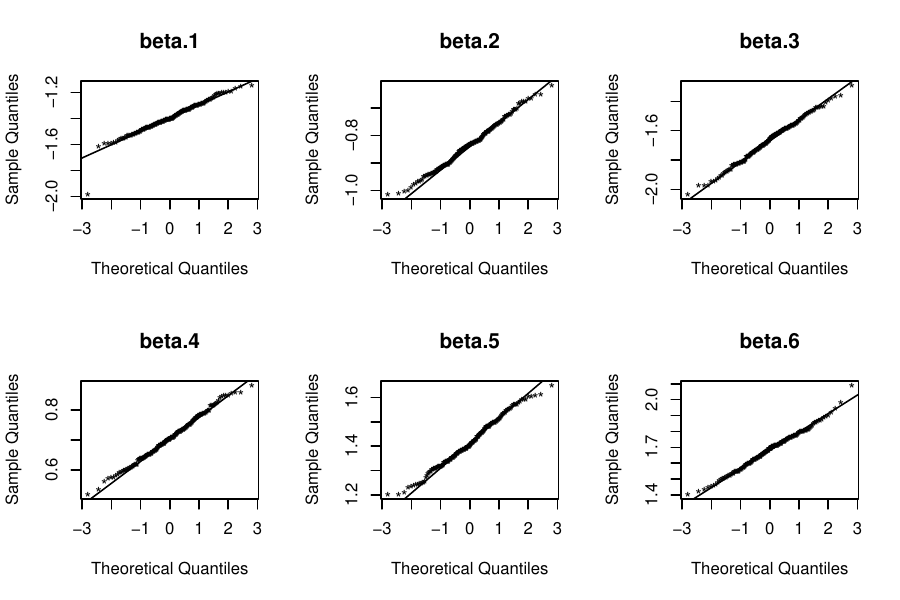}}
\caption[QQplot]
        {Normal Q-Q plots of the estimates of the six nonzero coefficients in the scenario with $p = 3000$, $n = 1000$, and $\rho = 0$}
    \end{figure*}

\begin{figure*}\centering
\subfloat[Lasso]
        {\includegraphics[width=0.5\textwidth]{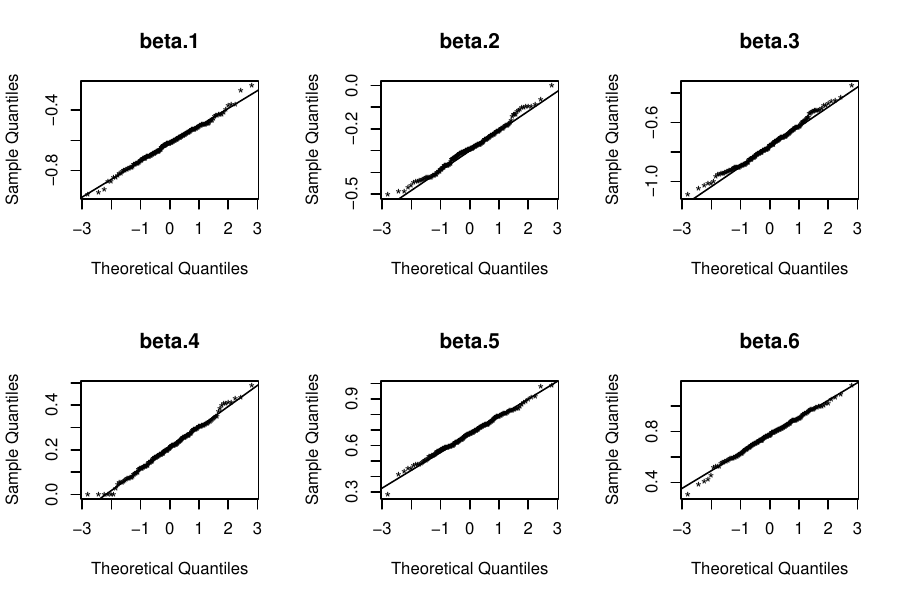}}
\subfloat[Adaptive Lasso \label{fig:mean and std of net24}]
         {\includegraphics[width=0.5\textwidth]{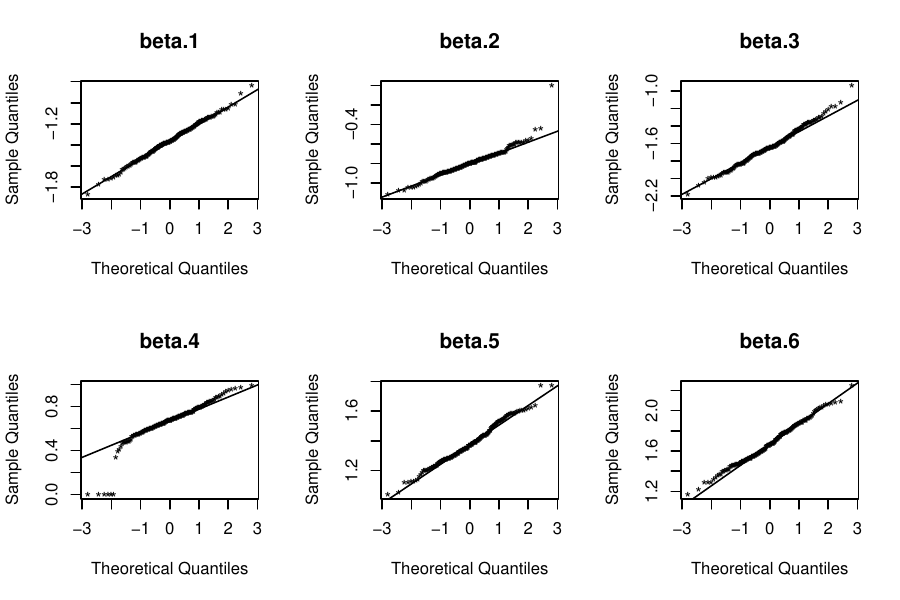}}
             \hfill
\subfloat[MCP]
         {\includegraphics[width=0.5\textwidth]{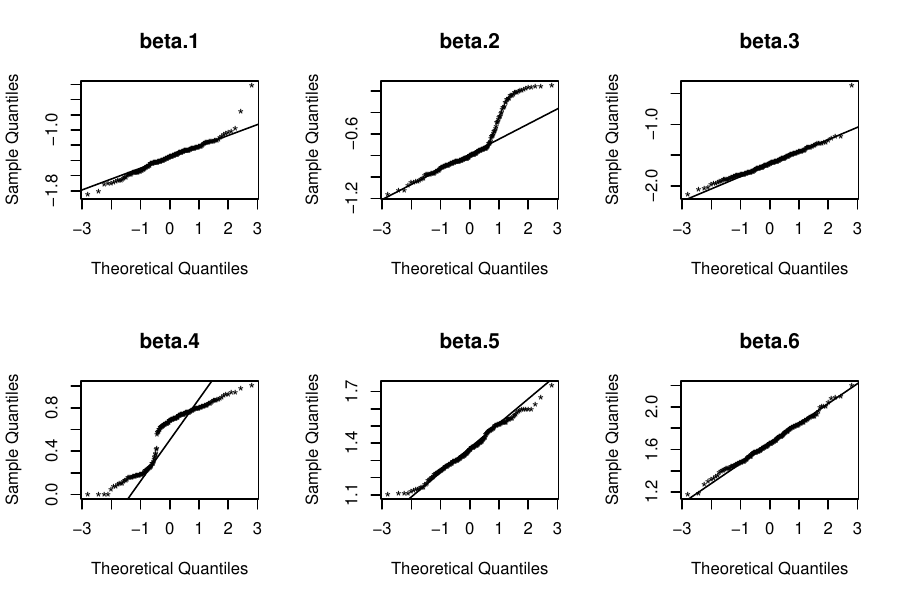}}
\subfloat[SCAD]
        {\includegraphics[width=0.5\textwidth]{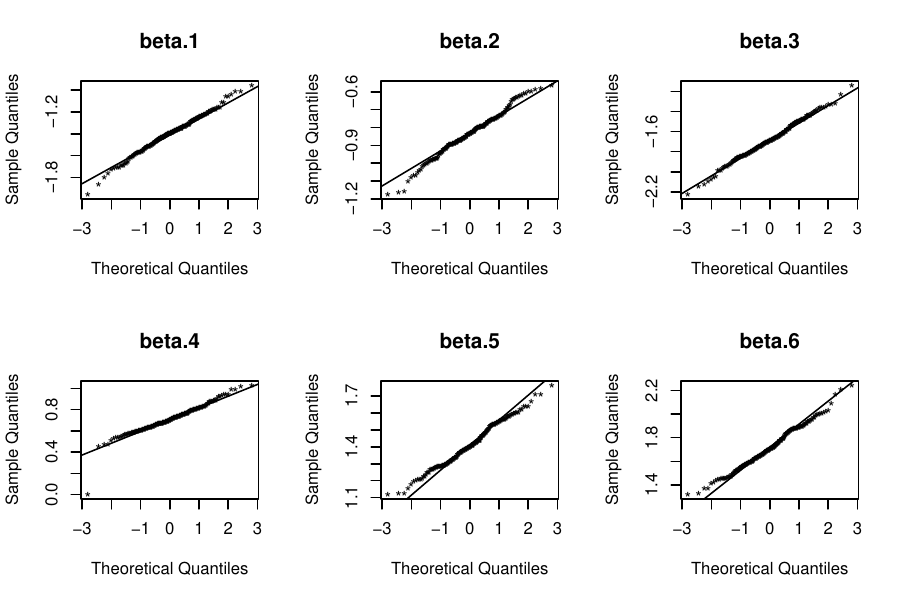}}
\caption[QQplot]
        {Normal Q-Q plots of the estimates of the six nonzero coefficients in the scenario with $p = 3000$, $n = 500$, and $\rho = 0.8$}
    \end{figure*}

\begin{figure*}\centering
\subfloat[Lasso]
        {\includegraphics[width=0.5\textwidth]{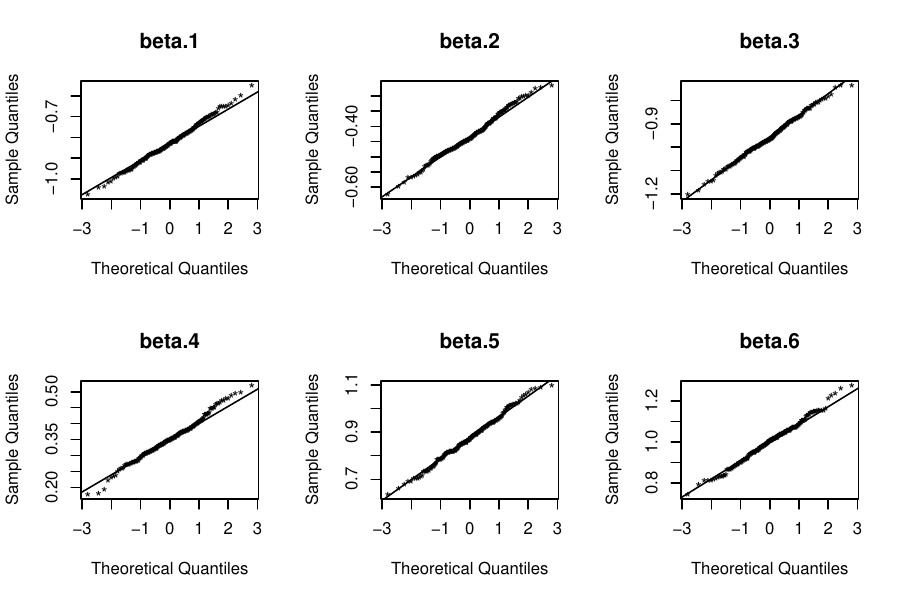}}
\subfloat[Adaptive Lasso \label{fig:mean and std of net24}]
         {\includegraphics[width=0.5\textwidth]{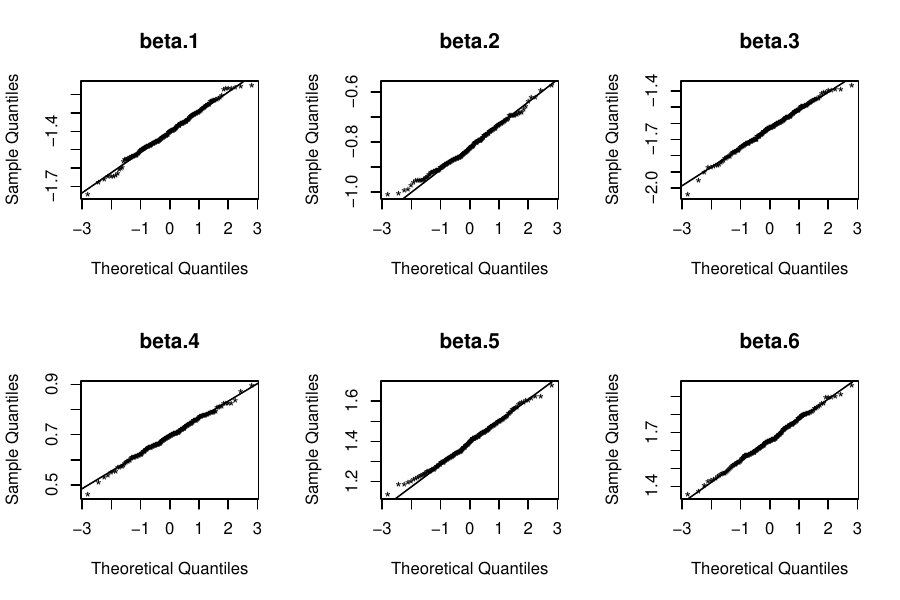}}
             \hfill
\subfloat[MCP]
         {\includegraphics[width=0.5\textwidth]{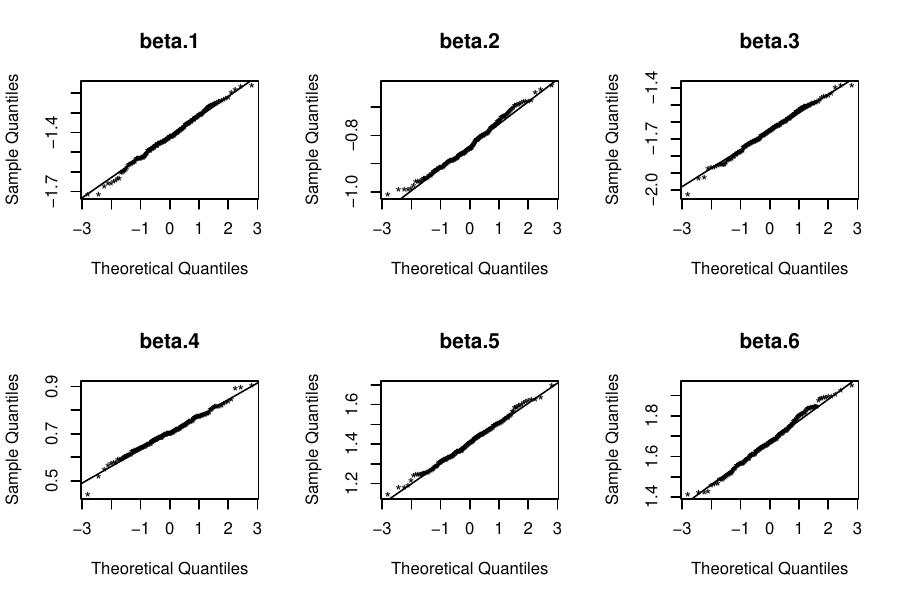}}
\subfloat[SCAD]
        {\includegraphics[width=0.5\textwidth]{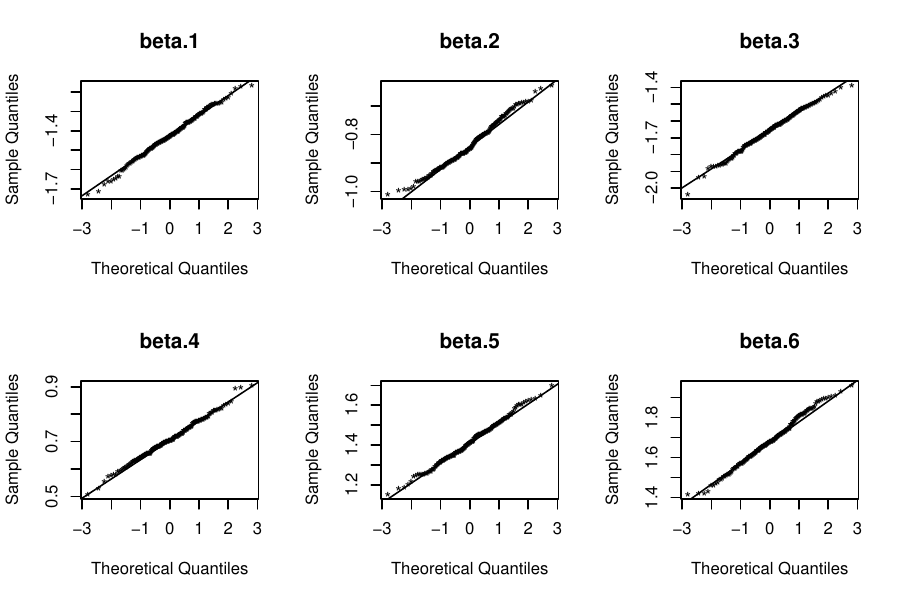}}
\caption[QQplot]
        {Normal Q-Q plots of the estimates of the six nonzero coefficients in the scenario with $p = 3000$, $n = 1000$, and $\rho = 0.8$}
    \end{figure*}

\begin{figure*}\centering
\subfloat[Lasso]
        {\includegraphics[width=0.5\textwidth]{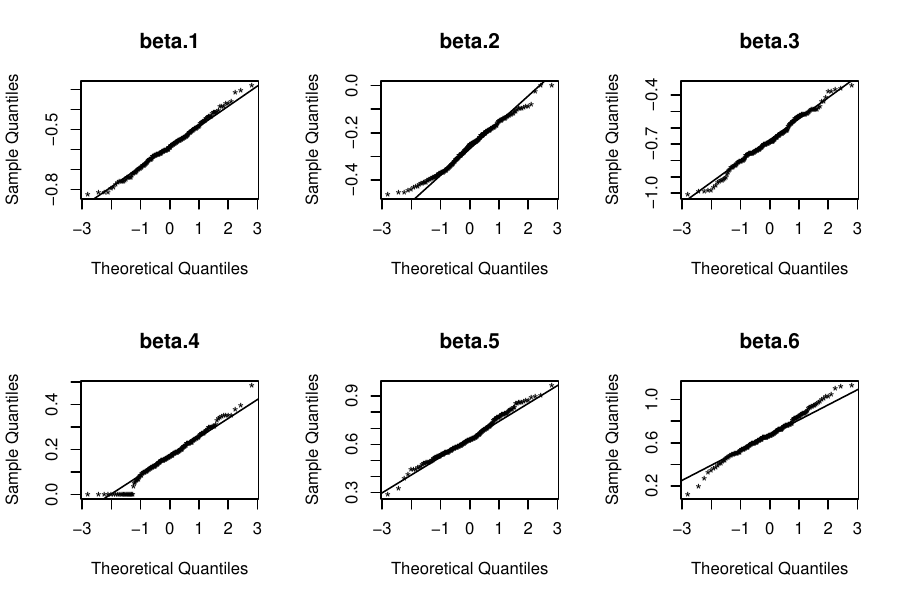}}
\subfloat[Adaptive Lasso \label{fig:mean and std of net24}]
         {\includegraphics[width=0.5\textwidth]{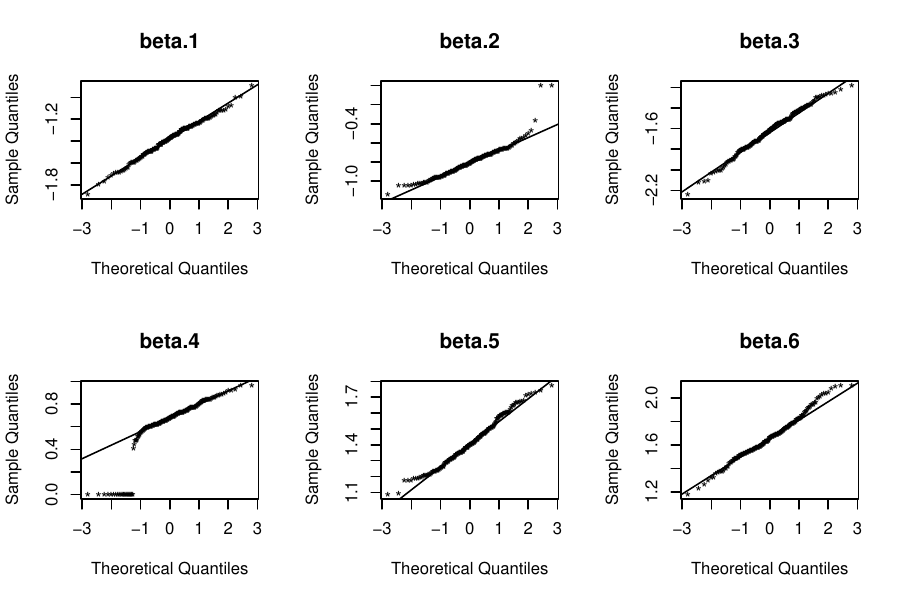}}
             \hfill
\subfloat[MCP]
         {\includegraphics[width=0.5\textwidth]{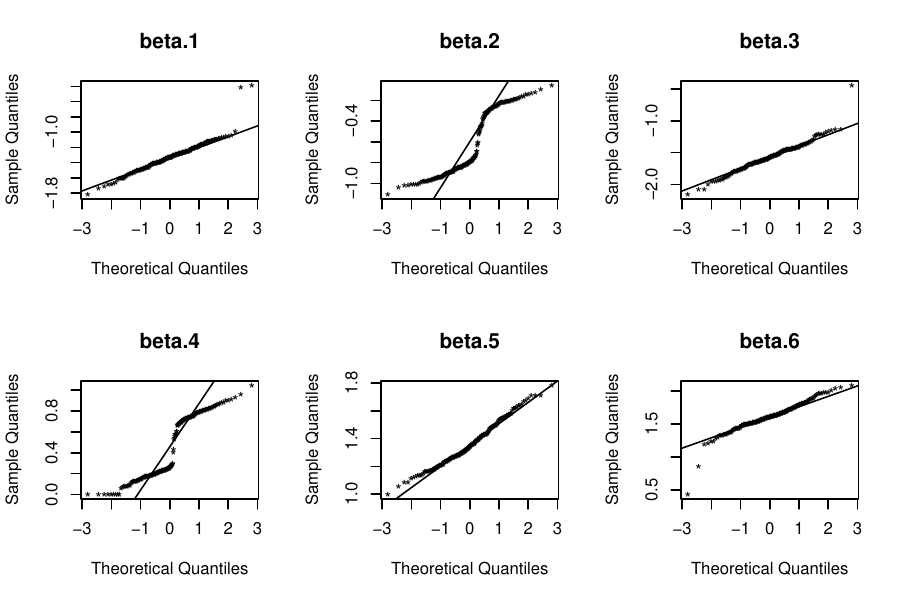}}
\subfloat[SCAD]
        {\includegraphics[width=0.5\textwidth]{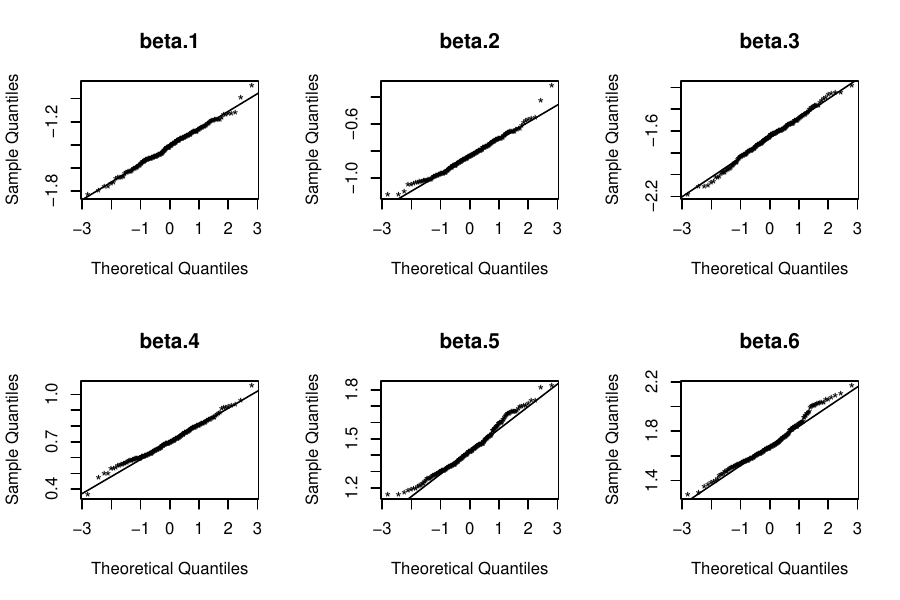}}
\caption[QQplot]
        {Normal Q-Q plots of the estimates of the six nonzero coefficients in the scenario with $p = 10000$, $n = 500$, and $\rho = 0$}
    \end{figure*}

\begin{figure*}\centering
\subfloat[Lasso]
        {\includegraphics[width=0.5\textwidth]{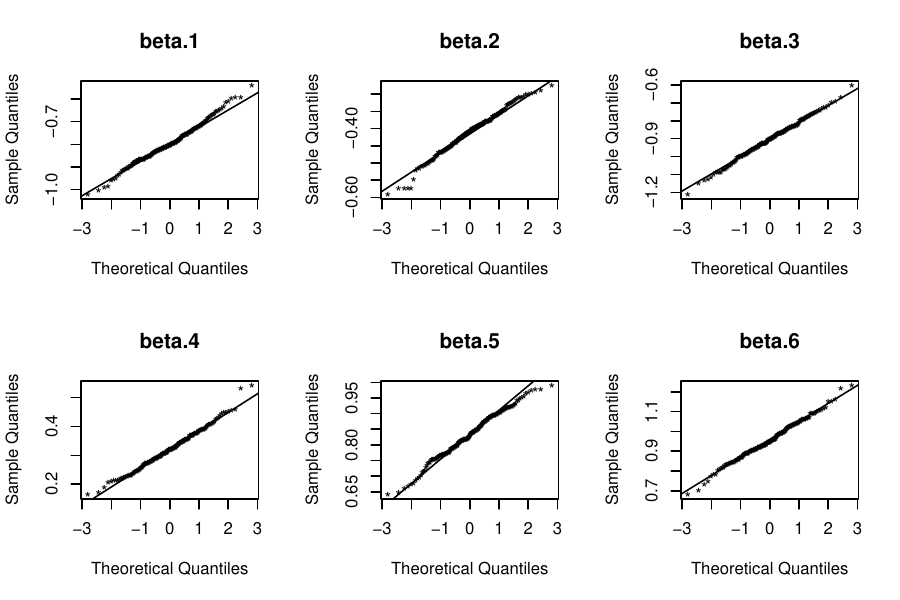}}
\subfloat[Adaptive Lasso \label{fig:mean and std of net24}]
         {\includegraphics[width=0.5\textwidth]{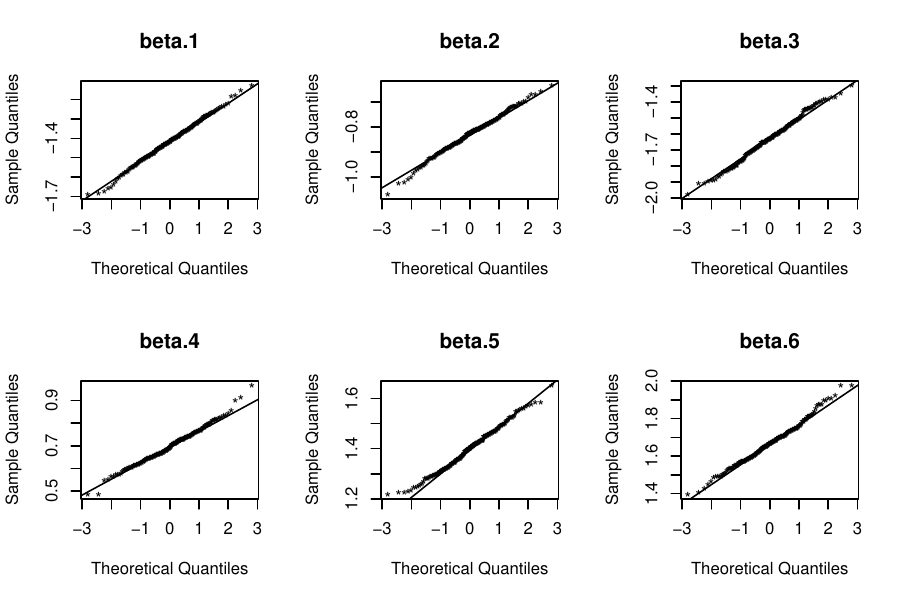}}
             \hfill
\subfloat[MCP]
         {\includegraphics[width=0.5\textwidth]{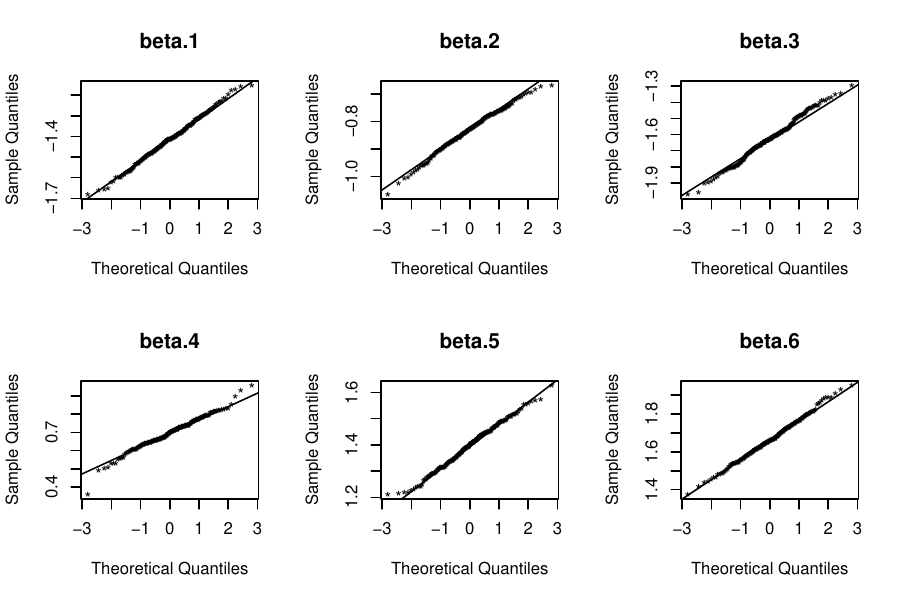}}
\subfloat[SCAD]
        {\includegraphics[width=0.5\textwidth]{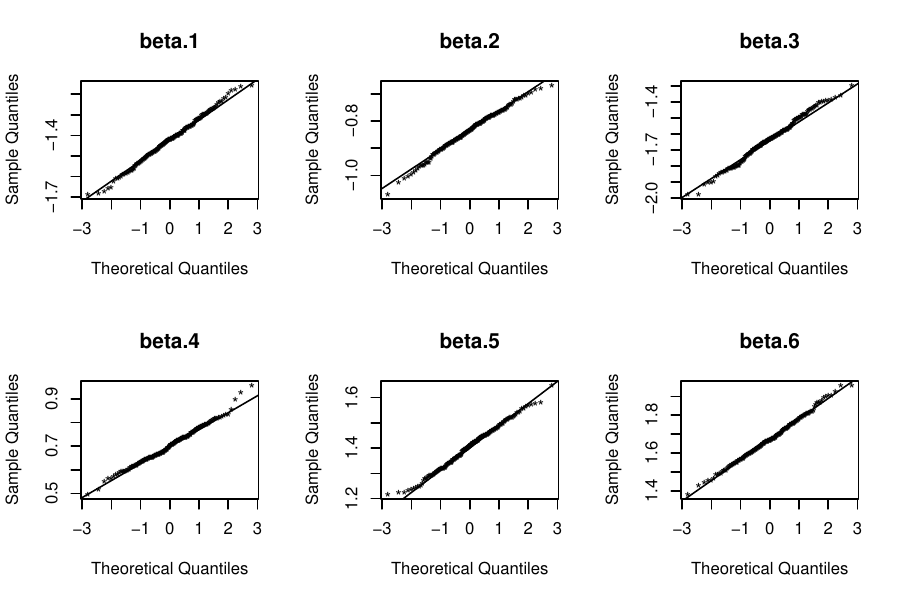}}
\caption[QQplot]
        {Normal Q-Q plots of the estimates of the six nonzero coefficients in the scenario with $p = 10000$, $n = 1000$, and $\rho = 0$}
    \end{figure*}

\begin{figure*}\centering
\subfloat[Lasso]
        {\includegraphics[width=0.5\textwidth]{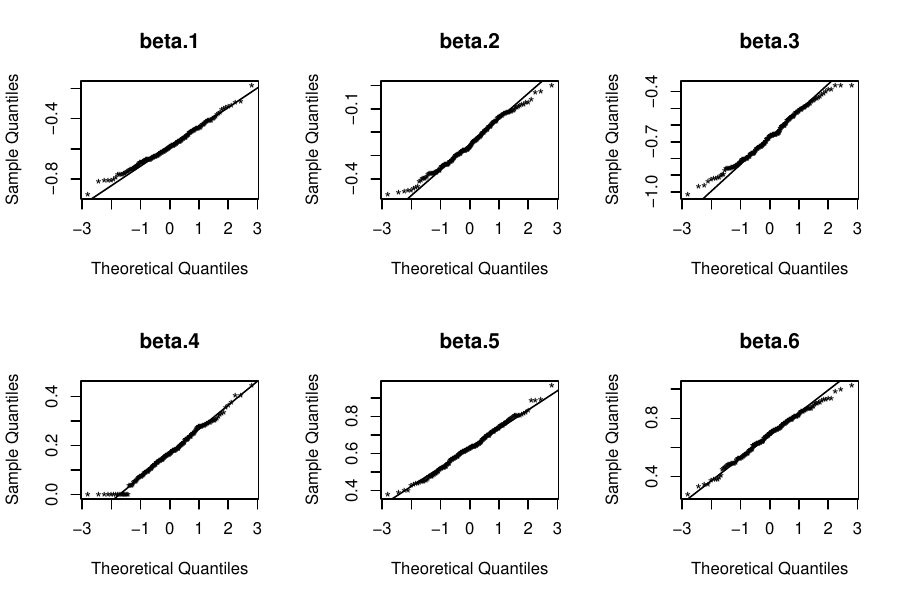}}
\subfloat[Adaptive Lasso \label{fig:mean and std of net24}]
         {\includegraphics[width=0.5\textwidth]{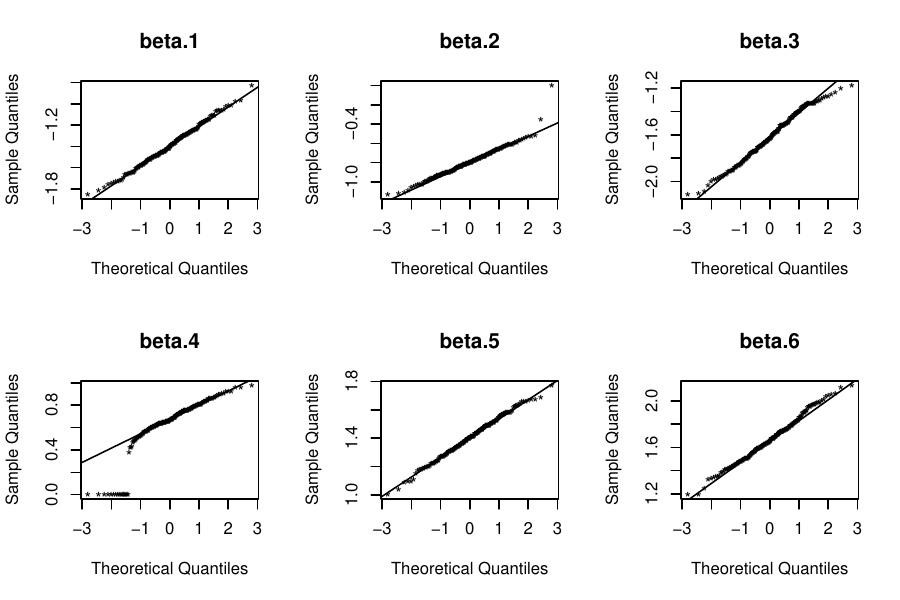}}
             \hfill
\subfloat[MCP]
         {\includegraphics[width=0.5\textwidth]{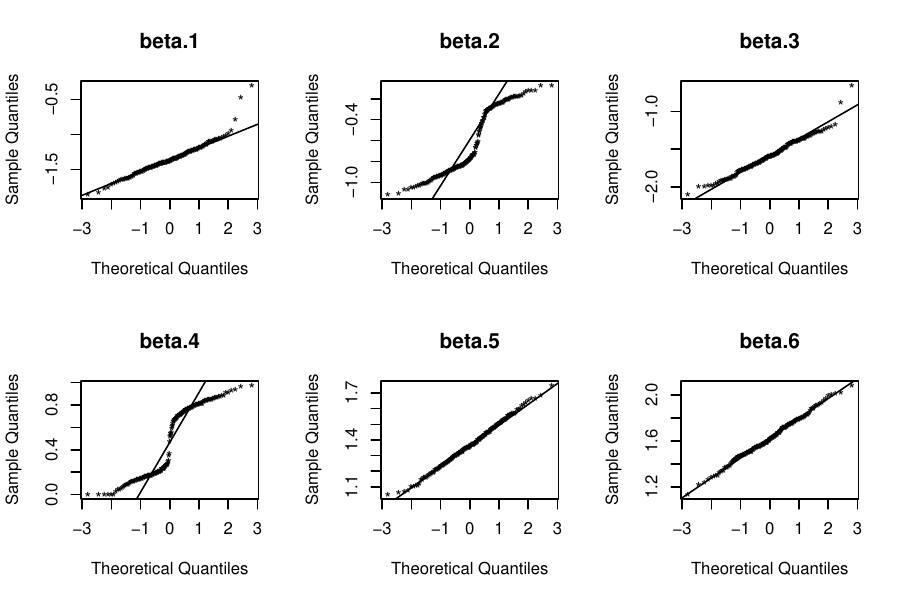}}
\subfloat[SCAD]
        {\includegraphics[width=0.5\textwidth]{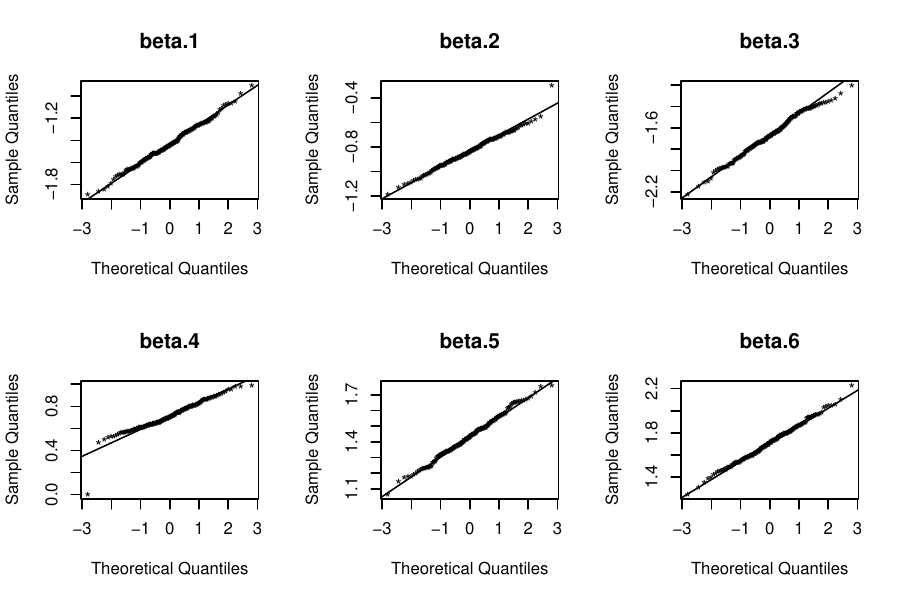}}
\caption[QQplot]
        {Normal Q-Q plots of the estimates of the six nonzero coefficients in the scenario with $p = 10000$, $n = 500$, and $\rho = 0.8$}
    \end{figure*}

\begin{figure*}\centering
\subfloat[Lasso]
        {\includegraphics[width=0.5\textwidth]{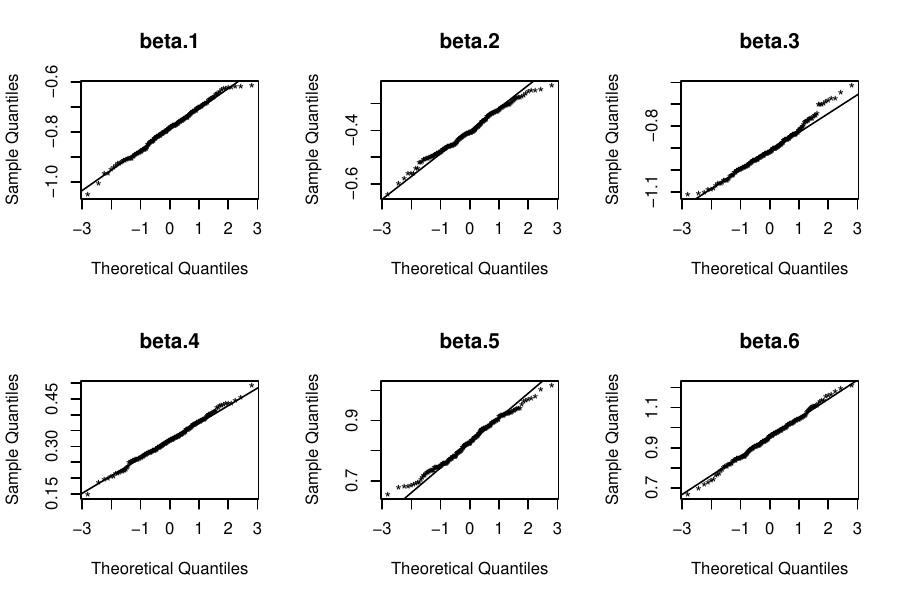}}
\subfloat[Adaptive Lasso \label{fig:mean and std of net24}]
         {\includegraphics[width=0.5\textwidth]{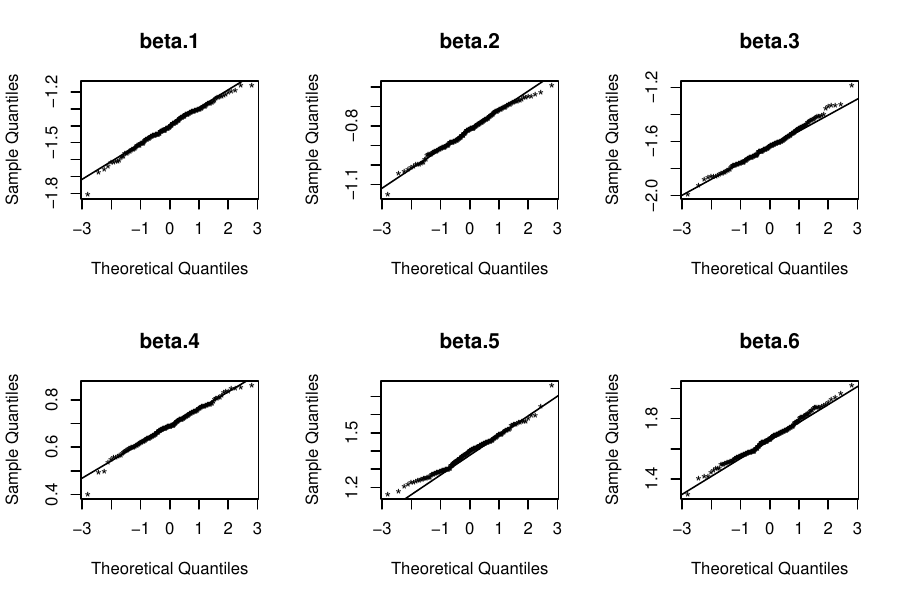}}
             \hfill
\subfloat[MCP]
         {\includegraphics[width=0.5\textwidth]{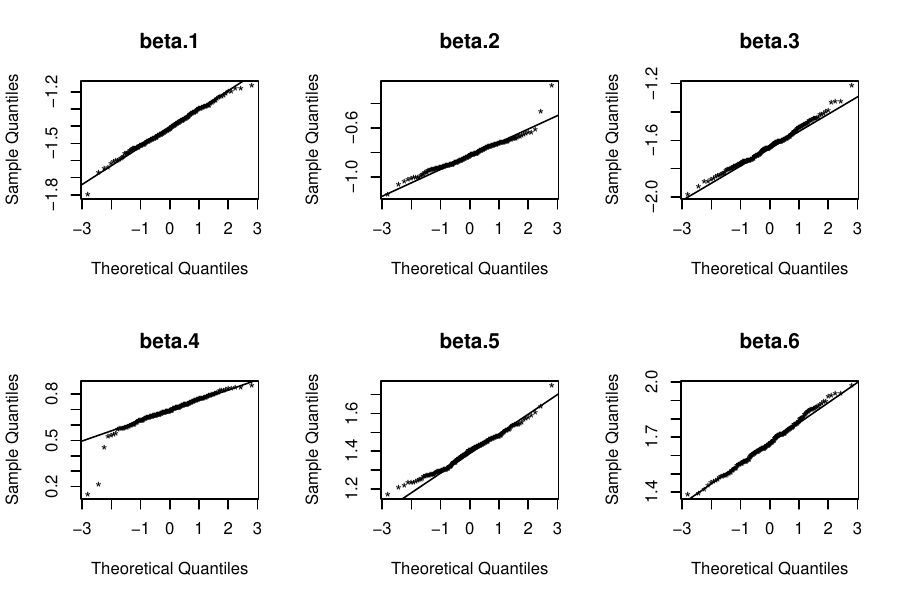}}
\subfloat[SCAD]
        {\includegraphics[width=0.5\textwidth]{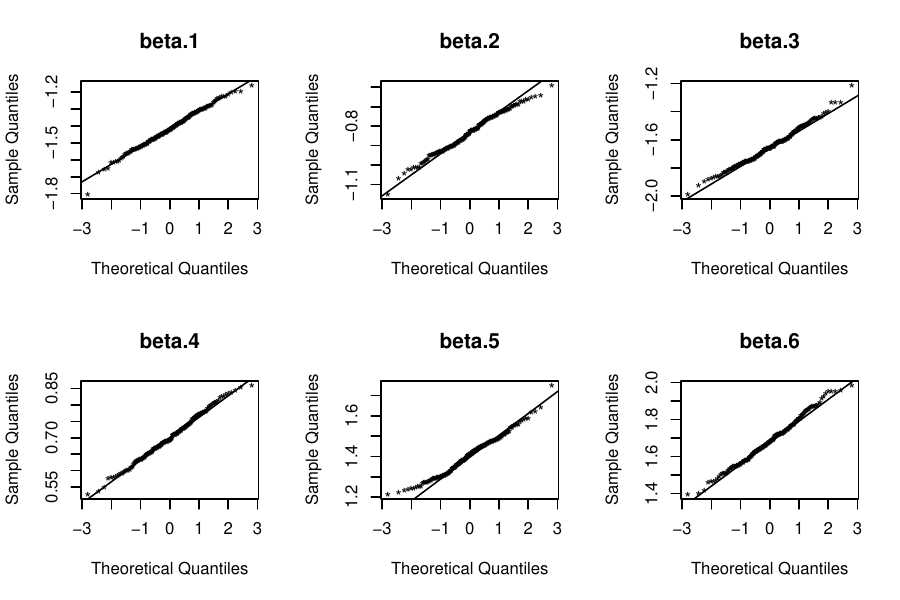}}
\caption[QQplot]
        {Normal Q-Q plots of the estimates of the six nonzero coefficients in the scenario with $p = 10000$, $n = 1000$, and $\rho = 0.8$}
    \end{figure*}

\begin{figure*}\centering
\subfloat[Lasso]
        {\includegraphics[width=0.5\textwidth]{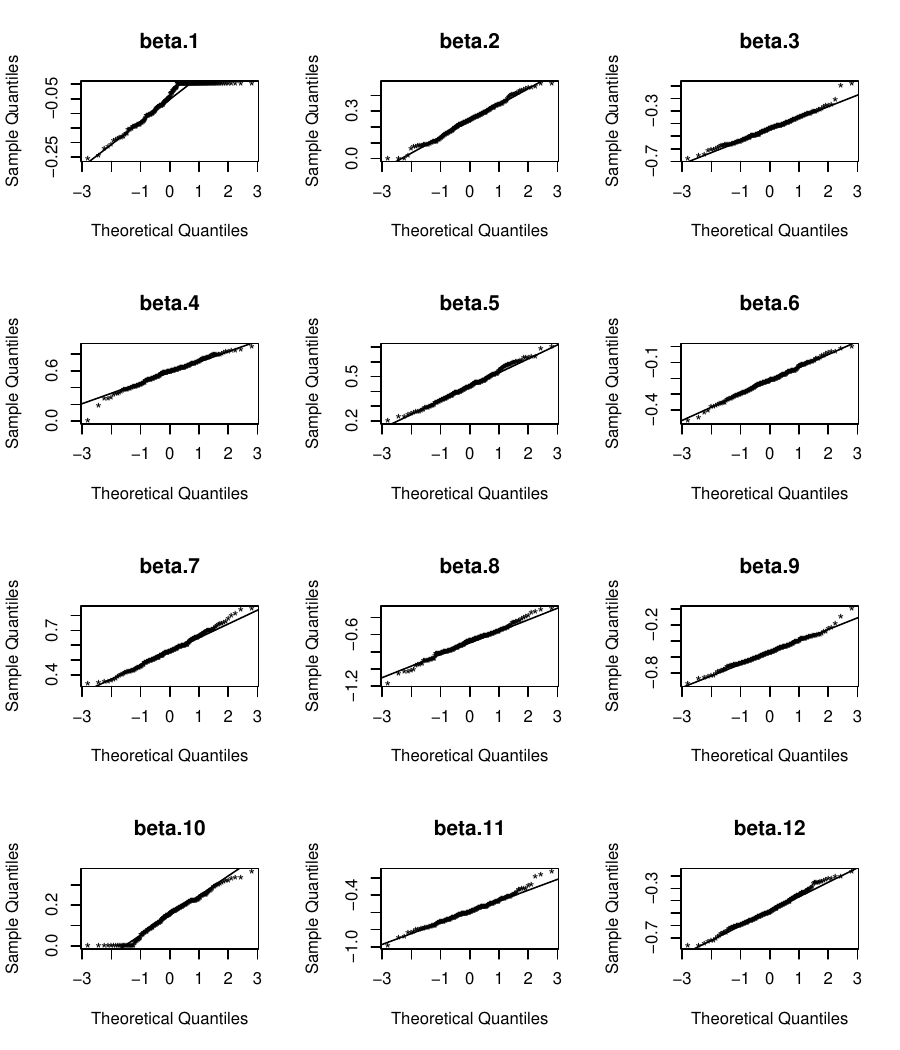}}
\subfloat[Adaptive Lasso \label{fig:mean and std of net24}]
         {\includegraphics[width=0.5\textwidth]{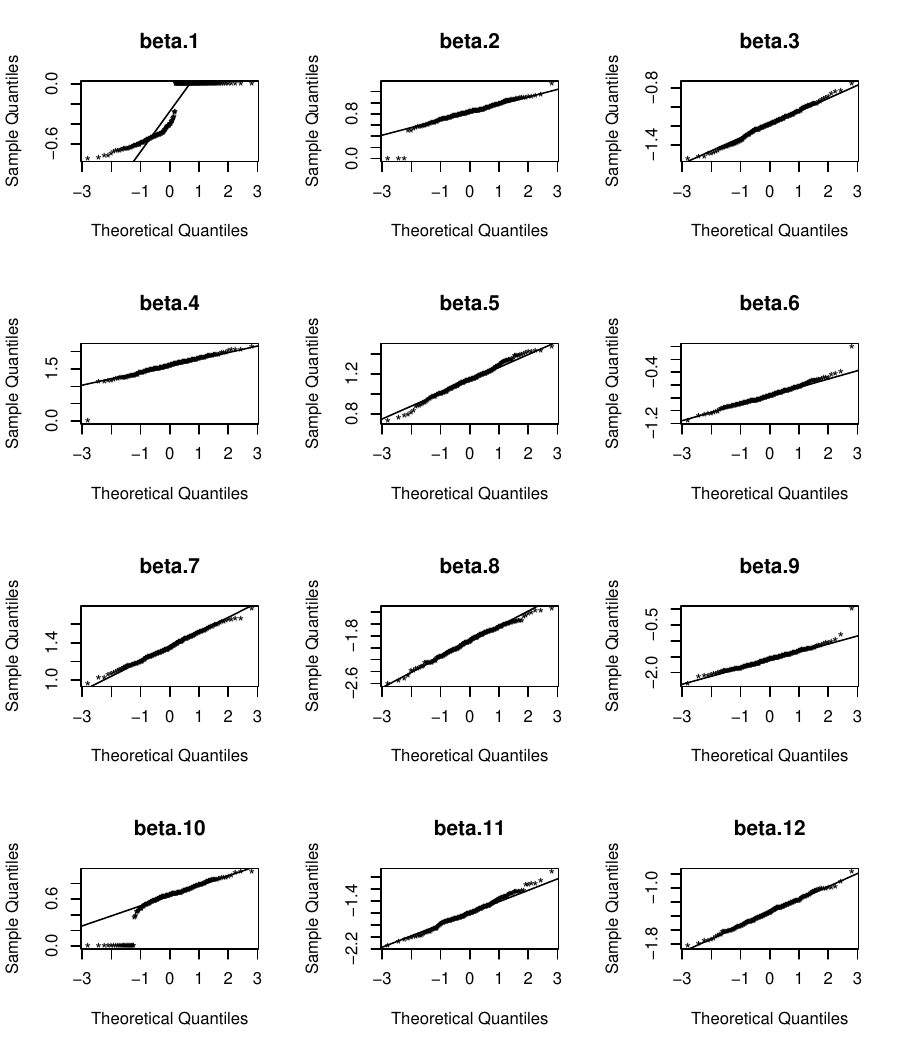}}
             \hfill
\subfloat[MCP]
         {\includegraphics[width=0.5\textwidth]{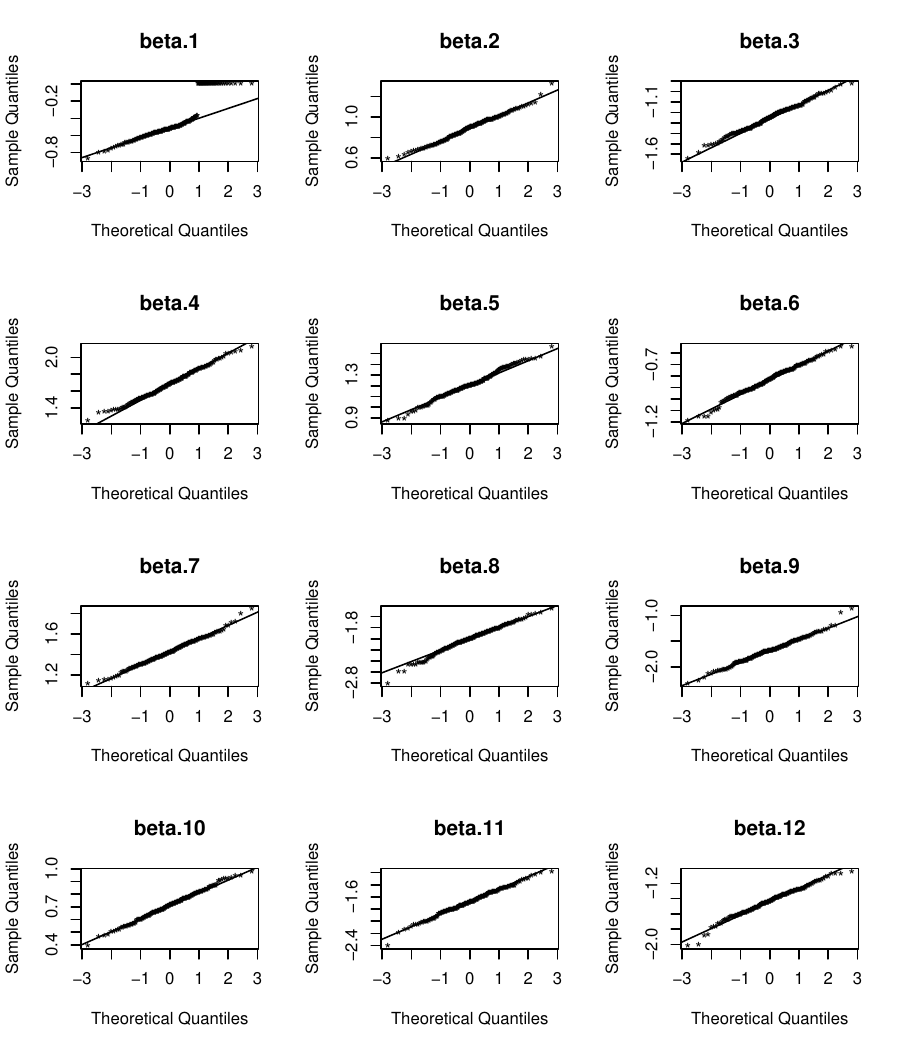}}
\subfloat[SCAD]
        {\includegraphics[width=0.5\textwidth]{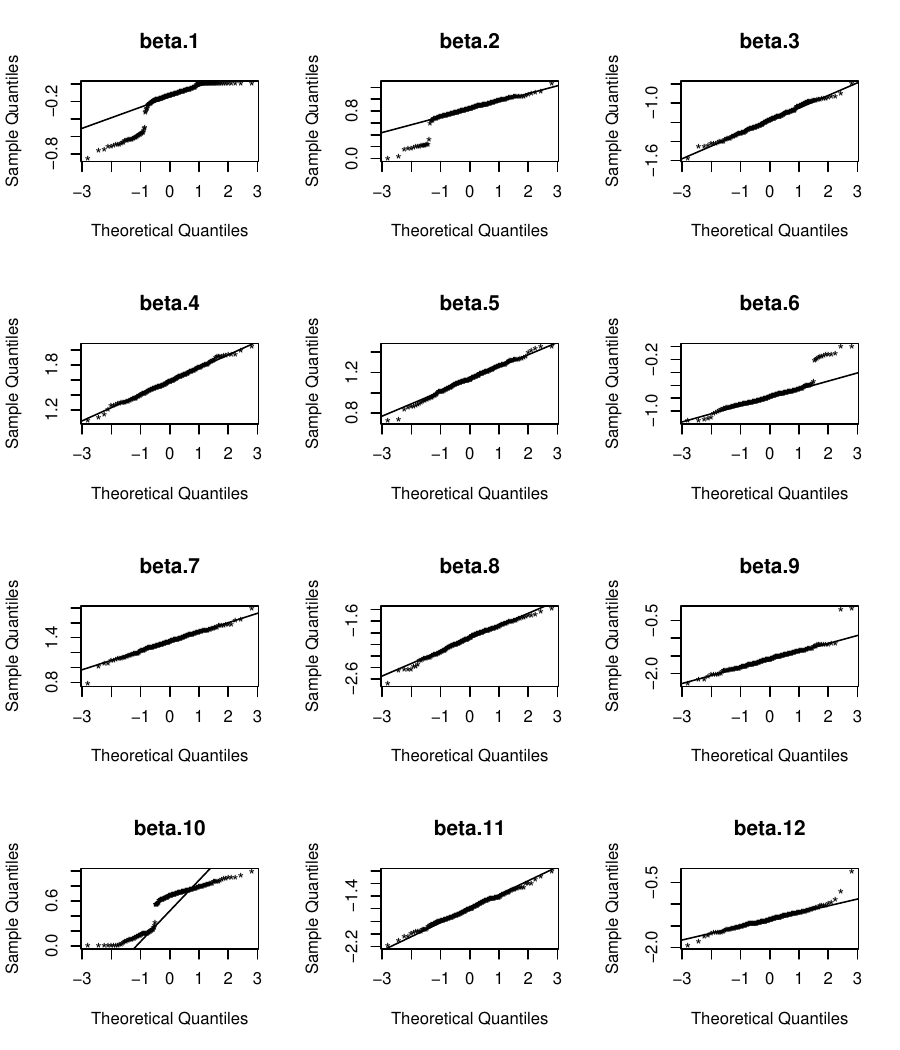}}
\caption[QQplot]
        {Normal Q-Q plots of the estimates of the twelve nonzero coefficients in the scenario with $p = 3000$, $n = 500$, and $\rho = 0$}
    \end{figure*}

\begin{figure*}\centering
\subfloat[Lasso]
        {\includegraphics[width=0.5\textwidth]{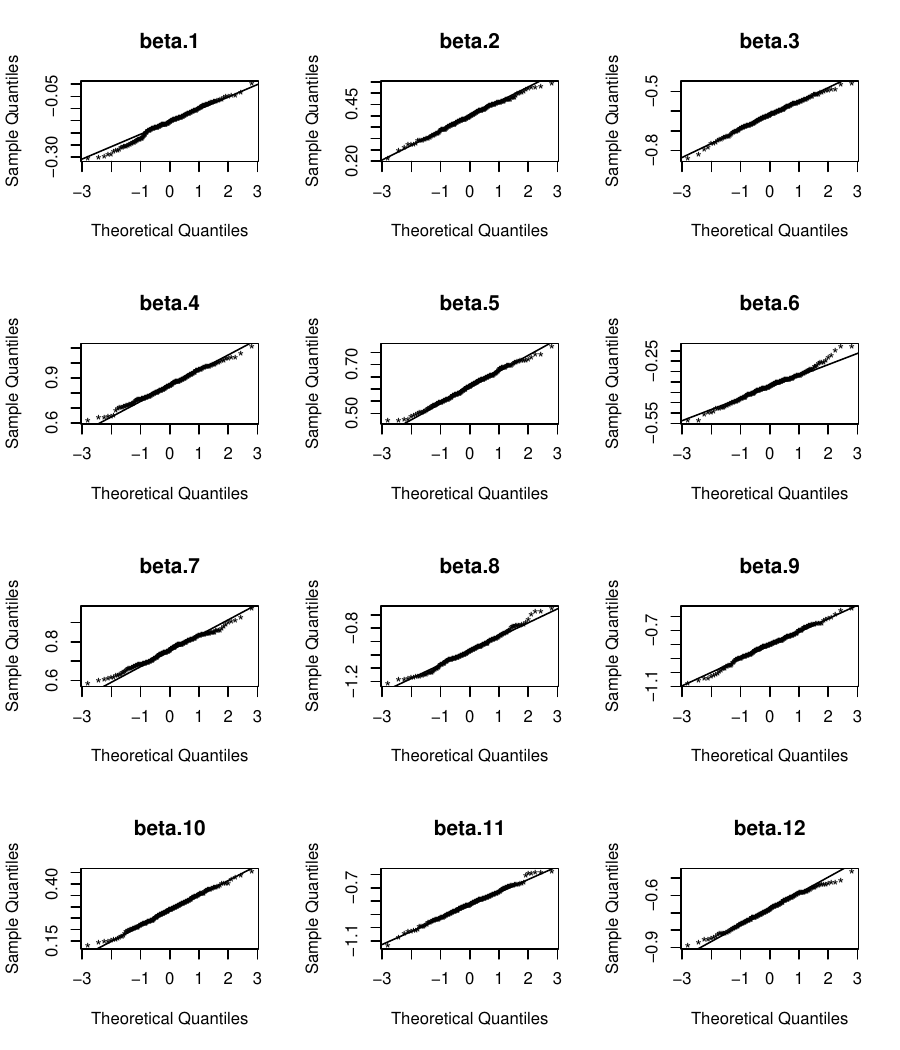}}
\subfloat[Adaptive Lasso \label{fig:mean and std of net24}]
         {\includegraphics[width=0.5\textwidth]{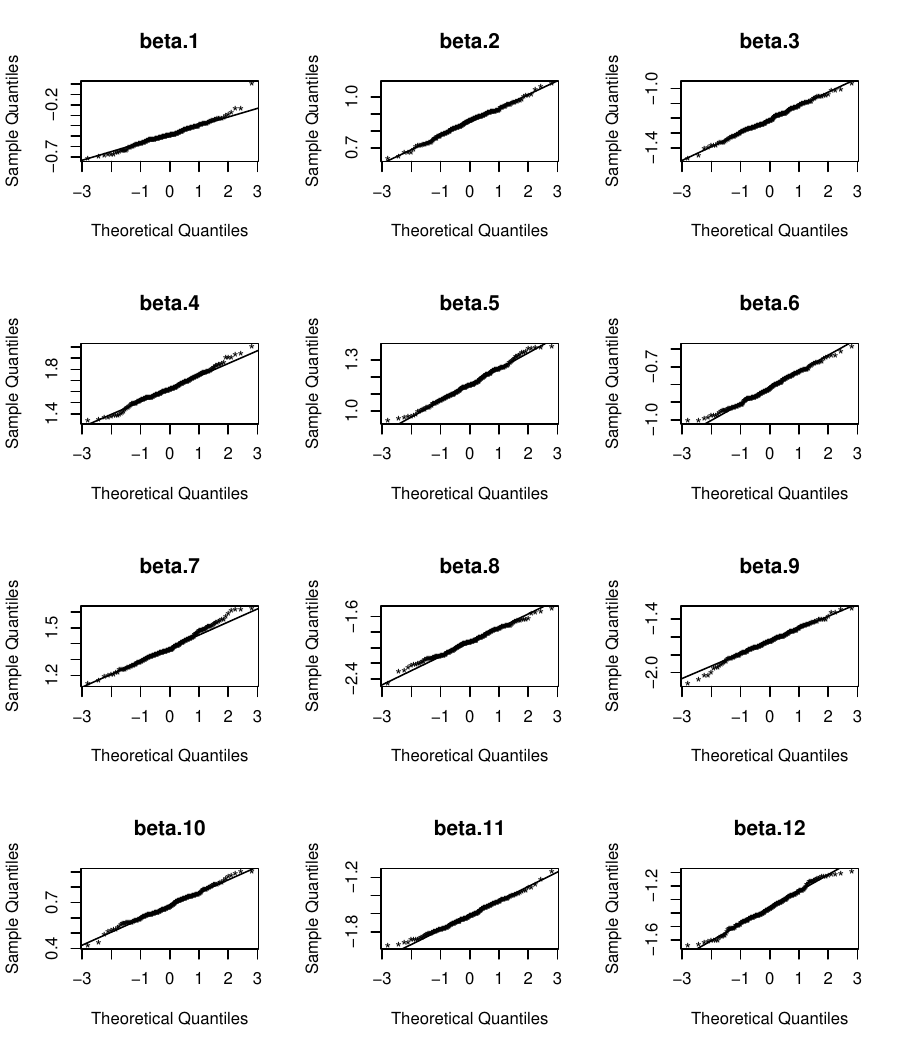}}
             \hfill
\subfloat[MCP]
         {\includegraphics[width=0.5\textwidth]{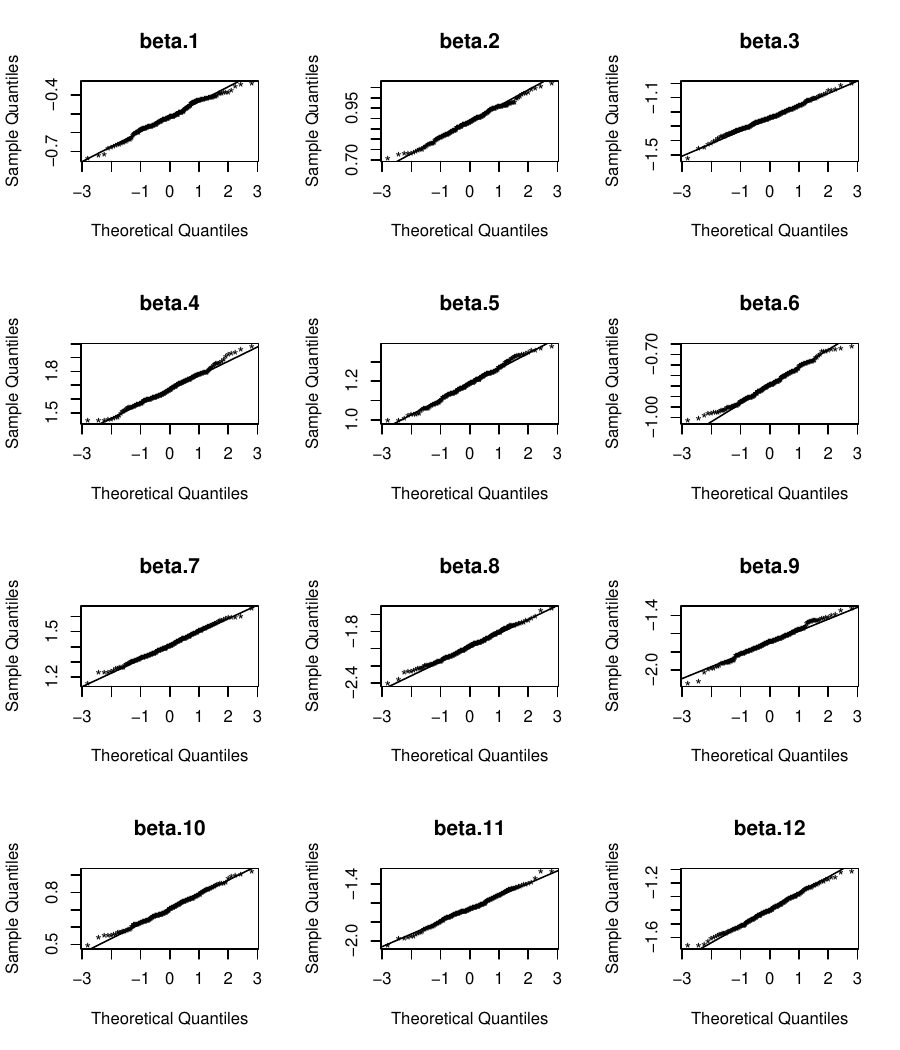}}
\subfloat[SCAD]
        {\includegraphics[width=0.5\textwidth]{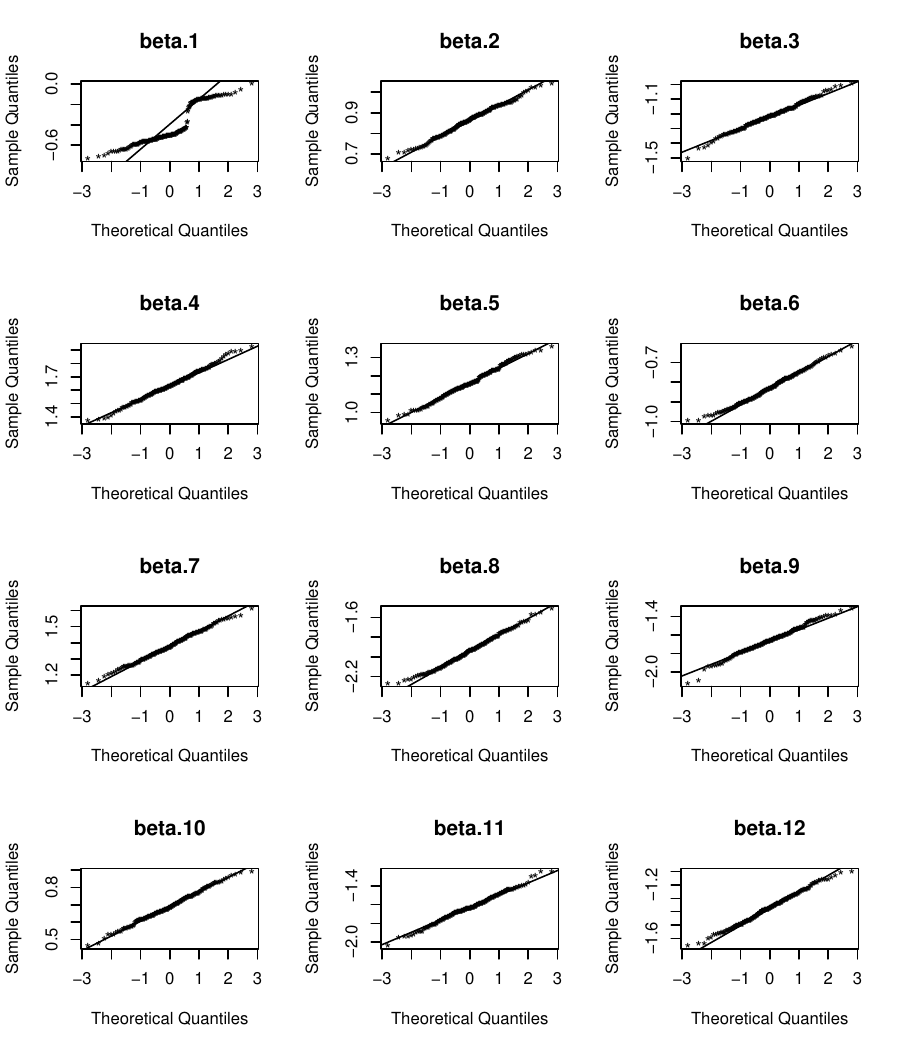}}
\caption[QQplot]
        {Normal Q-Q plots of the estimates of the twelve nonzero coefficients in the scenario with $p = 3000$, $n = 1000$, and $\rho = 0$}
    \end{figure*}

\begin{figure*}\centering
\subfloat[Lasso]
        {\includegraphics[width=0.5\textwidth]{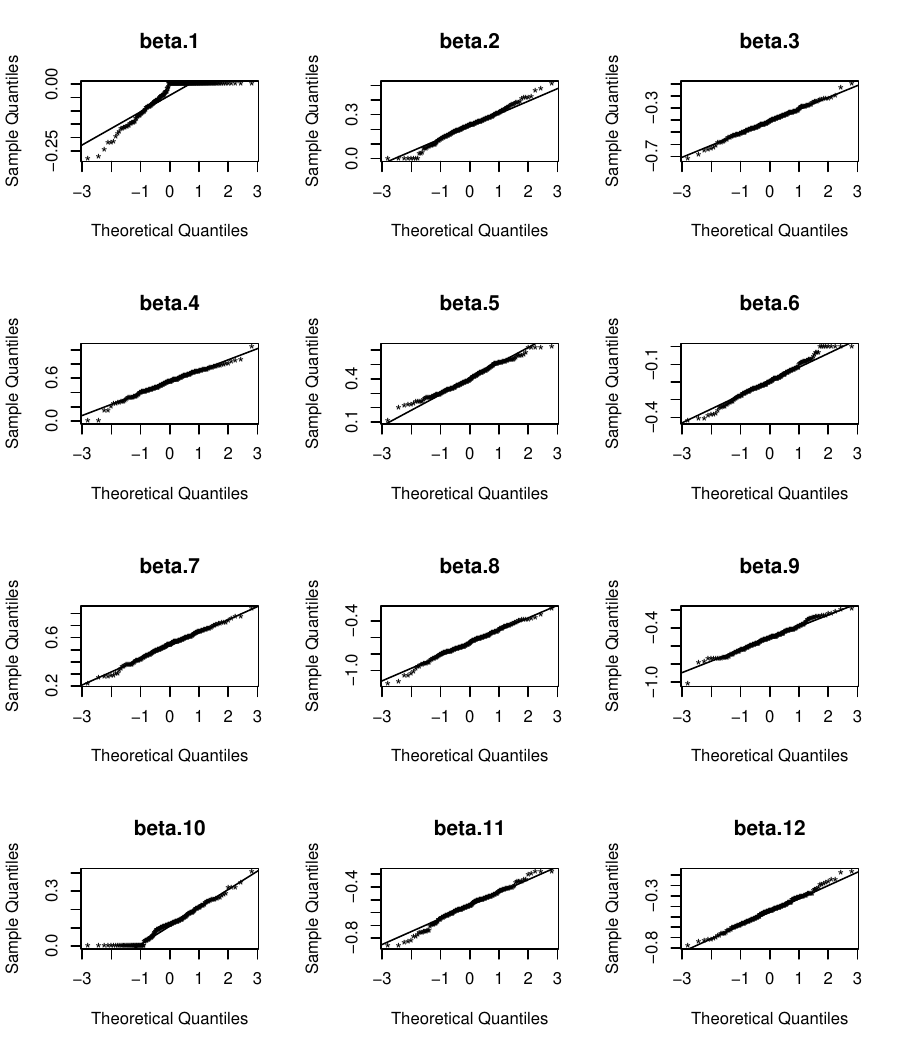}}
\subfloat[Adaptive Lasso \label{fig:mean and std of net24}]
         {\includegraphics[width=0.5\textwidth]{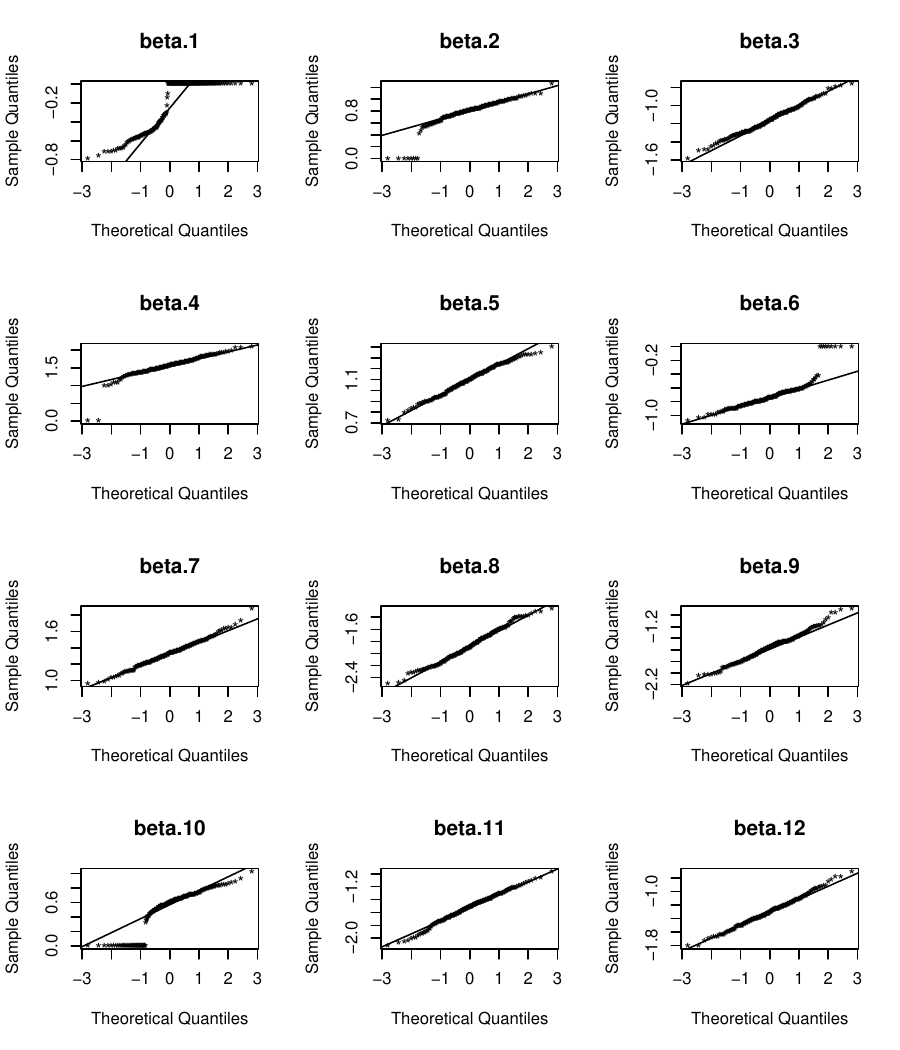}}
             \hfill
\subfloat[MCP]
         {\includegraphics[width=0.5\textwidth]{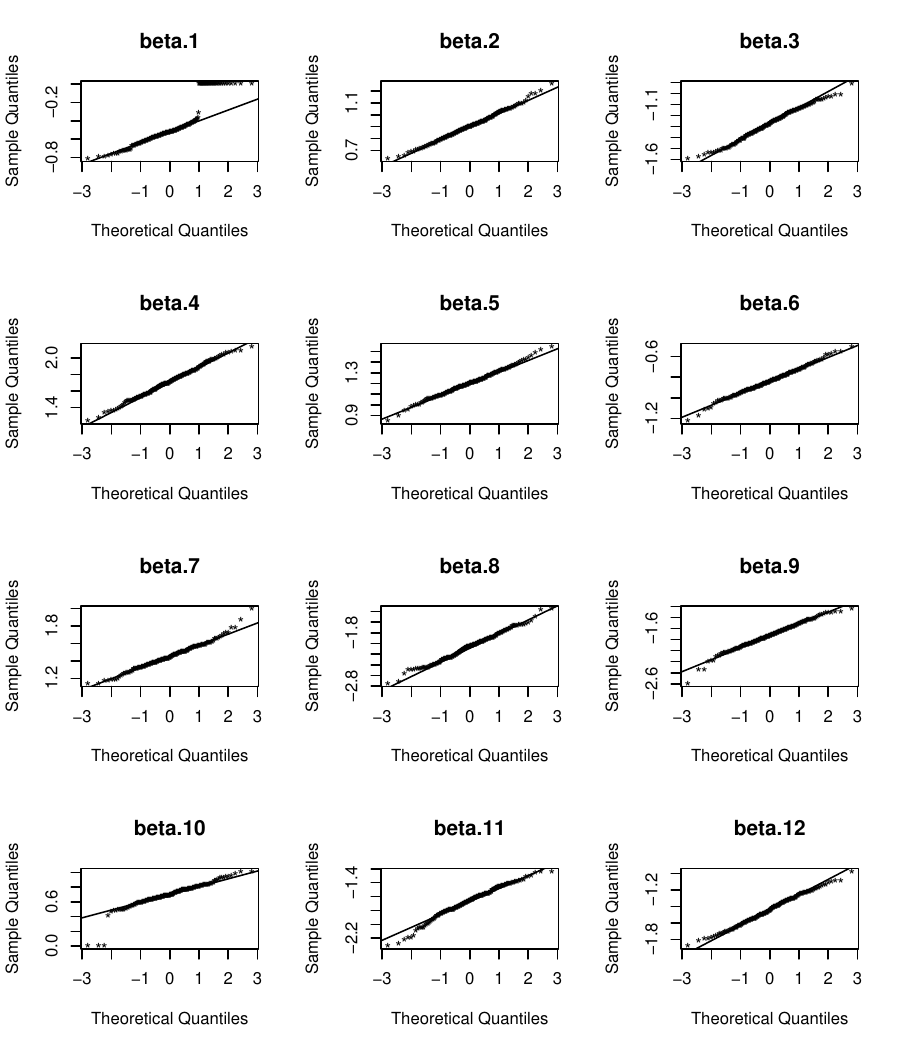}}
\subfloat[SCAD]
        {\includegraphics[width=0.5\textwidth]{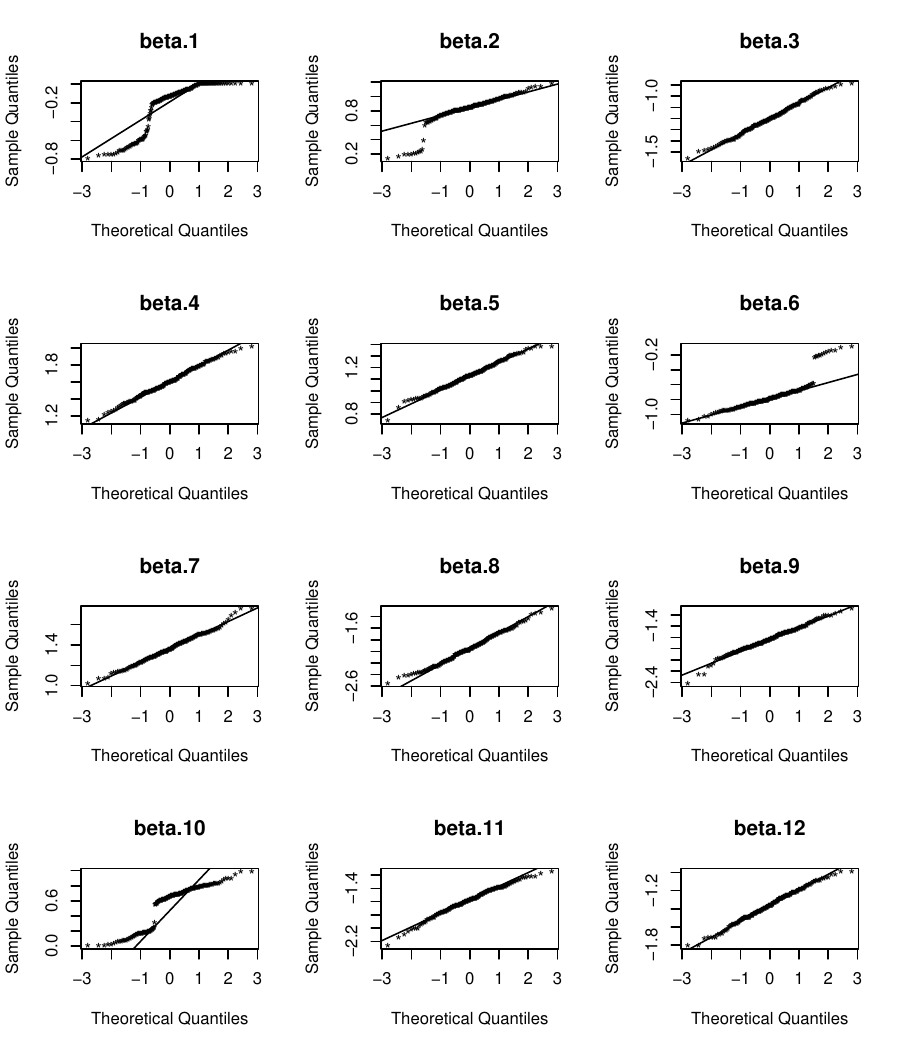}}
\caption[QQplot]
        {Normal Q-Q plots of the estimates of the twelve nonzero coefficients in the scenario with $p = 3000$, $n = 500$, and $\rho = 0.8$}
    \end{figure*}

\begin{figure*}\centering
\subfloat[Lasso]
        {\includegraphics[width=0.5\textwidth]{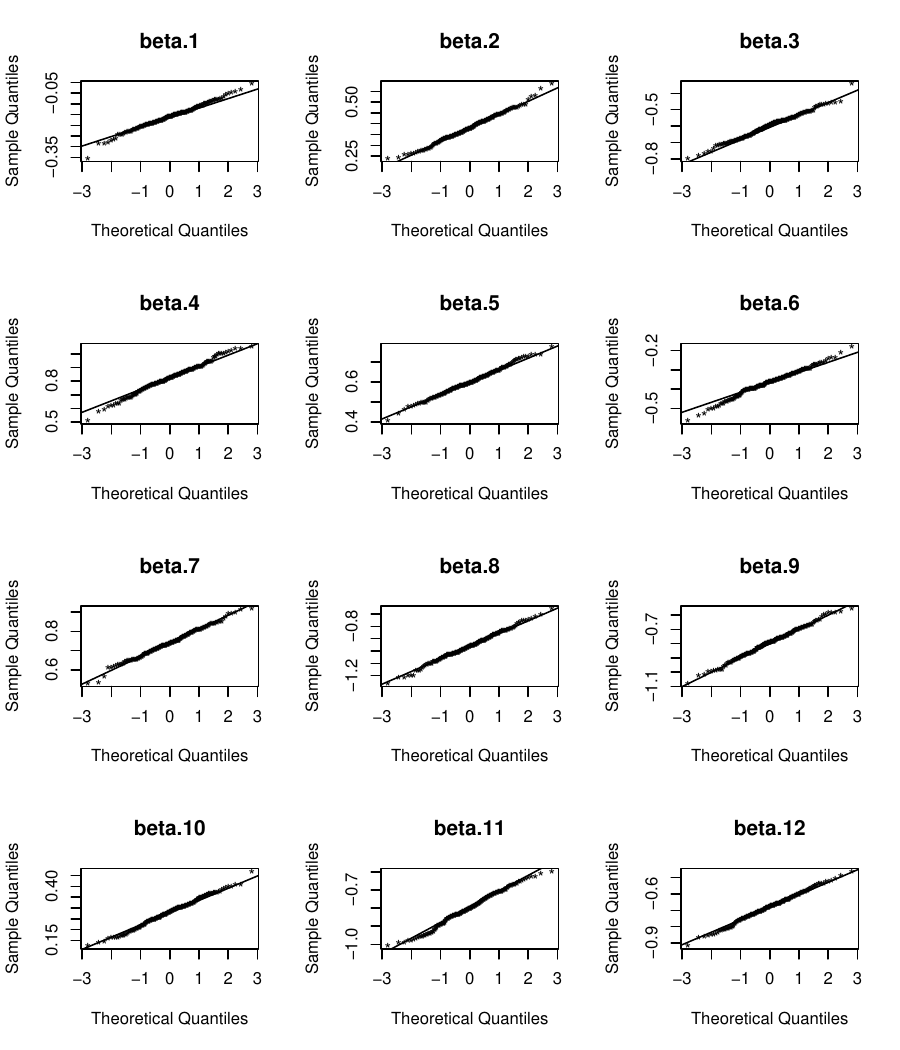}}
\subfloat[Adaptive Lasso \label{fig:mean and std of net24}]
         {\includegraphics[width=0.5\textwidth]{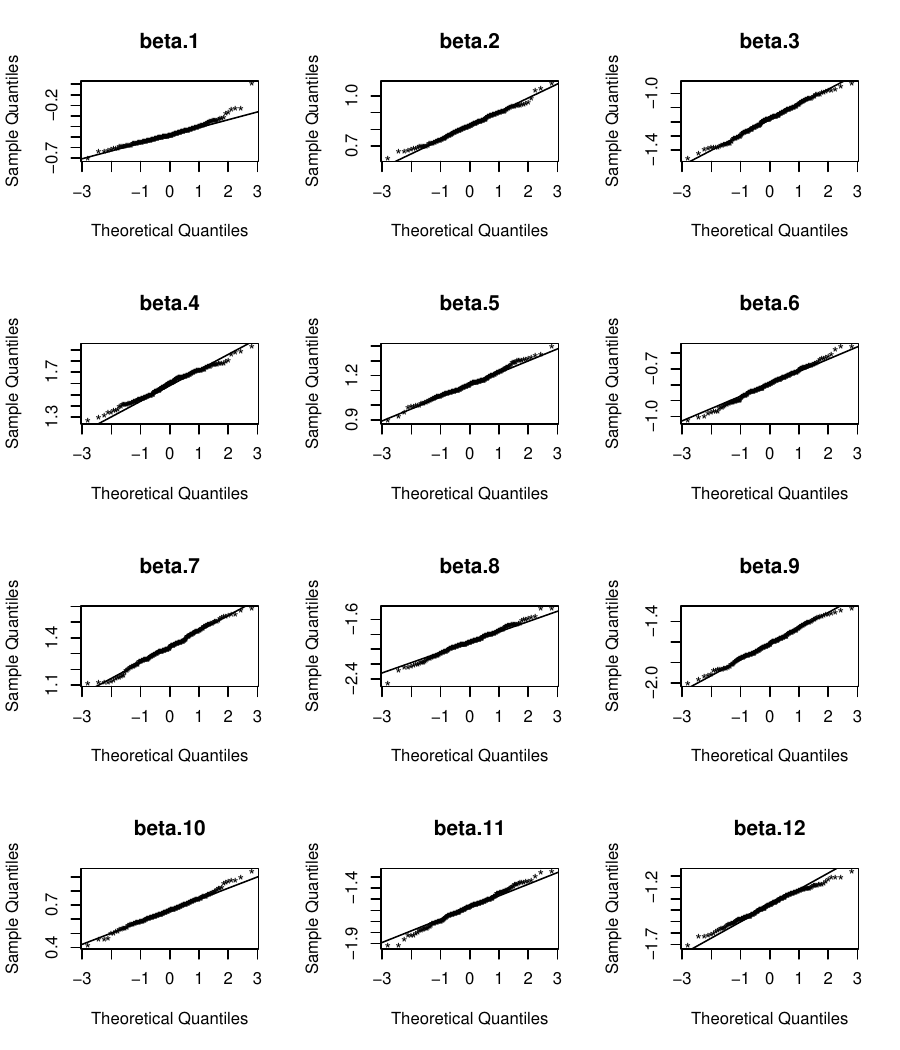}}
             \hfill
\subfloat[MCP]
         {\includegraphics[width=0.5\textwidth]{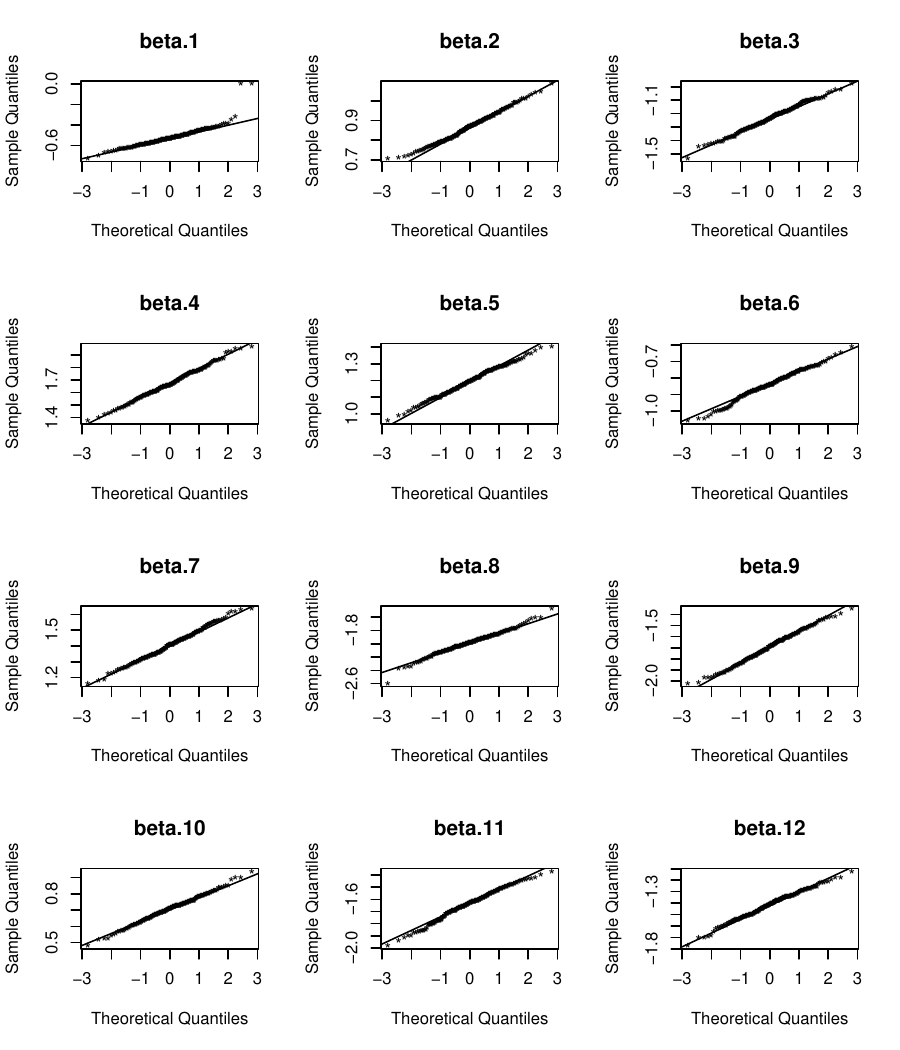}}
\subfloat[SCAD]
        {\includegraphics[width=0.5\textwidth]{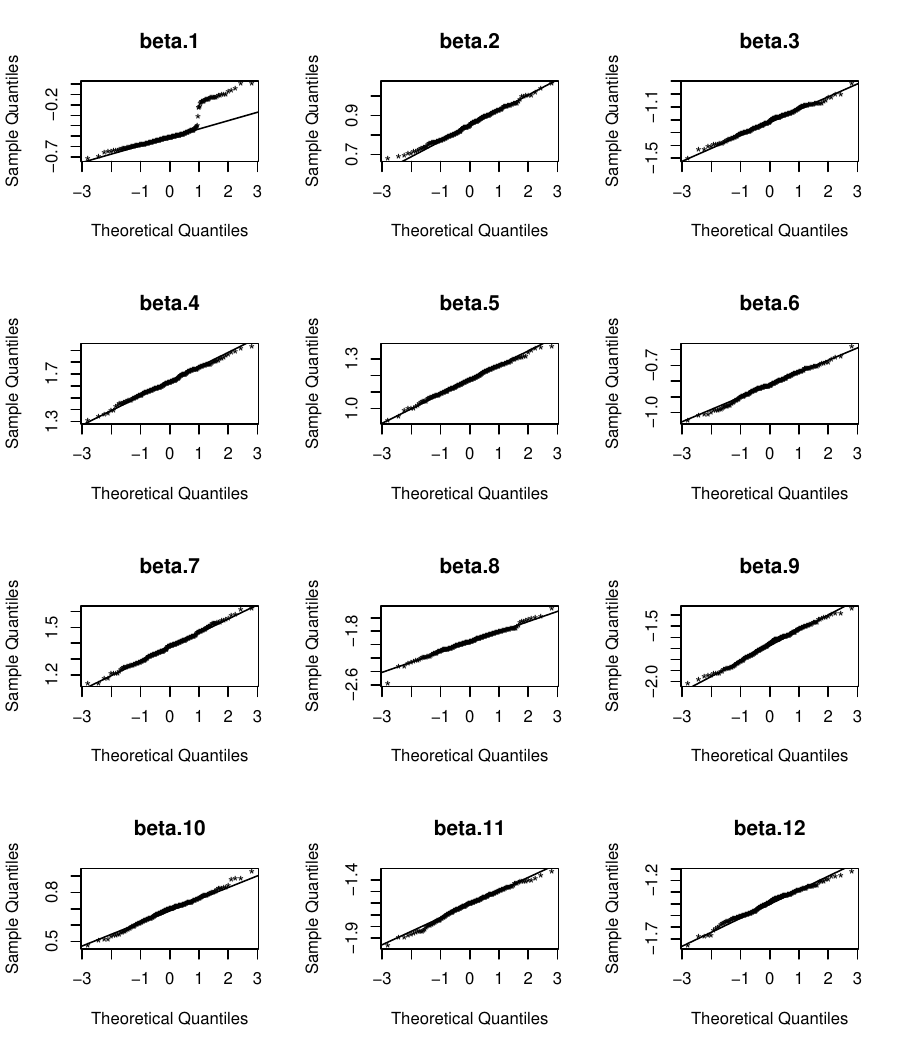}}
\caption[QQplot]
        {Normal Q-Q plots of the estimates of the twelve nonzero coefficients in the scenario with $p = 3000$, $n = 1000$, and $\rho = 0.8$}
    \end{figure*}

\begin{figure*}\centering
\subfloat[Lasso]
        {\includegraphics[width=0.5\textwidth]{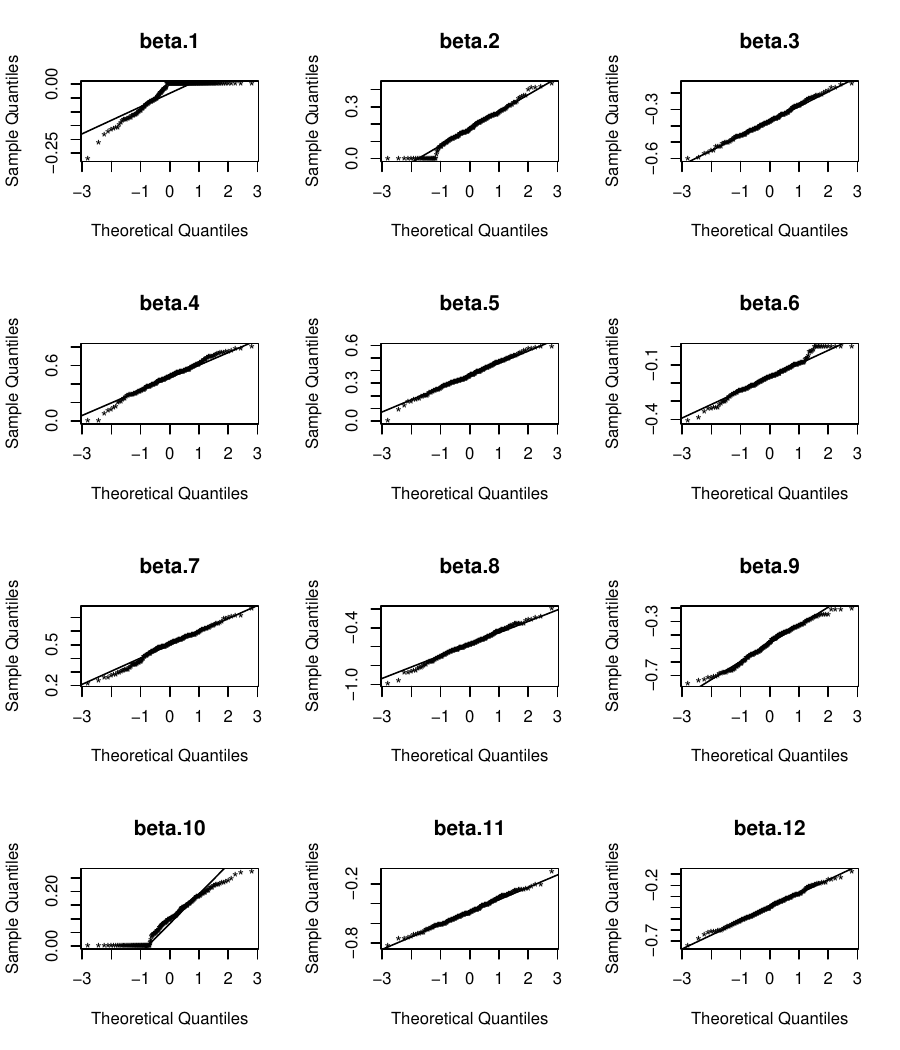}}
\subfloat[Adaptive Lasso \label{fig:mean and std of net24}]
         {\includegraphics[width=0.5\textwidth]{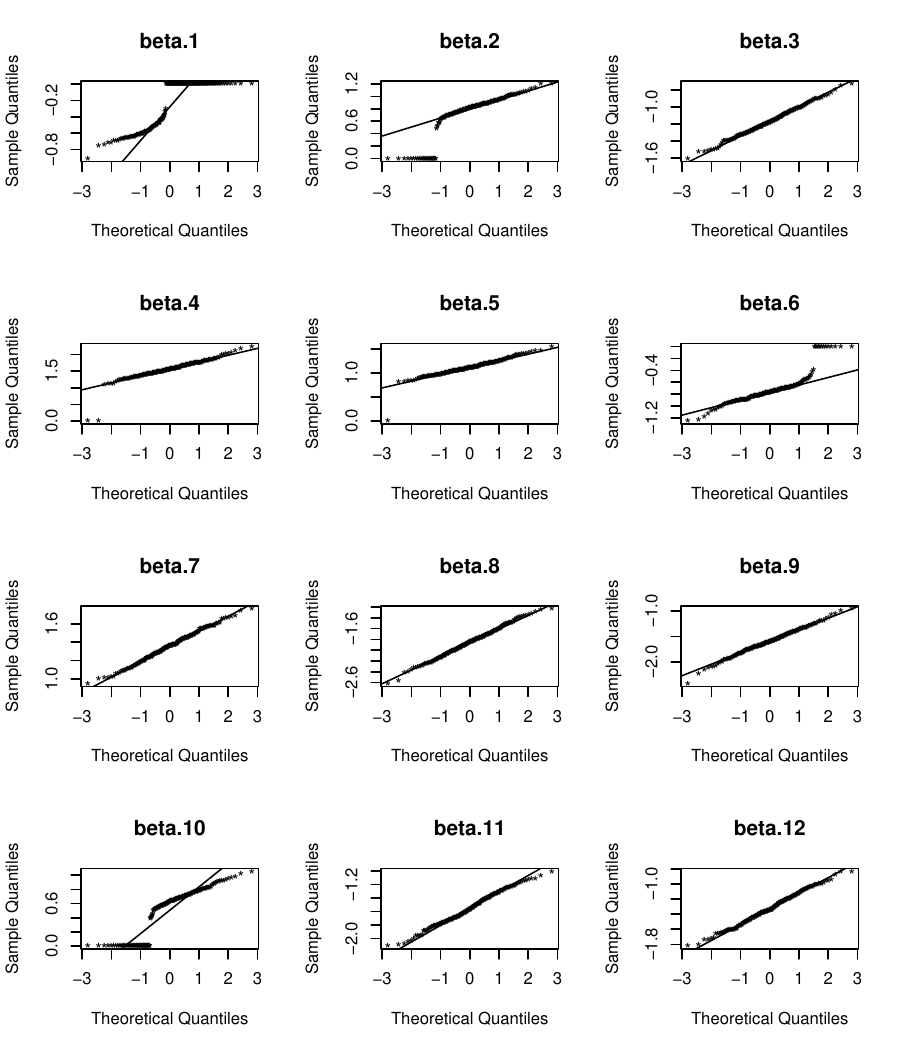}}
             \hfill
\subfloat[MCP]
         {\includegraphics[width=0.5\textwidth]{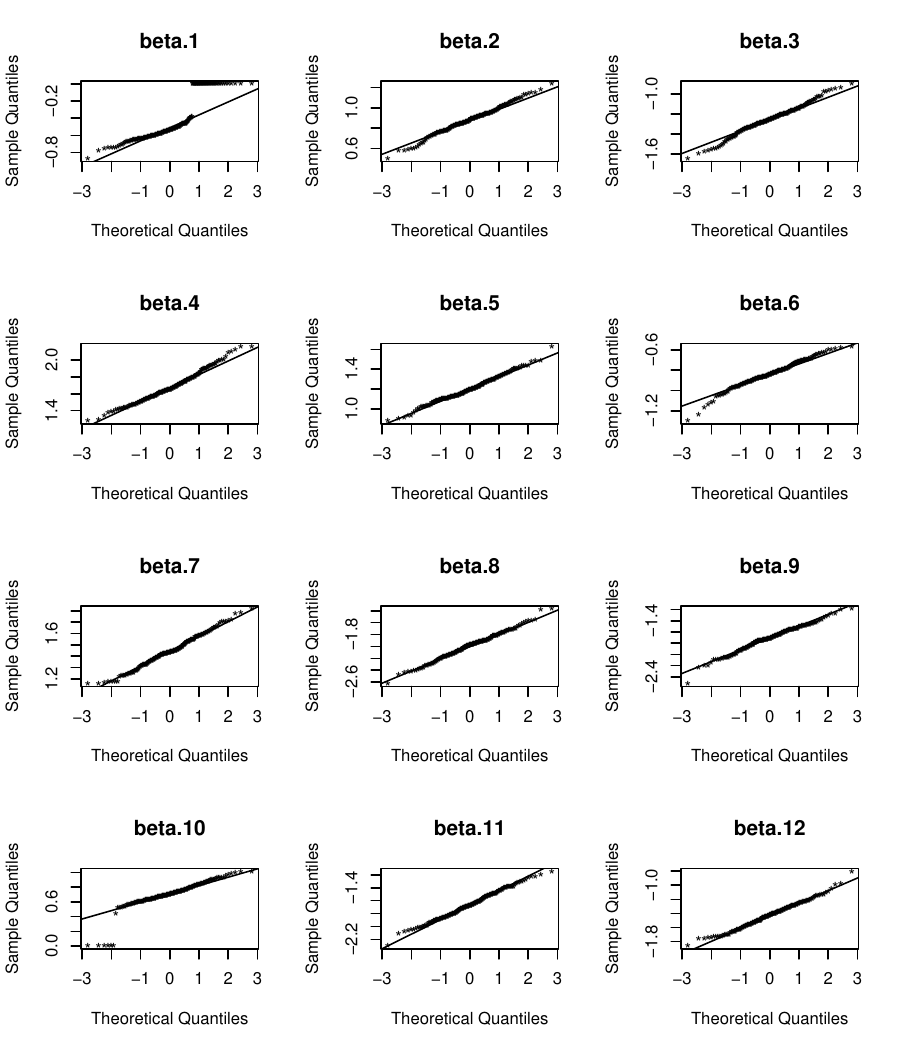}}
\subfloat[SCAD]
        {\includegraphics[width=0.5\textwidth]{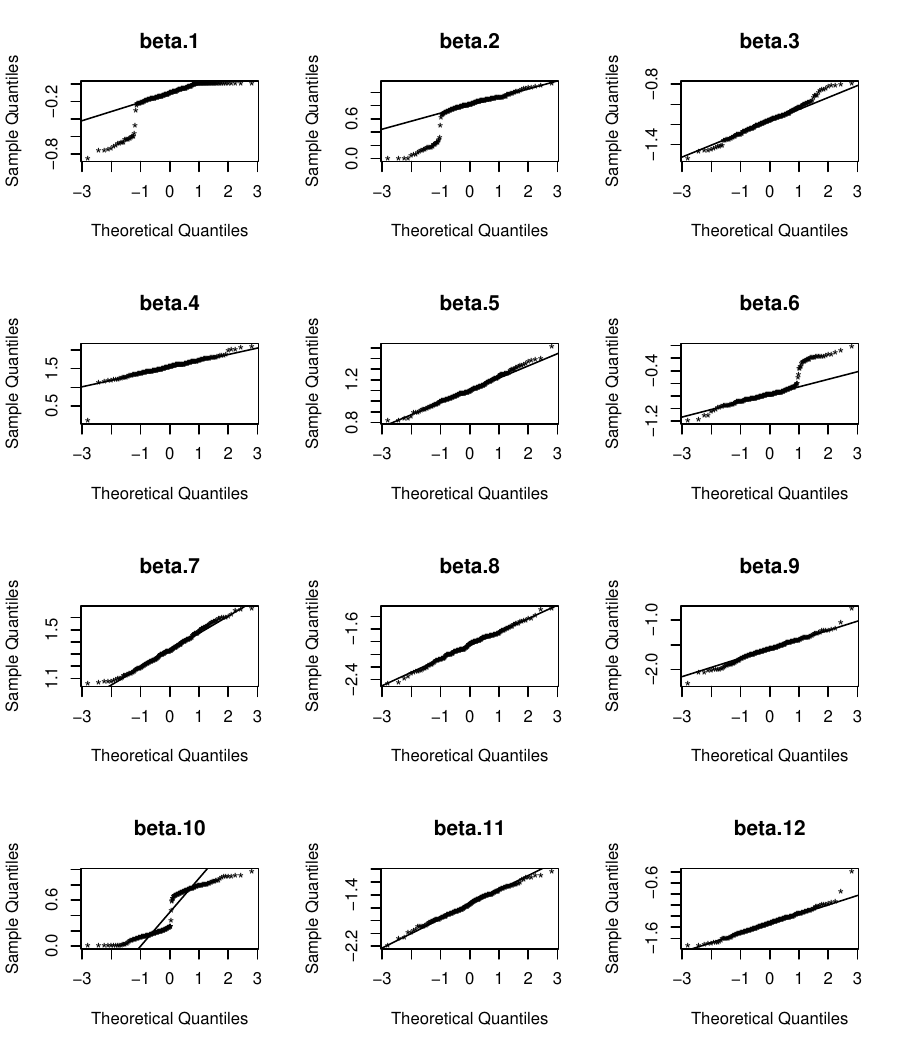}}
\caption[QQplot]
        {Normal Q-Q plots of the estimates of the twelve nonzero coefficients in the scenario with $p = 10000$, $n = 500$, and $\rho = 0$}
    \end{figure*}

\begin{figure*}\centering
\subfloat[Lasso]
        {\includegraphics[width=0.5\textwidth]{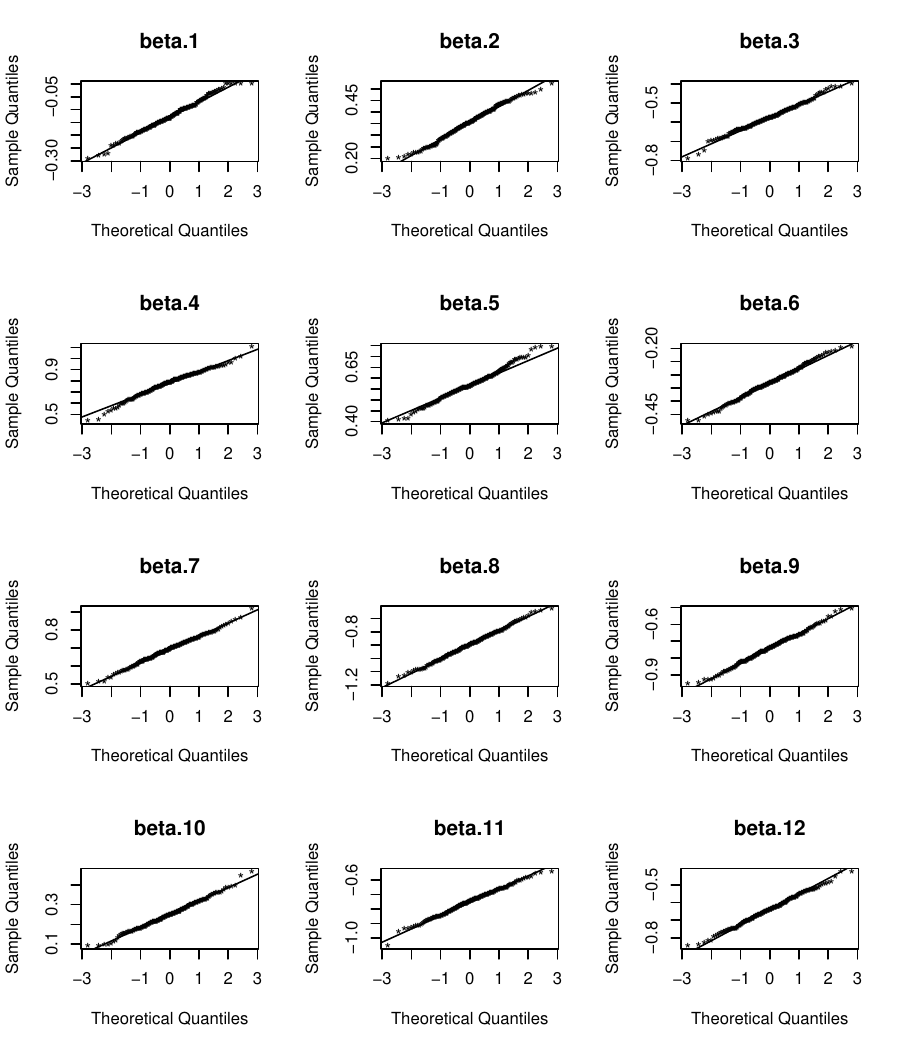}}
\subfloat[Adaptive Lasso \label{fig:mean and std of net24}]
         {\includegraphics[width=0.5\textwidth]{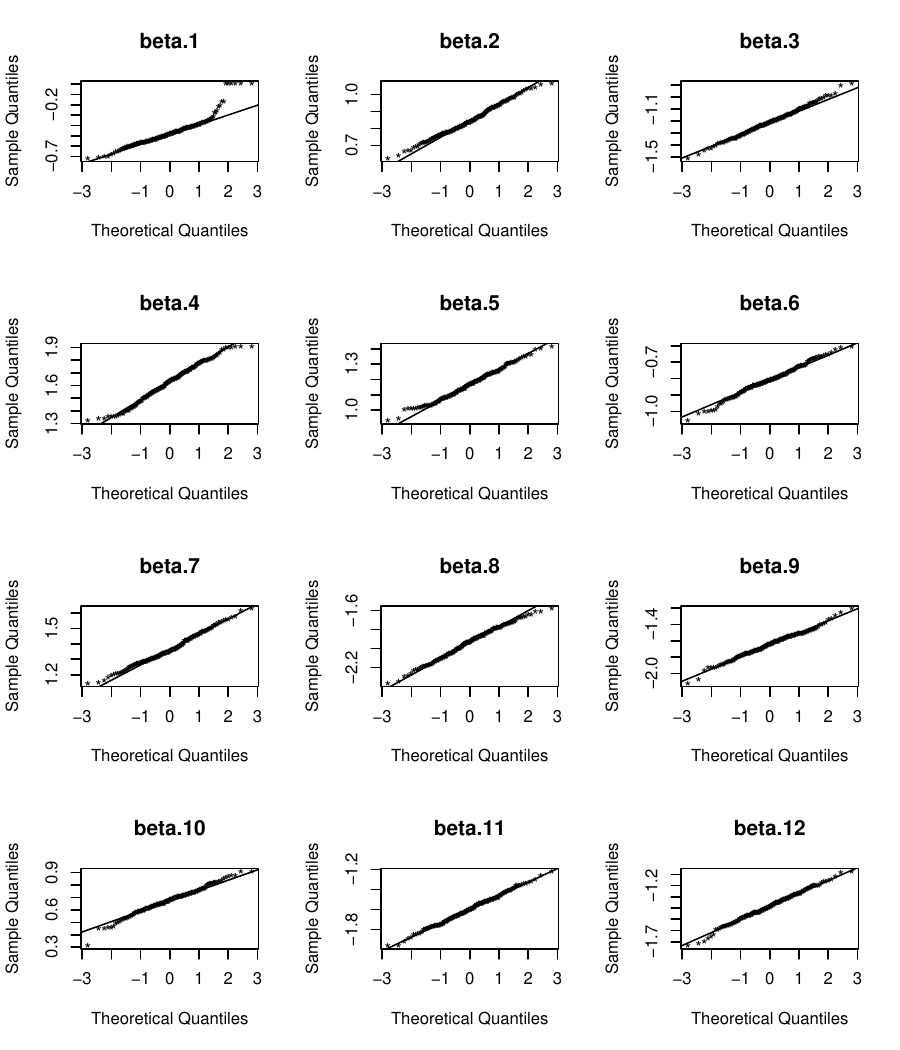}}
             \hfill
\subfloat[MCP]
         {\includegraphics[width=0.5\textwidth]{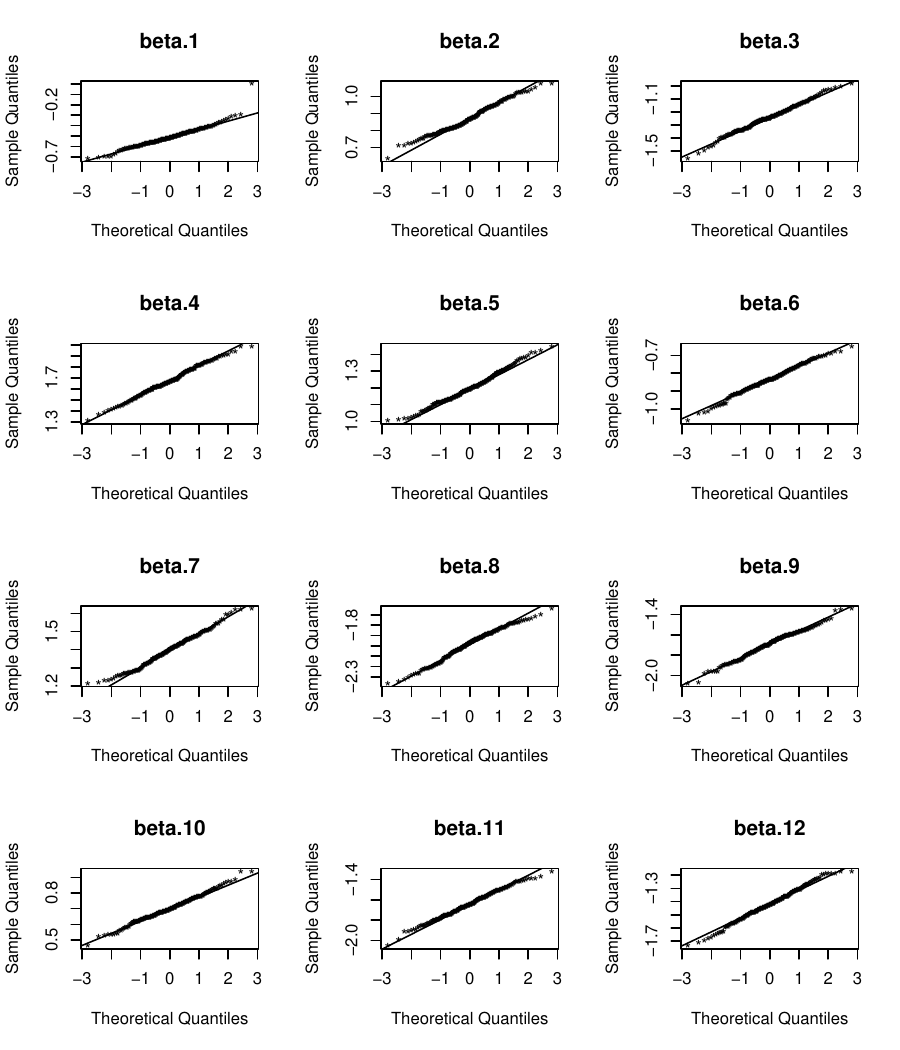}}
\subfloat[SCAD]
        {\includegraphics[width=0.5\textwidth]{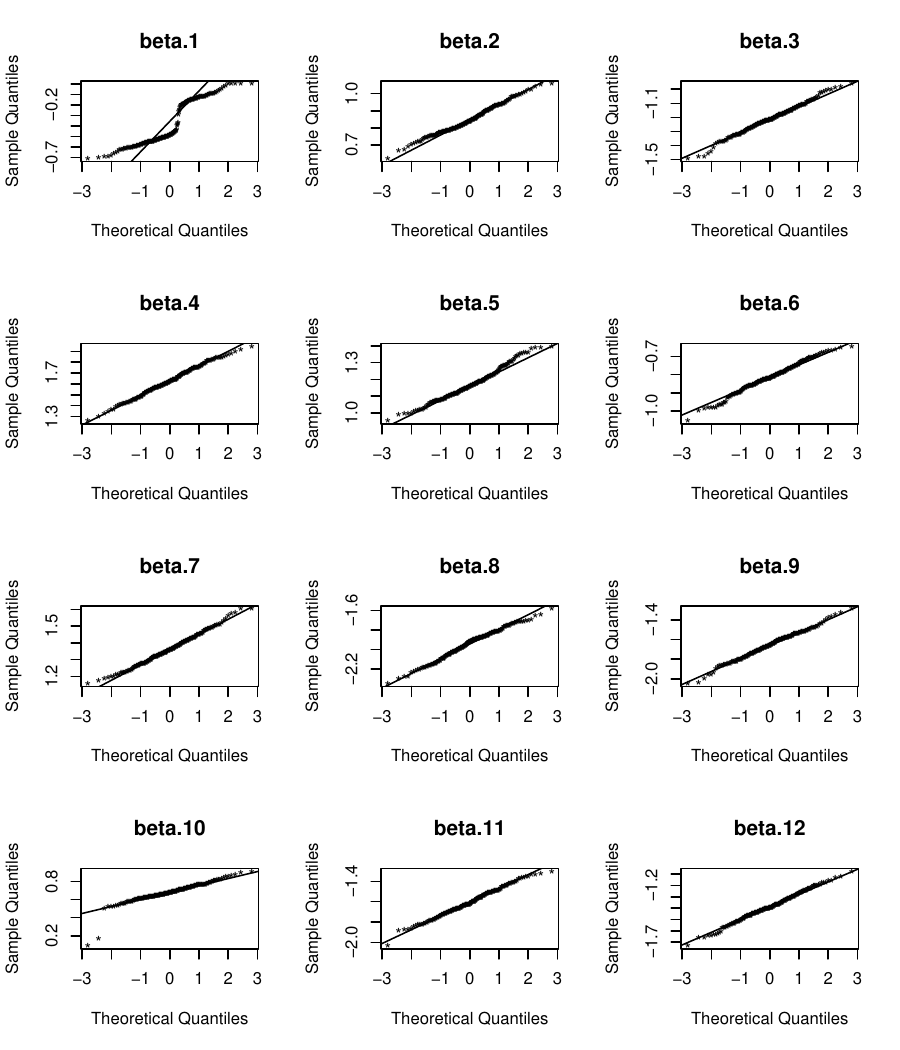}}
\caption[QQplot]
        {Normal Q-Q plots of the estimates of the twelve nonzero coefficients in the scenario with $p = 10000$, $n = 1000$, and $\rho = 0$}
    \end{figure*}

\begin{figure*}\centering
\subfloat[Lasso]
        {\includegraphics[width=0.5\textwidth]{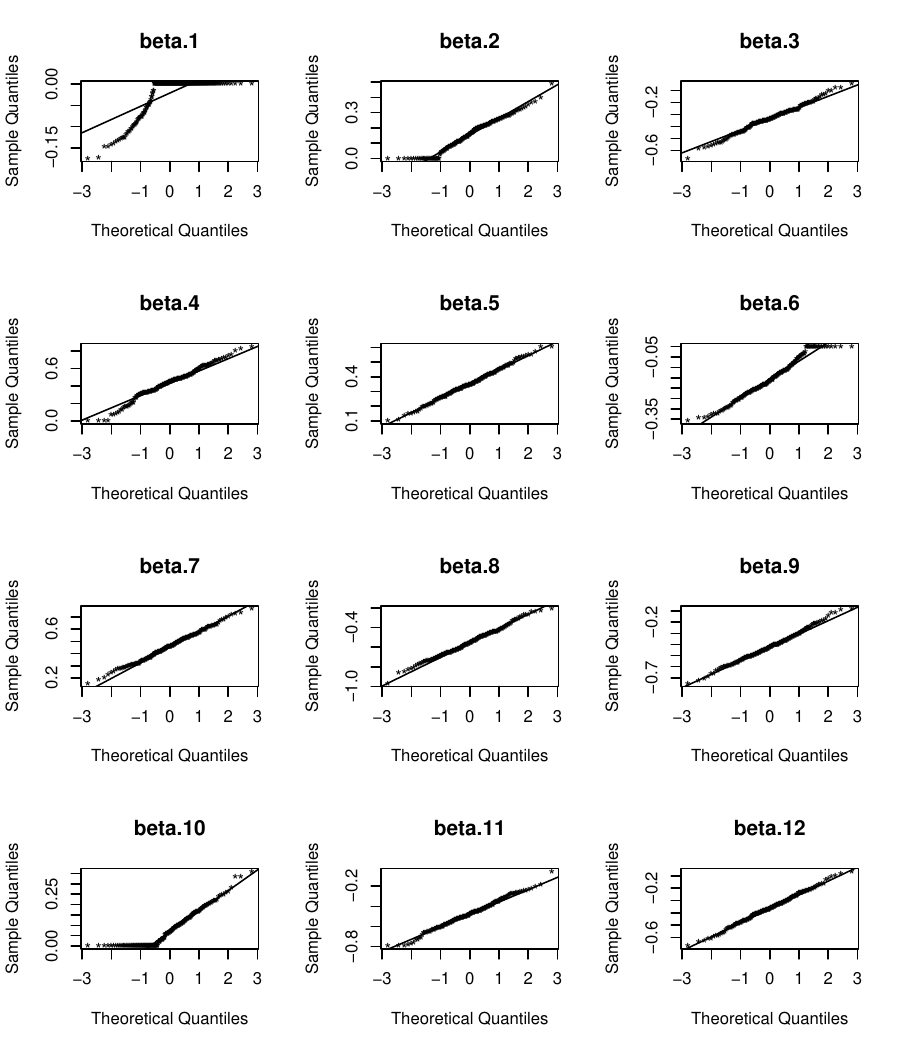}}
\subfloat[Adaptive Lasso \label{fig:mean and std of net24}]
         {\includegraphics[width=0.5\textwidth]{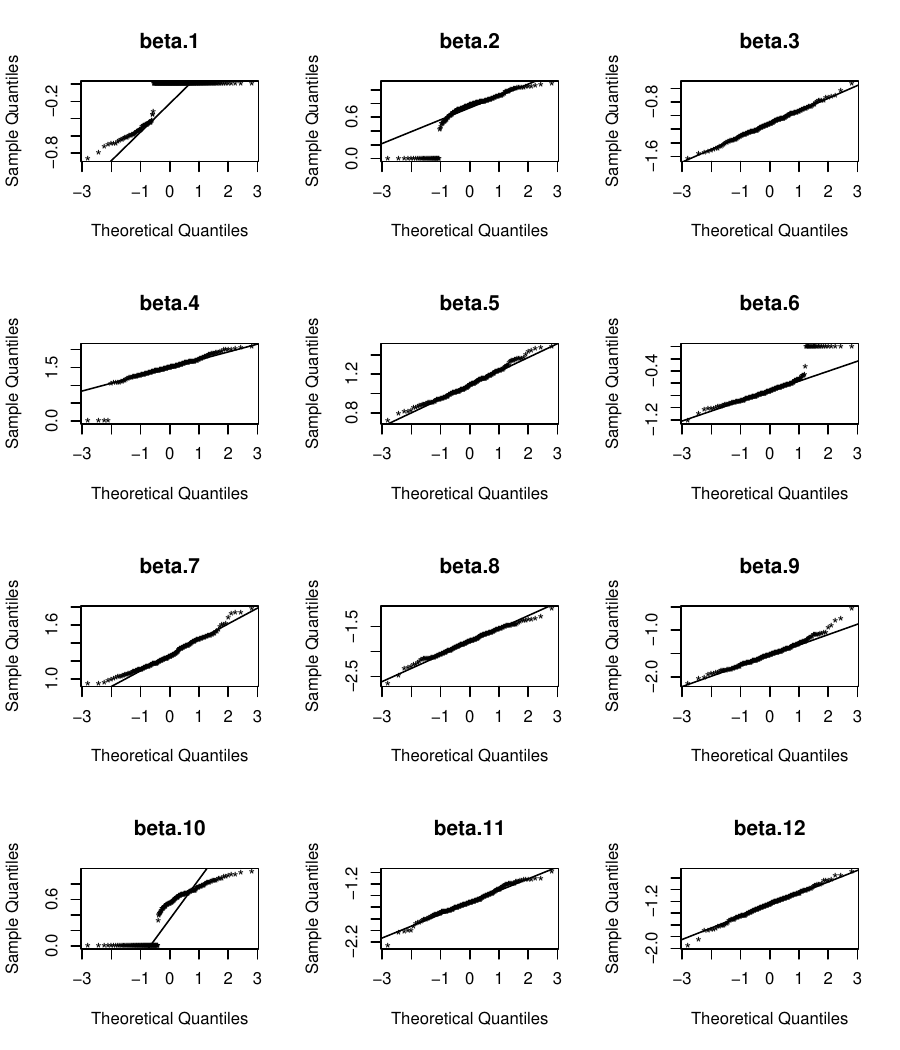}}
             \hfill
\subfloat[MCP]
         {\includegraphics[width=0.5\textwidth]{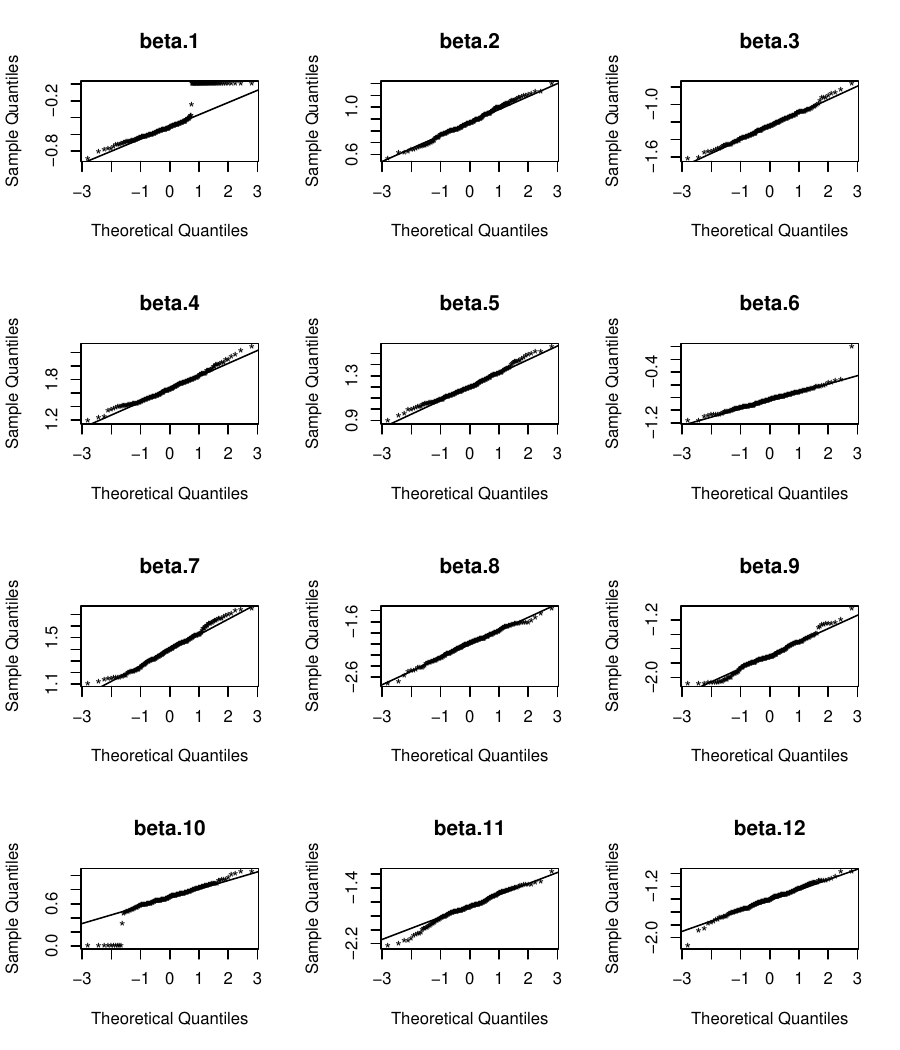}}
\subfloat[SCAD]
        {\includegraphics[width=0.5\textwidth]{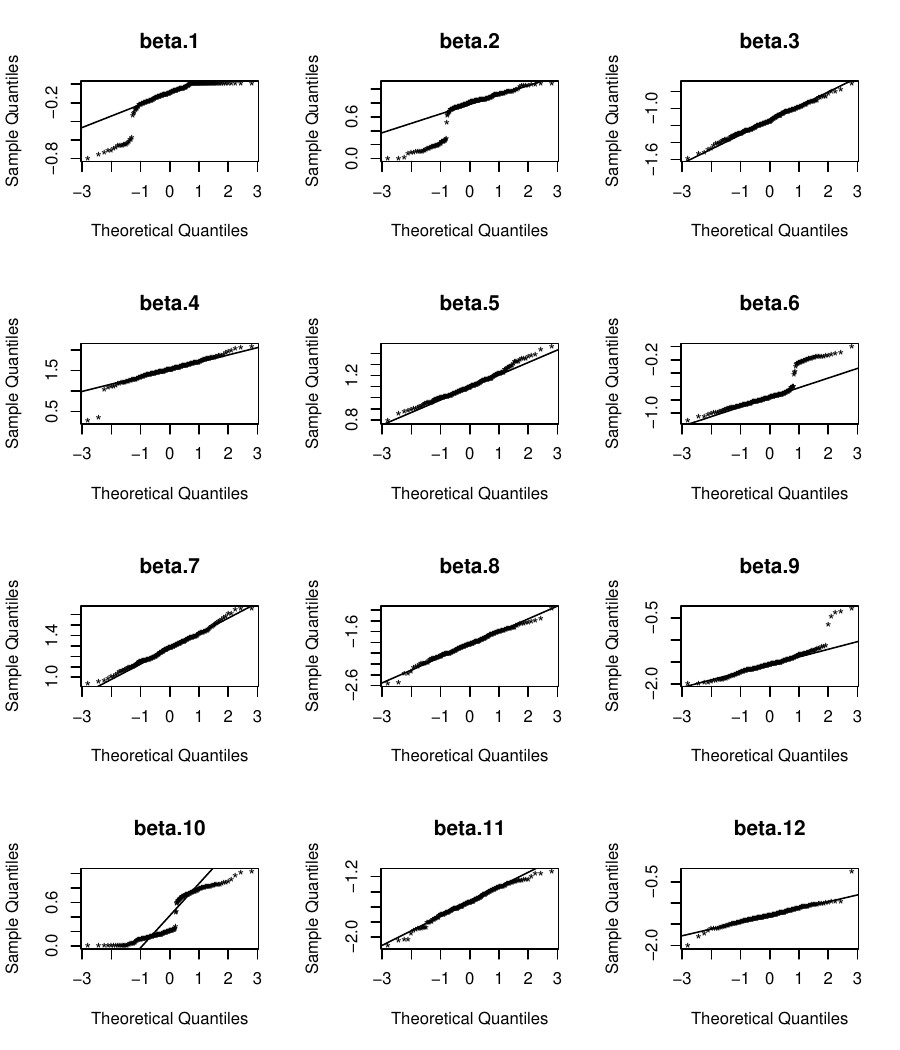}}
\caption[QQplot]
        {Normal Q-Q plots of the estimates of the twelve nonzero coefficients in the scenario with $p = 10000$, $n = 500$, and $\rho = 0.8$}
    \end{figure*}

\begin{figure*}\centering
\subfloat[Lasso]
        {\includegraphics[width=0.5\textwidth]{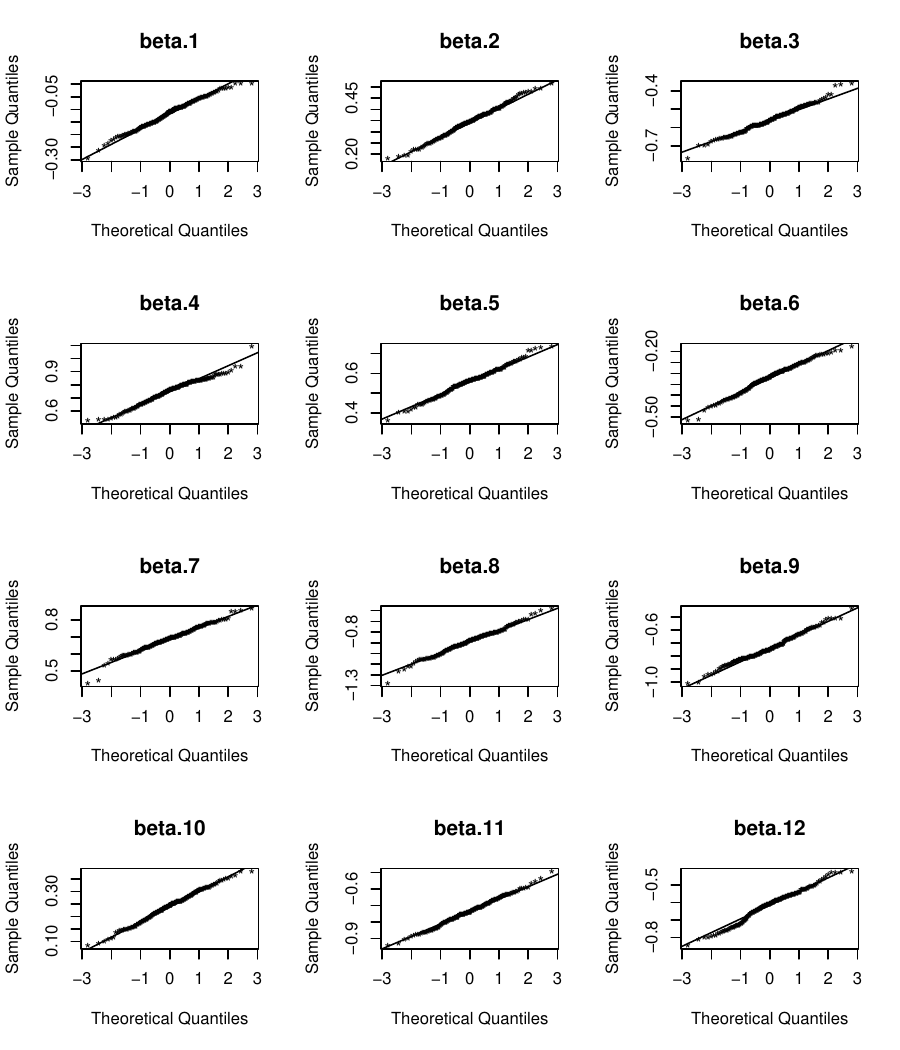}}
\subfloat[Adaptive Lasso \label{fig:mean and std of net24}]
         {\includegraphics[width=0.5\textwidth]{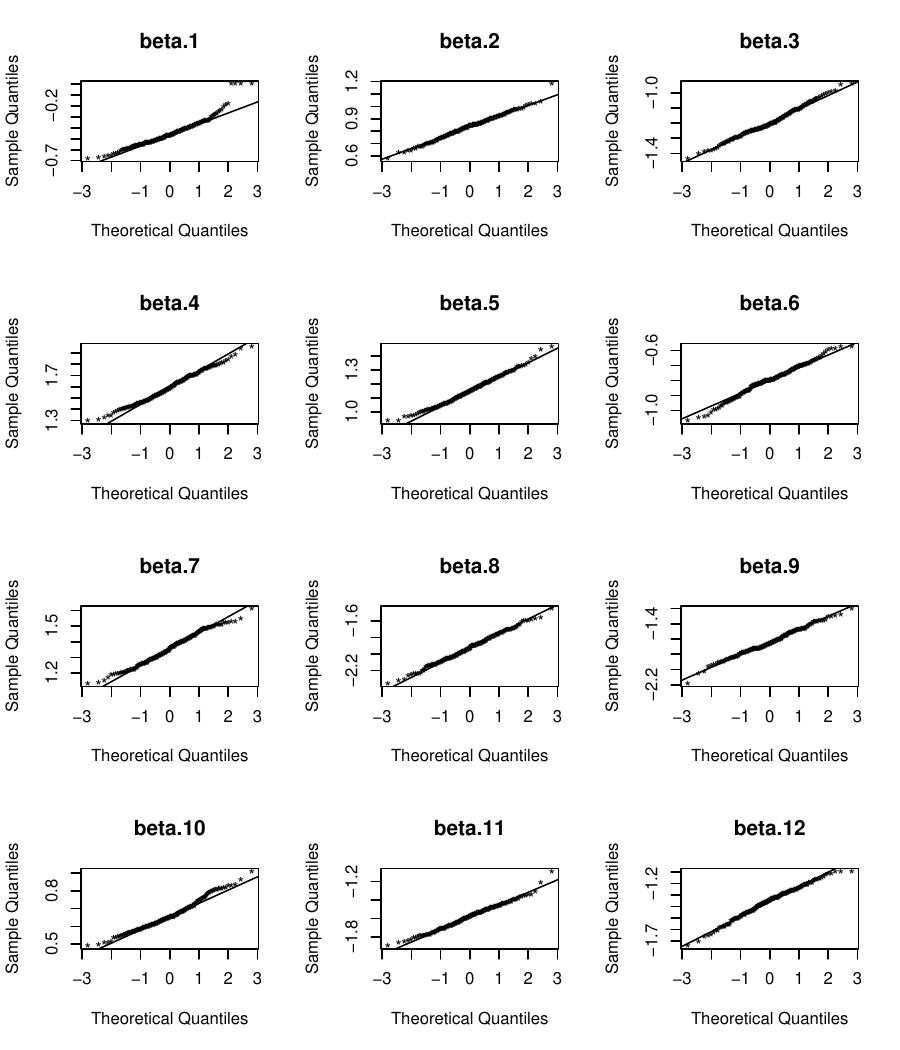}}
             \hfill
\subfloat[MCP]
         {\includegraphics[width=0.5\textwidth]{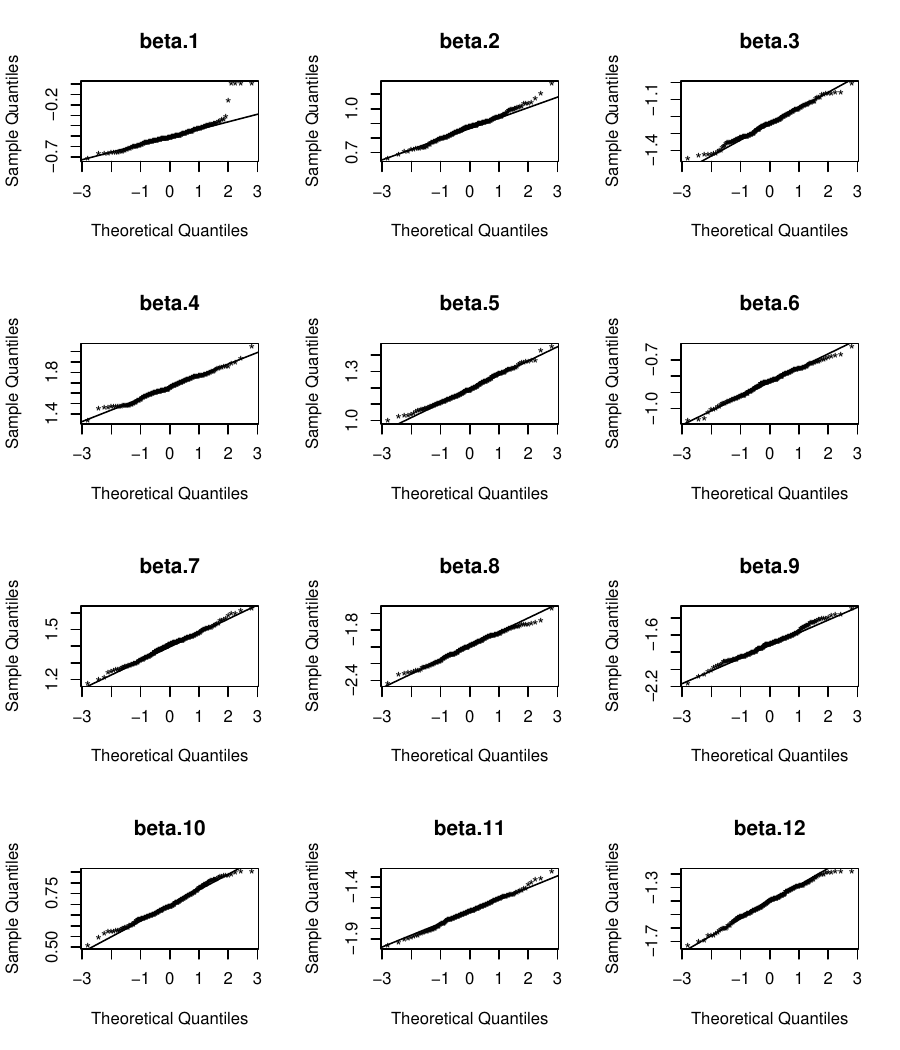}}
\subfloat[SCAD]
        {\includegraphics[width=0.5\textwidth]{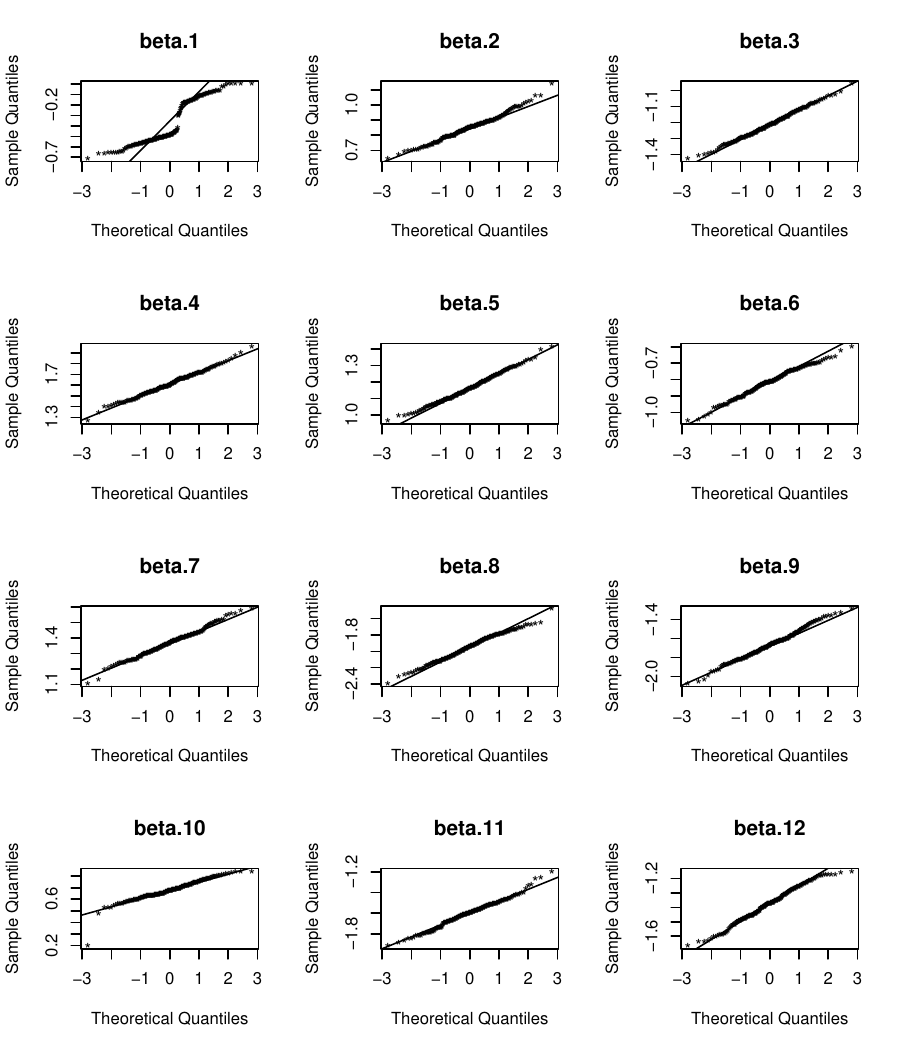}}
\caption[QQplot]
        {Normal Q-Q plots of the estimates of the twelve nonzero coefficients in the scenario with $p = 10000$, $n = 1000$, and $\rho = 0.8$}
    \end{figure*}

\bibliographystyle{chicago}      
\bibliography{ref}   